\def\llncs{0}
\def\fullpage{1}
\def\anonymous{0}
\def\authnote{1}
\def\notxfont{0}
\def\submission{0}

\ifnum\submission=1
\def\llncs{1}
\fi

\ifnum\llncs=1 	\documentclass[envcountsect,a4paper,runningheads,10pt]{llncs}
\else
	\documentclass[letterpaper,hmargin=1.0in,vmargin=1.0in,11pt]{article}
			\ifnum\fullpage=1
		\usepackage{fullpage}
		\fi
\fi

\usepackage[%
  colorlinks=true,
  citecolor=blue,
  pagebackref=true
]{hyperref}

\usepackage{amsmath, amsfonts, amssymb, mathtools,amscd}

\usepackage{amsthm}

\usepackage{lmodern}
\usepackage[T1]{fontenc}
\usepackage[utf8]{inputenc}

\usepackage{arydshln} 
\usepackage{url}
\usepackage{ifthen}
\usepackage{bm}
\usepackage{multirow}
\usepackage[dvips]{graphicx}
\usepackage[usenames]{color}
\usepackage{xcolor,colortbl} 
\usepackage{threeparttable}
\usepackage{comment}
\usepackage{paralist,verbatim}
\usepackage{cases}
\usepackage{booktabs}
\usepackage{braket}
\usepackage{cancel} 
\usepackage{ascmac} 
\usepackage{framed}
\usepackage{authblk}
\usepackage{pifont}
\usepackage{qcircuit}
\usepackage{tikz}
\definecolor{darkblue}{rgb}{0,0,0.6}
\definecolor{darkgreen}{rgb}{0,0.5,0}
\definecolor{maroon}{rgb}{0.5,0.1,0.1}
\definecolor{dpurple}{rgb}{0.2,0,0.65}

\usepackage[capitalise,noabbrev]{cleveref}
\usepackage[absolute]{textpos}
\usepackage[final]{microtype}
\usepackage[absolute]{textpos}
\usepackage{everypage}
\DeclareMathAlphabet{\mathpzc}{OT1}{pzc}{m}{it}

\usepackage{algorithmic}
\usepackage{algorithm}
\usepackage{here}

\newtheoremstyle{thicktheorem}%
{\topsep}
{\topsep}
{\itshape}{}%
{\bfseries}%
{.}
{ }%
{\thmname{#1}\thmnumber{ #2}%
		\thmnote{ (#3)}%
}

\newtheoremstyle{remark}
{\topsep}
{\topsep}
	{}
	{}
	{}
	{.}
	{ }
	{\textit{\thmname{#1}}\thmnumber{ #2}
			\thmnote{ (#3)}%
	}

\ifnum\llncs=0
	\theoremstyle{thicktheorem}
	\newtheorem{theorem}{Theorem}[section]
	\newtheorem{lemma}[theorem]{Lemma}
	\newtheorem{corollary}[theorem]{Corollary}
	
	\newtheorem{definition}[theorem]{Definition}

	\theoremstyle{remark}
	
	\newtheorem{remark}[theorem]{Remark}

\else
\fi

\Crefname{MyClaim}{Claim}{Claims}

	\crefname{theorem}{Theorem}{Theorems}
	\crefname{assumption}{Assumption}{Assumptions}
	\crefname{construction}{Construction}{Constructions}
	\crefname{corollary}{Corollary}{Corollaries}
	\crefname{conjecture}{Conjecture}{Conjectures}
	\crefname{definition}{Definition}{Definitions}
	\crefname{exmaple}{Example}{Examples}
	\crefname{experiment}{Experiment}{Experiments}
	\crefname{counterexample}{Counterexample}{Counterexamples}
	\crefname{lemma}{Lemma}{Lemmata}
	\crefname{observation}{Observation}{Observations}
	\crefname{proposition}{Proposition}{Propositions}
	\crefname{remark}{Remark}{Remarks}
	\crefname{claim}{Claim}{Claims}
	\crefname{fact}{Fact}{Facts}
	\crefname{note}{Note}{Notes}

\ifnum\llncs=1
 \crefname{appendix}{App.}{Appendices}
 \crefname{section}{Sec.}{Sections}
\else
\fi

\ifnum\llncs=1
\renewcommand*{\backref}[1]{}
\else
	\renewcommand*{\backref}[1]{(Cited on page~#1.)}
	\ifnum\notxfont=1
	\else
		\usepackage{newtxtext}
	\fi
\fi
\ifnum\authnote=0  
\newcommand{\mor}[1]{}
\newcommand{\minki}[1]{}
\newcommand{\takashi}[1]{}

\else
\newcommand{\mor}[1]{$\ll$\textsf{\color{red} Tomoyuki: { #1}}$\gg$}
\newcommand{\takashi}[1]{$\ll$\textsf{\color{orange} Takashi: { #1}}$\gg$}
\newcommand{\minki}[1]{$\ll$\textsf{\color{darkgreen} Minki: { #1}}$\gg$}
\newcommand{\shira}[1]{$\ll$\textsf{\color{violet} Yuki: { #1}}$\gg$}
\fi

\newcommand{\Tr}{\mathrm{Tr}}
\newcommand{\com}{\mathsf{com}}
\newcommand{\SD}{\mathsf{SD}} 

\newcommand{\Good}{\mathsf{Good}}

\newcommand{\BQP}{\mathbf{BQP}}

\newcommand{\puzz}{\mathsf{puzz}}
\newcommand{\ans}{\mathsf{ans}}

\newcommand{\Samp}{\algo{Samp}}
\newcommand{\Supp}{\mathrm{Supp}}

\newcommand{\seteq}{\coloneqq}

\newcommand{\cA}{\mathcal{A}}
\newcommand{\cB}{\mathcal{B}}
\newcommand{\cC}{\mathcal{C}}
\newcommand{\cD}{\mathcal{D}}
\newcommand{\cE}{\mathcal{E}}
\newcommand{\cF}{\mathcal{F}}

\newcommand{\cO}{\mathcal{O}}
\newcommand{\cP}{\mathcal{P}}
\newcommand{\cQ}{\mathcal{Q}}
\newcommand{\cR}{\mathcal{R}}
\newcommand{\cS}{\mathcal{S}}

\newcommand{\cU}{\mathcal{U}}
\newcommand{\cV}{\mathcal{V}}

\def\makeuppercase#1{
\expandafter\newcommand\csname tl#1\endcsname{\widetilde{#1}}
}

\def\makelowercase#1{
\expandafter\newcommand\csname tl#1\endcsname{\widetilde{#1}}
}

\newcommand{\N}{\mathbb{N}}

\newcommand{\secp}{\lambda}

\newcommand{\B}{\entity{B}}

\newcommand{\view}{\mathsf{view}}

\newcommand*{\algo}[1]{\ensuremath{\mathsf{#1}}}

\newcommand*{\entity}[1]{\mathcal{#1}}

\newenvironment{boxfig}[2]{\begin{figure}[#1]\fbox{\begin{minipage}{0.97\linewidth}
                        \vspace{0.2em}
                        \makebox[0.025\linewidth]{}
                        \begin{minipage}{0.95\linewidth}
            {{
                        #2 }}
                        \end{minipage}
                        \vspace{0.2em}
                        \end{minipage}}}{\end{figure}}

\newcommand{\bit}{\{0,1\}}

\newcommand{\Ver}{\algo{Ver}}

\newcommand{\st}{\algo{st}}

\newcommand{\TD}{\algo{TD}}

\newcommand{\negl}{{\mathsf{negl}}}

\newcommand{\poly}{{\mathrm{poly}}}

\makeatletter
\DeclareRobustCommand
  \myvdots{\vbox{\baselineskip4\p@ \lineskiplimit\z@
    \hbox{.}\hbox{.}\hbox{.}}}
\makeatother

\title{Cryptographic Characterization of Quantum Advantage}

\ifnum\anonymous=1
\ifnum\llncs=1
\author{\empty}\institute{\empty}
\else
\author{}
\fi
\else
\ifnum\llncs=1
\author{
Tomoyuki Morimae\inst{1} \and Yuki Shirakawa\inst{1} \and Takashi Yamakawa\inst{2,3,1}
}
\institute{
 Yukawa Institute for Theoretical Physics, Kyoto University, Kyoto, Japan \and NTT Social Informatics Laboratories, Tokyo, Japan \and NTT Research Center for Theoretical Quantum Information, Atsugi, Japan
}
\else
\author[1]{Tomoyuki Morimae}
\author[1]{ Yuki Shirakawa}
\author[2,3,1]{ Takashi Yamakawa}
\affil[1]{{\small Yukawa Institute for Theoretical Physics, Kyoto University, Kyoto, Japan}\authorcr{\small tomoyuki.morimae@yukawa.kyoto-u.ac.jp} \authorcr{\small yuki.shirakawa@yukawa.kyoto-u.ac.jp}}
\affil[2]{{\small NTT Social Informatics Laboratories, Tokyo, Japan}\authorcr{\small takashi.yamakawa@ntt.com}}
\affil[3]{{\small NTT Research Center for Theoretical Quantum Information, Atsugi, Japan}}
\fi 
\fi

\date{}

\begin{document}

\maketitle

\begin{abstract}
Quantum computational advantage refers to an existence of computational tasks that are easy for quantum computing but
hard for classical one.
Unconditionally showing quantum advantage is beyond our current understanding of
complexity theory, and therefore some computational assumptions are needed.
Which complexity assumption is necessary and sufficient for quantum advantage?
In this paper, we show that inefficient-verifier proofs of quantumness (IV-PoQ) exist if and only if
classically-secure one-way puzzles (OWPuzzs) exist.
As far as we know, this is the first time that a complete cryptographic characterization of
quantum advantage is obtained.
IV-PoQ are a generalization of proofs of quantumness (PoQ) where 
the verifier is efficient during the interaction but may use unbounded time afterward.
IV-PoQ capture various types of quantum advantage previously studied, such as sampling-based quantum advantage and searching-based one.
Previous work [Morimae and Yamakawa, Crypto 2024] showed that IV-PoQ can be constructed from OWFs, but 
a construction of IV-PoQ from weaker assumptions was left open. Our result solves the open problem,
because OWPuzzs are believed to be weaker than OWFs.
OWPuzzs are one of the most fundamental quantum cryptographic primitives 
implied by many quantum cryptographic primitives weaker than one-way functions (OWFs), such as 
pseudorandom unitaries (PRUs), pseudorandom state generators (PRSGs), and one-way state generators (OWSGs). 
The equivalence between IV-PoQ and classically-secure OWPuzzs 
therefore highlights that if
there is no quantum advantage, then these fundamental cryptographic primitives do not exist.
The equivalence also means that quantum advantage is an example of the applications of OWPuzzs.
Except for commitments, no application of OWPuzzs was known before. Our result shows that quantum advantage is another
application of OWPuzzs, which solves the open question of [Chung, Goldin, and Gray, Crypto 2024]. 
Moreover, it is the first quantum-computation-classical-communication (QCCC) application of OWPuzzs.
To show the main result, we introduce 
several new concepts and show some results that will be of independent interest. 
In particular, we introduce an interactive (and average-case) version of 
sampling problems where the task is to sample the transcript obtained by a classical interaction between two quantum polynomial-time algorithms.
We show that quantum advantage in interactive sampling problems is equivalent to the existence of IV-PoQ,
which
is considered as an interactive (and average-case) version of Aaronson's result [Aaronson, TCS 2014],
$\mathbf{SampBQP}\neq\mathbf{SampBPP}\Leftrightarrow \mathbf{FBQP}\neq\mathbf{FBPP}$.
Finally, we also introduce zero-knowledge IV-PoQ and study sufficient and necessary conditions for their existence.
\end{abstract}

\thispagestyle{empty}
\newpage

\setcounter{tocdepth}{2}
\tableofcontents
\thispagestyle{empty}
\newpage

\setcounter{page}{1}
\section{Introduction}
\label{sec:introduction} 
Quantum computational advantage refers to the existence of computational tasks that are easy for quantum computing
but hard for classical one. Unconditionally showing quantum advantage is extremely hard, and is beyond our current understanding of
complexity theory.\footnote{There are several interesting results (such as \cite{BGK18}) that show unconditional quantum
advantage by restricting classical computing. In this paper, we consider any polynomial-time classical computing.}
Some computational assumptions are therefore required.
Which complexity assumption is necessary and sufficient for quantum advantage?
As far as we know, no complete characterization of quantum advantage has been achieved before.

In this paper, we identify a cryptographic assumption that is
necessary and sufficient for quantum advantage.
Our main result is the following one:\footnote{
In this paper, all classically-secure OWPuzzs are 
ones with $(1-\negl(\secp))$-correctness and $(1-1/\poly(\secp))$-security.
Unlike quantumly-secure OWPuzzs, we do not know how to amplify the gap for classically-secure OWPuzzs.}\footnote{\label{footnote_nonuniform}In this paper, we consider the uniform adversarial model (i.e., adversaries are modeled as Turing machines), and some steps of the proofs of \Cref{thm:main} crucially rely on the uniformity of the adversary. (See \Cref{sec:uniform} for the detail.) We leave it open to prove (or disprove) the non-uniform variant of \Cref{thm:main}.}

\begin{theorem}
\label{thm:main}
Inefficient-verifier proofs of quantumness
(IV-PoQ) exist if and only if classically-secure one-way puzzles (OWPuzzs) exist.    
\end{theorem}
As far as we know, this is the first time that a complete cryptographic characterization of quantum advantage is obtained.

\paragraph{What are IV-PoQ?}
IV-PoQ are a generalization of proofs of quantumness (PoQ)~\cite{JACM:BCMVV21}. 
A PoQ is an interactive protocol between a prover and a classical probabilistic polynomial-time (PPT)
verifier over a classical channel. 
There exists a quantum polynomial-time (QPT) prover such that the verifier accepts with high probability ({\it completeness}), but
for any PPT prover the verifier rejects with high probability ({\it soundness}).
PoQ can be constructed from several cryptographic assumptions, such as
(noisy) trapdoor claw-free functions with the adaptive-hardcore-bit property~\cite{JACM:BCMVV21},
trapdoor 2-to-1 collision-resistant hash functions~\cite{NatPhys:KMCVY22}, 
(full-domain) trapdoor permutations~\cite{ITCS:MorYam23},
quantum homomorphic encryptions~\cite{STOC:KLVY23},
or knowledge assumptions~\cite{knowledge_assumptions}.
Non-interactive PoQ are possible based on the hardness of factoring~\cite{FOCS:Shor94} or random oracles \cite{10.1145/3658665}.

IV-PoQ~\cite{C:MorYam24} are the same as PoQ except that the verifier's final computation to make the decision can be unbounded.
IV-PoQ are a generalization of PoQ, and as we will explain later, IV-PoQ capture various types of quantum advantage studied so far,
such as sampling-based quantum advantage and searching-based one.

Completely identifying a necessary and sufficient assumption for the existence of (IV-)PoQ remained open.
In particular, in \cite{C:MorYam24}, IV-PoQ were constructed from
classically-secure OWFs,
but the problem of constructing IV-PoQ from weaker assumptions was left open in that paper.
Our main result \cref{thm:main} solves the open problem, because as we will explain later,
OWPuzzs are believed to be weaker than OWFs~\cite{Kre21,C:MorYam22,STOC:KhuTom24}.
Moreover, a known necessary condition for the existence of IV-PoQ
was only the almost trivial one: ${\bf BPP}\neq{\bf PP}$\footnote{The output probability distribution of any QPT algorithm can be computed by 
a classical polynomial-time deterministic algorithm that queries the $\mathbf{PP}$ oracle~\cite{ForRog99}. Therefore, if $\mathbf{BPP}=\mathbf{PP}$, a PPT prover can cheat the verifier.}~\cite{C:MorYam24}.
Our main result \cref{thm:main} improves this to
a highly non-trivial necessary condition, namely, the existence of classically-secure OWPuzzs.\footnote{This is an improvement,
because if classically-secure OWPuzzs exist then $\mathbf{BPP}\neq\mathbf{PP}$.}

\paragraph{What are OWPuzzs?}
In classical cryptography, the existence of OWFs is the minimum assumption~\cite{FOCS:ImpLub89}, because
many primitives exist if and only if OWFs exist, such as pseudorandom generators (PRGs), pseudorandom functions (PRFs),
zero-knowledge, commitments, digital signatures, and secret-key encryptions (SKE), 
and almost all primitives imply OWFs. 
On the other hand, recent active studies have revealed that in quantum cryptography 
OWFs are not necessarily the minimum assumption.
Many fundamental primitives have been introduced, such as pseudorandom unitaries (PRUs)~\cite{C:JiLiuSon18}, pseudorandom function-like state generators (PRFSGs)~\cite{C:AnaQiaYue22}, unpredictable state generators (UPSGs)~\cite{cryptoeprint:2024/701},
pseudorandom state generators (PRSGs)~\cite{C:JiLiuSon18}, one-way state generators (OWSGs)~\cite{C:MorYam22},
EFI pairs~\cite{ITCS:BCQ23}, and one-way puzzles (OWPuzzs)~\cite{STOC:KhuTom24}. 
They could exist even if OWFs do not exist~\cite{Kre21,STOC:KQST23,STOC:LomMaWri24},
but still imply several useful applications such as message authentication codes~\cite{C:AnaQiaYue22},
commitments~\cite{C:MorYam22,C:AnaQiaYue22}, multi-party computations~\cite{C:MorYam22,C:AnaQiaYue22,C:BCKM21b,EC:GLSV21}, 
secret-key encryptions~\cite{C:AnaQiaYue22}, private-key quantum money~\cite{C:JiLiuSon18}, 
digital signatures~\cite{C:MorYam22}, etc.

In particular, one-way puzzles (OWPuzzs) are one of the most fundamental primitives in this ``cryptographic world below OWFs''.
A OWPuzz is a pair $(\Samp,\Ver)$ of two algorithms. $\Samp$ is a QPT algorithm that takes $1^\secp$ as input
and outputs two classical bit strings $\puzz$ and $\ans$.
$\Ver$ is an unbounded algorithm that takes $\puzz$ and $\ans'$ as input, and outputs $\top$ or $\bot$.
We require two properties, correctness and security.
Correctness requires that
$\Ver$ accepts $(\puzz,\ans)$ sampled by $\Samp$ with large probability.
Security requires that for any QPT algorithm $\cA$ that takes $\puzz$ as input and outputs $\ans'$, 
$\Ver$ accepts $(\puzz,\ans')$ with only small probability.
In particular, when the security is required only for all PPT adversaries, we say that a OWPuzz is classically-secure. 
(Note that the $\Samp$ algorithm of classically-secure OWPuzzs is QPT, not PPT, even when we consider classical security.)

As is shown in \cref{fig:result_summary}, OWPuzzs (and therefore classically-secure OWPuzzs) are implied by
many primitives such as
\begin{itemize}
    \item 
    PRUs, PRFSGs, UPSGs, PRSGs, (pure) OWSGs,
    \item
(pure) private-key quantum money,
secret-key encryption schemes,
digital signatures,
\item
many quantum-computation-classical-communication
(QCCC) primitives\footnote{Here, QCCC primitives are primitives with local quantum computation and classical communication. For example, in QCCC commitments, sender and receiver are QPT, while the message exchanged between them are classical.}~\cite{C:ChuGolGra24,STOC:KhuTom24,cryptoeprint:2024/1707},
\item 
quantum EFID pairs.\footnote{It is a pair $(G_0,G_1)$ of two QPT algorithms that output classical bit strings
such that the output distributions are statistically far but computationally indistinguishable.
EFID pairs imply OWPuzzs by defining $(\Samp,\Ver)$ as follows.
$\Samp:$ On input $1^\secp$, choose $(b_1,...,b_{\ell(\secp)})\gets\bit^{\ell(\secp)}$ and set $\ans\coloneqq (b_1,...,b_{\ell(\secp)})$,
where $\ell$ is a certain polynomial.
Run $x_{b_i}\gets G_{b_i}(1^\secp)$ for each $i\in[\ell(\secp)]$. Set $\puzz \coloneqq (x_{b_1},...,x_{b_\ell})$.
$\Ver:$ Given $\puzz=(x_{b_1},...,x_{b_\ell})$ and $\ans'=(b_1',...,b_\ell')$, 
output $\top$ if and only if $x_{b_i}$ is in the set of all outputs of $G_{b_i'}(1^\secp)$ for each $i\in[\ell(\secp)]$.}
\end{itemize}
Our \cref{thm:main} therefore highlights that if there is no quantum advantage,
then all of these quantum cryptographic primitives do not exist.

On the other hand, although many primitives imply OWPuzzs, no application of OWPuzzs is known except for commitments~\cite{STOC:KhuTom24}.
Finding more applications of OWPuzzs is one of the most important goals in this field.
Our result \cref{thm:main} shows that OWPuzzs imply quantum advantage, which
demonstrates that quantum advantage is another application of OWPuzzs. 
Moreover, we emphasize that this is the first application of OWPuzzs in the QCCC setting:
IV-PoQ are a QCCC primitive because the communication between the verifier and the prover is classical, while
commitments~\cite{STOC:KhuTom24} constructed from OWPuzzs are those over quantum channels.
The question of the existence of QCCC applications of OWPuzzs was raised in \cite{C:ChuGolGra24}.
We solve the open problem.

\paragraph{Why IV-PoQ?}
In addition to (IV-)PoQ, there are mainly two other approaches to demonstrate quantum advantage,
namely sampling-based quantum advantage and searching-based one.
Here we argue that IV-PoQ capture both of them,
and therefore
identifying a necessary and sufficient assumption for the existence of IV-PoQ is significant.

A sampling problem
is a task of sampling from some distributions. There are several distributions that are easy to sample with QPT algorithms
but hard with PPT algorithms,
such as output distributions of random quantum circuits~\cite{NatPhys:BFNV19}, Boson Sampling circuits~\cite{STOC:AarArk11}, constant-depth circuits~\cite{TD04}, IQP circuits~\cite{BreJozShe10,BreMonShe16}, and one-clean-qubit circuits~\cite{FKMNTT18,MorDQC1additive}.
Several assumptions are known to be sufficient for quantum advantage in sampling problems,
but these assumptions are
newly-introduced assumptions that were not studied before such as an average-case 
$\#{\bf P}$-hardness of approximating some functions.
Moreover, quantum advantage in sampling problems is in general not known to be verifiable (even inefficiently).
On the other hand, one advantage of sampling-based quantum advantage (and others relying on newly-introduced assumptions) is that 
experimental realizations with NISQ machines seem to be easier.\footnote{Although NISQ experimental realizations of quantum advantage are very important goals,
in this paper, we focus on theoretical upper and lower bounds by assuming that any polynomial-time quantum computing is possible.}
As we will explain later, we introduce an average-case version of 
$\mathbf{SampBQP}\neq \mathbf{SampBPP}$\footnote{For the definitions of $\mathbf{SampBQP}$ and $\mathbf{SampBPP}$, see \cref{def:Samplingproblems,def:SampBQP}.}, and show that it
is equivalent to the existence of non-interactive IV-PoQ.
IV-PoQ therefore capture sampling-based quantum advantage.

A search problem is a task of finding an element $z$ that satisfies a relation $R(z)=1$.
Several search problems have been shown to be easy for QPT algorithms but hard for PPT algorithms.
Their classical hardness, however, relies on 
newly-introduced assumptions that were not studied before,
such as QUATH~\cite{CCC:AarChe17} and XQUATH~\cite{AarGun19}, 
or relies on random oracles~\cite{STOC:Aaronson10,STOC:ACCGSW23}.
One advantage of searching-based quantum advantage over sampling-based one is that quantum advantage can be verified at least inefficiently
when $R$ is computable.
(We can check $R(z)=1$ or not by computing $R(z)$.) 
Because such inefficiently-verifiable searching-based quantum advantage is equivalent to the existence of non-interactive IV-PoQ,
IV-PoQ capture inefficiently-verifiable searching-based quantum advantage.
There are some search problems that are efficiently verifiable such as Factoring~\cite{FOCS:Shor94} 
and Yamakawa-Zhandry problem~\cite{10.1145/3658665}
but the former is based on the hardness of a specific problem,
and the latter relies on the random oracle model.
Quantum advantage in efficiently-verifiable search problems is captured by non-interactive PoQ,
and therefore by IV-PoQ.

\paragraph{Summary.}
In summary, we have shown that IV-PoQ are existentially equivalent to classically-secure OWPuzzs.
We believe that this result is significant mainly because of the following five reasons.
\begin{enumerate}
    \item 
    As far as we know, this is the first time that a complete cryptographic characterization of quantum advantage is achieved.
    \item 
    IV-PoQ capture various types of quantum advantage studied so far including sampling-based quantum advantage and searching-based one.
    \item 
    The previous result \cite{C:MorYam24} constructed IV-PoQ from classically-secure OWFs, but the problem
    of constructing IV-PoQ from weaker assumptions was left open. We solve the open problem.
    \item 
    OWPuzzs are implied by many important primitives, such as PRUs, PRSGs, and OWSGs. Therefore, our main result shows that
    if there is no quantum advantage, then these quantum cryptographic primitives do not exist.
    \item
    No application of OWPuzzs was known before except for commitments (and therefore multiparty computations).
    We show that quantum advantage is another application of OWPuzzs.
    Moreover, it is the first QCCC application of OWPuzzs.
    This solves the open problem of \cite{C:ChuGolGra24}.
\end{enumerate}

\begin{figure}[H]
    \centering
    \begin{tikzpicture}[transform shape,scale=0.75]
        \node (QAS/OWF) at (-2,-1) {QAS/OWF} ;
        \node (cs-OWPuzz) at (-5,-1) {cs-OWPuzzs};
        \node (qs-OWPuzz) at (-5,1) {qs-OWPuzzs};
        \node (OWSG) at (-5,2) {OWSGs};
        \node (EFID) at (-7,2) {EFID};
        \node (EFI) at (-7,0) {EFI};
        \node (PRSG) at (-5,3) {PRSGs};
        \node (1-PRSG) at (-8,2.5) {1-PRSGs};
        \node (PRFSGs) at (-5,4) {PRFSGs};
        \node (PRU) at (-5,5) {PRUs};
        \node (qs-OWF) at (-5,6) {qs-OWFs};
        \node (cs-OWF) at (-3,2.5) {cs-OWFs};
        \node (IV-PoQ) at (1,-1) {IV-PoQ};
        \node[align=center] (PC-IV-PoQ) at (-0.5,-3) {public-coin \\ IV-PoQ};
        \node[align=center] (QV-IV-PoQ) at (2.5,-3) {quantum-verifier \\ IV-PoQ};
        \node[align=center] (NI-IV-PoQ) at (1,0.5) {non-interactive \\ IV-PoQ};
        \node[align=center] (HV-SZK-IV-PoQ) at (1,4) {honest-verifier \\ statistical zero-knowledge \\ IV-PoQ};
        \node[align=center] (CZK-IV-PoQ) at (1,2) {computational zero-knowledge \\ IV-PoQ};
        \node (Int-QAS) at (3.5,-1) {Int-QASs};
        \node (QAS) at (3.5,0.5) {QASs};
        \node (OWF or samp) at (-2,-4) {qs-OWFs or $\mathbf{SampBPP}\neq\mathbf{SampBQP}$};
        \node (samp) at (6.7,-0.4) {$\mathbf{SampBPP}\neq\mathbf{SampBQP}$};
        \node (QAA) at (6.5,4) {QAA + $\mathbf{P}^{\mathbf{\#P}}\not\subseteq\mathbf{(io)BPP}^{\mathbf{NP}}$};

        \draw[->,red] (QAA) -- (QAS);
        \draw[->] (QAA) -- (samp);
        
        \draw[->] (qs-OWF) -- (PRU);
        \draw[->] (PRU) -- (PRFSGs);
        \draw[->] (PRFSGs) -- (PRSG);
        \draw[->] (PRSG) -- (OWSG);
        \draw[->] (OWSG) -- (qs-OWPuzz);
        \draw[->] (qs-OWPuzz) -- (cs-OWPuzz);
        \draw[->] (qs-OWF) -- (cs-OWF);
        \draw[->] (cs-OWF) -- (cs-OWPuzz);
        \draw[->] (EFID) -- (qs-OWPuzz);
        \draw[->] (EFID) -- (EFI);
        \draw[->] (qs-OWPuzz) -- (EFI);
        \draw[->] (PRSG) -- (1-PRSG);
        \draw[->] (1-PRSG) -- (EFI);
        \draw[->] (cs-OWF) -- (IV-PoQ);

        \draw[<->,red] (cs-OWPuzz) -- (QAS/OWF);
        \draw[<->,red] (QAS/OWF) -- (IV-PoQ);
        \draw[<->,red] (IV-PoQ) -- (Int-QAS);

        \draw[->] (NI-IV-PoQ) -- (IV-PoQ);
        \draw[->] (QAS) -- (Int-QAS);
        \draw[<->,red] (NI-IV-PoQ) -- (QAS);
        \draw[->,red] (QAS/OWF) -- (OWF or samp);
        \draw[->,red] (QAS) -- (samp);

        \draw[<->,red] (PC-IV-PoQ) -- (IV-PoQ);
        \draw[<->,red] (QV-IV-PoQ) -- (IV-PoQ);

        \draw[->,red] (HV-SZK-IV-PoQ) -- (cs-OWF);
        \draw[->,red] (cs-OWF) -- (CZK-IV-PoQ);
    \end{tikzpicture}
    \caption{Summary of results. Black arrows are known or trivial implications. Red arrows are our results.
    ``qs'' stands for quantumly-secure and ``cs'' stands for classically-secure. 1-PRSGs are single-copy-secure PRSGs~\cite{C:MorYam22}.
    QAA stands for the quantum advantage assumption.
    }\label{fig:result_summary}
\end{figure}

\subsection{Additional Results}
In addition to the main result, \cref{thm:main}, we obtain several important results.
In the following, we explain them.
All our results are also summarized in \cref{fig:result_summary}.

\paragraph{Relations to the sampling complexity.}
In the field of quantum advantage, more studied notion of sampling-based quantum advantage is 
$\mathbf{SampBQP}\neq\mathbf{SampBPP}$. 
We can show the following relation between IV-PoQ and $\mathbf{SampBQP}\neq\mathbf{SampBPP}$:\footnote{Note that ``quantumly-secure'' in the theorem is not a typo. The reason why we can get quantumly-secure OWFs from
classically-secure OWPuzzs is, roughly speaking, that if $\mathbf{SampBPP}=\mathbf{SampBQP}$, then classically-secure OWFs means quantumly-secure OWFs. For details, see \cref{sec:QASOWFcondition}.}
\begin{theorem}
\label{thm:main2}
If IV-PoQ exist, then
quantumly-secure OWFs exist
or $\mathbf{SampBPP}\neq \mathbf{SampBQP}$. 
\end{theorem}

Because the existence of quantumly-secure OWFs implies $\mathbf{NP}\nsubseteq\mathbf{BQP}$,
\cref{thm:main2} also means that
if IV-PoQ exist, then
$\mathbf{NP}\nsubseteq \mathbf{BQP}$ or 
$\mathbf{SampBPP}\neq \mathbf{SampBQP}$. 
This characterizes a lower bound of IV-PoQ in terms of worst-case complexity class assumptions.
Note that this lower bound improves the previous known bound, $\mathbf{BPP}\neq\mathbf{PP}$, of \cite{C:MorYam24}.\footnote{First, $\mathbf{NP}\not\subseteq\mathbf{BQP}$ implies 
$\mathbf{BPP}\neq\mathbf{PP}$, because if $\mathbf{BPP}=\mathbf{PP}$, then $\mathbf{NP}\subseteq\mathbf{PP}=\mathbf{BPP}\subseteq\mathbf{BQP}$.
Second, $\mathbf{SampBPP}\neq\mathbf{SampBQP}$ implies $\mathbf{BPP}\neq\mathbf{PP}$, because a classical deterministic
polynomia-time algorithm that queries the $\mathbf{PP}$ oracle can compute the output distribution of any QPT algorithm~\cite{ForRog99}.}

\paragraph{\bf Quantum advantage samplers (QASs).}
To show the main result, we introduce a new concept, which we call quantum advantage
samplers (QASs). The existence of QASs is an average-case version of $\mathbf{SampBQP}\neq\mathbf{SampBPP}$.

Let $\cA$ be a QPT algorithm that takes $1^\secp$ as input and outputs classical bit strings.
$\cA$ is a quantum advantage sampler (QAS) 
if there exists a polynomial $p$ such that for any PPT algorithm $\cB$,
    \begin{align}     
    \label{QASs}
    \SD(\cA(1^\secp),\cB(1^\secp)) > \frac{1}{p(\secp)}
    \end{align}
holds for all sufficiently large $\secp\in\mathbb{N}$,
where $\SD$ is the statistical distance, 
and $\cA(1^\secp)$ (resp. $\cB(1^\secp)$) is the output probability distribution 
of $\cA$ (resp. $\cB$) on input $1^\secp$. 

Intuitively, \cref{QASs} means that the output distribution of the QPT algorithm $\cA$ cannot be classically efficiently sampled. 
How is this different from the more studied notion, $\mathbf{SampBPP}\neq\mathbf{SampBQP}$?
The existence of QASs can be considered as an average-case version of $\mathbf{SampBPP}\neq\mathbf{SampBQP}$.\footnote{
$\mathbf{SampBPP}$ and $\mathbf{SampBQP}$ are worst-case complexity classes, and therefore
$\{\cA(1^\secp)\}_\secp\not\in\mathbf{SampBPP}$ means that there is at least one $\secp\in\mathbb{N}$
such that $\cA(1^\secp)$ cannot be classically efficiently sampled, while the 
definition of QASs requires that
$\{\cA(1^\secp)\}_\secp$ cannot be classically efficiently sampled for all sufficiently large $\secp\in\mathbb{N}$. 
In addition, there is a subtle technical difference: $\{\cA_\secp\}_\secp \not\in\mathbf{SampBPP}$ means that $\cA(1^\secp)$ cannot be classically efficiently sampled  
for a precision $\epsilon$, but this $\epsilon$ could be $\negl(\secp)$,
while QASs requires that the precision is $1/\poly(\secp)$.
For more details, see \cref{sec:quantum_sampling}.}
In fact, the existence of QASs implies $\mathbf{SampBPP}\neq \mathbf{SampBQP}$ (\cref{lem:QAS_to_SampAd}),
but the inverse does not seem to hold. 
In order to show our main result, the worst-case notion of
$\mathbf{SampBPP}\neq\mathbf{SampBQP}$
is not enough
for our cryptographic (and therefore average-case) argument,
and therefore we introduce QASs.
We believe that the new concept, QASs, will be useful for other
future studies of quantum sampling advantage in the context of quantum cryptography.

We show the following result.
\begin{theorem}\label{thm:samp_search}
QASs exist if and only if non-interactive IV-PoQ exist.    
\end{theorem}
Because
the existence of QASs is an average-case version of $\mathbf{SampBQP}\neq\mathbf{SampBPP}$
while
the existence of non-interactive IV-PoQ is an average-case version of
$\mathbf{FBQP}\neq\mathbf{FBPP}$,
\cref{thm:samp_search} can be considered as an average-case version of
\cite{Aar14}'s result $\mathbf{SampBPP}\neq\mathbf{SampBQP}\Leftrightarrow\mathbf{FBPP}\neq\mathbf{FBQP}$.

In previous works on sampling-based quantum advantage~\cite{NatPhys:BFNV19,STOC:AarArk11,BreMonShe16,MorDQC1additive},
$\mathbf{SampBQP}\neq\mathbf{SampBPP}$
was derived based on three assumptions.
First one is the complexity assumption of $\mathbf{P}^{\mathbf{\# P}}\not\subseteq\mathbf{BPP}^{\mathbf{NP}}$, which especially means that the polynomial-time hierarchy will not collapse
to the third level.
Second one is the assumption that computing some functions (such as matrix permanents and Ising partition functions)
within a certain multiplicative error is $\mathbf{\# P}$-hard on average.
The third one is so-called the anti-concentration, which roughly means that the output probability of the quantum algorithm is not so concentrating.\footnote{For several models, such as the IQP model~\cite{BreMonShe16} and the one-clean-qubit model~\cite{MorDQC1additive}, the anti-concentration can be shown.}
For simplicity, we call the combination of the last two assumptions just quantum advantage assumption.

The quantum advantage assumption has been traditionally studied in the field of quantum advantage to show sampling-based quantum advantage (in terms of $\mathbf{SampBQP}\neq\mathbf{SampBPP}$) of several ``non-universal'' models such as
random quantum circuits~\cite{NatPhys:BFNV19}, Boson Sampling circuits~\cite{STOC:AarArk11}, IQP circuits~\cite{BreMonShe16}, and one-clean-qubit circuits~\cite{MorDQC1additive}.
We can show that the quantum advantage assumption also implies
the existence of QASs.\footnote{This result was not included in the previous version of this manuscript. We obtained this result after reading \cite{cryptoeprint:2024/1490}.}
\begin{theorem}
If the quantum advantage assumption holds and $\mathbf{P}^{\mathbf{\# P}}\not\subseteq\mathbf{ioBPP}^{\mathbf{NP}}$, then QASs exist.    
\end{theorem}
As we have explained above, QASs are cryptographically a more natural notion of sampling-based quantum advantage, and the existence of QASs seems to be stronger than
$\mathbf{SampBQP}\neq\mathbf{SampBPP}$. Therefore this theorem reveals that the quantum advantage assumption traditionally studied in the field actually implies
a stronger form of quantum advantage
(modulo the difference between $\mathbf{BPP}$ and $\mathbf{ioBPP}$).

\paragraph{\bf Interactive quantum advantage samplers (Int-QASs).}
We also introduce an interactive version of QASs, which we call interactive QASs (Int-QASs).
Int-QASs  
are a generalization of QASs to interactive settings.
An Int-QAS is a pair $(\cA,\cC)$ of two interactive QPT algorithms $\cA$ and $\cC$ that 
communicate over a classical channel. Its security roughly says that
no PPT algorithm $\cB$ that interacts with $\cC$ can sample the transcript of $(\cA,\cC)$.
(Here, the transcript is the sequence of all classical messages exchanged between $\cA$ and $\cC$.)
More precisely, 
$(\cA,\cC)$ is an Int-QAS if
there exists a polynomial $p$ such that for 
any PPT algorithm $\cB$ that interacts with $\cC$,
    \begin{equation}     
    \SD(\langle\cA,\cC\rangle(1^\secp),\langle\cB,\cC\rangle(1^\secp)) > \frac{1}{p(\secp)}
    \end{equation}
holds for all sufficiently large $\secp\in\mathbb{N}$,
where $\SD$ is the statistical distance, and
$\langle\cA,\cC\rangle(1^\secp)$ (resp. $\langle\cB,\cC\rangle(1^\secp)$) is the probability distribution over the transcript of the interaction between $\cA$ (resp. $\cB$) and $\cC$
on input $1^\secp$. 

It is easy to see that IV-PoQ imply Int-QASs: for any IV-PoQ $(\cP,\cV_1,\cV_2)$, where $\cP$ is the prover,
$\cV_1$ is the efficient verifier, and $\cV_2$ is the inefficient verifier, we have only to take $(\cA,\cC)=(\cP,\cV_1)$.
However, the opposite direction is not immediately clear. We show that the opposite direction can be also shown.
We hence have the following result.
\begin{theorem}
\label{thm:IVPoQ_IntQAS}
Int-QASs exist
if and only if 
IV-PoQ exist. 
\end{theorem}
This theorem is considered as an interactive (and average-case) version 
of $\mathbf{SampBPP}\neq\mathbf{SampBQP}\Leftrightarrow
\mathbf{FBPP}\neq\mathbf{FBQP}$,
because 
the existence of Int-QASs is an interactive (and average-case) version of $\mathbf{SampBQP}\neq\mathbf{SampBPP}$
while
the existence of IV-PoQ is an interactive (and average-case) version of $\mathbf{FBQP}\neq\mathbf{FBPP}$.

\paragraph{\bf QAS/OWF condition.}
In addition to QASs and Int-QASs, we also introduce another new concept, the QAS/OWF condition, which is inspired by the SZK/OWF condition~\cite{Vad06}.
As we will explain later, the QAS/OWF condition plays a pivotal role to show our main result. 
Roughly speaking, the QAS/OWF condition is satisfied if 
there is a pair of candidates of a QAS and a classically-secure OWF such that for all sufficiently large security parameters, either of them is secure.\footnote{See \cref{def:QAS/OWF} for the precise definition.} If a QAS exists or a classically-secure OWF exists, then the QAS/OWF is satisfied, but the converse is unlikely. For example, if there are (candidates of) a QAS that is secure for all odd security parameters and a OWF that is secure for all even security parameters, then the QAS/OWF is satisfied, but it does not necessarily imply either of a QAS or a OWF.  

We show that the QAS/OWF condition is equivalent to both
the existence of IV-PoQ and the existence of classically-secure OWPuzzs:
\begin{theorem}
\label{thm:IVPoQ_QASOWF}
IV-PoQ exist if and only if the QAS/OWF condition holds.
\end{theorem}
\begin{theorem}
Classically-secure OWPuzzs exist if and only if the QAS/OWF condition holds.
\end{theorem}
By combining these two results, we obtain our main result, \cref{thm:main}.

\paragraph{Variants of IV-PoQ.}
Recall that the verifier of IV-PoQ must be PPT during the interaction. We consider the following two variants of IV-PoQ, \emph{public-coin IV-PoQ}, where all the verifier's messages must be uniformly random strings, and \emph{quantum-verifier IV-PoQ}, where the verifier is allowed to be QPT instead of PPT during the interaction. Clearly, public-coin IV-PoQ is a special case of IV-PoQ (since uniformy random strings can be sampled in PPT), and IV-PoQ is a special case of quantum-verifier IV-PoQ (since PPT computations can be simulated in QPT).  
We show implications in the other direction, making them equivalent in terms of existence. \begin{theorem}\label{thm:intro_equivalence_variants}
    The existence of public-coin IV-PoQ, IV-PoQ, and quantum-verifier IV-PoQ are equivalent. 
\end{theorem} 
This theorem suggests that the power of IV-PoQ is robust to the choice of the computational power of the verifier during the interaction. 

\paragraph{Zero-knowledge IV-PoQ.}
We define the zero-knowledge property for IV-PoQ, which roughly requires that the verifier's view can be simulated by a PPT simulator. Intuitively, this ensures that the verifier learns nothing from the prover beyond what could have been computed in PPT. We say that an IV-PoQ satisfies statistical (resp. computational) zero-knowledge if for any PPT malicious verifier,  there is a PPT simulator that statistically (resp. computationally) simulates the verifier's view. We say that an IV-PoQ satisfies honest-verifier statistical zero-knowledge if there is a PPT simulator that statistically simulates the honest verifier's view. We prove the following results. 
\begin{theorem}
    If honest-verifier statistical zero-knowledge IV-PoQ exist, then classically-secure OWFs exist.
\end{theorem}
\begin{theorem}
    If classically-secure OWFs exist, then computational zero-knowledge IV-PoQ exist.
\end{theorem}
The above theorems establish a loose equivalence between zero-knowledge IV-PoQ and classically-secure OWFs. 
However, there is still a gap between them, and filling it is left as an open problem.

\subsection{Technical Overview}\label{sec:overview}
Our main result is \cref{thm:main}, which shows that IV-PoQ exist if and only if classically-secure OWPuzzs exist.
In this technical overview, we provide intuitive explanations for the result.
To show it, the QAS/OWF condition plays a pivotal role.
For ease of presentation, we think of the QAS/OWF condition as just the condition that ``a QAS exists or a classically-secure OWF exists'' in this overview. Though this is stronger than the actual definition, this rough description is enough for understanding our ideas.  

We first show IV-PoQ $\Leftrightarrow$ QAS/OWF condition.
We next show classically-secure OWPuzzs $\Leftrightarrow$ QAS/OWF condition.
By combining them, we finally obtain the main result.
In the following, we explain each step.

\paragraph{Step 1: IV-PoQ $\Rightarrow$ QAS/OWF condition.}
Its proof is inspired by \cite{CCC:Ost91}. 
However, we emphasize that our proof is not a trivial application of \cite{CCC:Ost91}. As we will explain later,
the same proof of \cite{CCC:Ost91} does not work in our setting, and therefore
we had to overcome several technical challenges to show the result.

Let $(\cP,\cV_1,\cV_2)$ be an $\ell$-round IV-PoQ,
where $\cP$ is the prover, $\cV_1$ is the efficient verifier, and $\cV_2$ is the inefficient verifier.
(We count two messages as a single round.)
Without loss of generality, we can assume that $\cV_1$ first sends a message.
Let $(c_1,a_1,...,c_\ell,a_\ell)$ be the transcript (i.e., the sequence of all messages exchanged) of the interaction between $\cP$ and $\cV_1$, 
where $c_i$ is $\cV_1$'s $i$-th message and $a_i$ is $\cP$'s $i$-th message.
Our goal is, by assuming that the QAS/OWF condition is not satisfied, to construct a classical PPT adversary $\cP^*$
that breaks the soundness of the IV-PoQ.

If the QAS/OWF condition is not satisfied, then, roughly speaking,
both of the following two conditions are satisfied:
\begin{enumerate}
    \item[(a)]
QASs do not exist. In other words, the output probability distribution of any QPT algorithm can be approximately sampled with a PPT algorithm.
\item[(b)]
Classically-secure OWFs do not exist.
\end{enumerate}
From (a), the distribution of the transcript generated by the interaction between $\cP$ and $\cV_1$ 
can be approximately sampled with a PPT algorithm $\cS$.
One might think that if a malicious PPT prover of the IV-PoQ just runs $\cS$, the soundness of the IV-PoQ is broken.
However, it is not correct: The ability to classically efficiently sample from the distribution $\langle \cP,\cV_1\rangle(1^\secp)$ 
is not enough to break the soundness of the IV-PoQ, because what the PPT adversary $\cP^*$ has to do is not to sample $(c_1,a_1,...,c_\ell,a_\ell)$
but to sample ``correct'' $a_k$ given the transcript $(c_1,a_1,...,c_{k-1},a_{k-1},c_k)$ obtained so far for every $k\in[\ell]$.

We use (b) to solve it.
From $\cS$, we define a function $f$ as follows.
\begin{enumerate}
    \item 
    Get an input $(k,r)$.
    \item 
    Run $(c_1,a_1,...,c_\ell,a_\ell)=\cS(1^\secp;r)$.\footnote{For a PPT algorithm $\cA$, $y=\cA(x;r)$ means that $\cA$'s output is $y$ when the input is $x$ and the random seed is $r$.}
    \item 
    Output $(k,c_1,a_1,...,c_{k-1},a_{k-1},c_k)$.
\end{enumerate}
From (b), OWFs do not exist. Then, distributional OWFs do not exist as well~\cite{FOCS:ImpLub89}.\footnote{An efficiently-computable function $f:\bit^*\to\bit^*$ is called
a classically-secure (resp. quantumly-secure) distributional OWF if for any PPT (resp. QPT) adversary $\cA$, 
the statistical distance between $(x,f(x))$ and $(\cA(f(x)),f(x))$ is large for random $x$.}
This means that there exists a PPT algorithm $\cR$ such that
the statistical distance between $(x,f(x))$ and $(\cR(f(x)),f(x))$ is small for random $x$.
Therefore, for each $k\in[\ell]$, the following PPT adversary $\cP^*$ can return ``correct'' $a_k$ given
$(c_1,a_1,...,c_{k-1},a_{k-1},c_k)$.
\begin{enumerate}
\item
Take $(c_1,a_1,...,c_{k-1},a_{k-1},c_k)$ as input.
    \item 
    Run $(k',r')\gets \cR(k,c_1,a_1,...,c_{k-1},a_{k-1},c_k)$.
    \item 
    Run $(c_1',a_1',...,c_\ell',a_\ell')=\cS(1^\secp;r')$.
    \item 
    Output $a'_{k'}$.
\end{enumerate}
In this way, we can break the soundness of the IV-PoQ.
Hence we have shown that IV-PoQ $\Rightarrow$ the QAS/OWF condition.

The idea underlying this proof is similar to that of \cite{CCC:Ost91}. In \cite{CCC:Ost91}, it was shown that
if SZK is average hard then OWFs exist.
To show it, \cite{CCC:Ost91} used the zero-knowledge property to
guarantee the existence of a PPT simulator that can sample the transcript between the verifier and the prover.
From that simulator, \cite{CCC:Ost91} constructed a OWF.
Very roughly speaking, our $\cS$ that comes from (a) corresponds to the zero-knowledge simulator of \cite{CCC:Ost91}:
the transcript of \cite{CCC:Ost91} can be PPT sampled because of the zero-knowledge property
while our transcript can be PPT sampled because QASs do not exist.
However, there are several crucial differences between our setting and \cite{CCC:Ost91}'s.
In particular,
in the setting of \cite{CCC:Ost91} the constructed OWF can depend on the simulator.
On the other hand, in our setting, we finally want to construct a OWF that is independent of $\cS$.\footnote{In the precise definition of QAS/OWF condition,
the OWFs should be independent of $\cS$. Otherwise, we do not know how to show
the other direction, namely, the QAS/OWF condition $\Rightarrow$ IV-PoQ.}
In order to solve the issue,
we use the universal construction of OWFs \cite{STOC:Levin85,EC:HKNRR05}.
We first construct a OWF $f_\cS$ from each $\cS$ as we have explained above, and next
construct a OWF $g$ that is independent of $\cS$
by using the universal construction.
In addition to this issue, there are several other points where the direct application of \cite{CCC:Ost91} does not work, but for details,
see the main body of the paper.

Note that in the actual proof, we do not directly show
IV-PoQ $\Rightarrow$ the QAS/OWF condition.
We first show 
IV-PoQ $\Rightarrow$ Int-QAS, and then show
Int-QAS $\Rightarrow$ the QAS/OWF condition
in order to obtain stronger results
and to avoid repeating similar proofs twice.
However, the proof of
Int-QAS $\Rightarrow$ the QAS/OWF condition is essentially the same as that explained above.

\paragraph{Step 2: Classically-secure OWPuzzs $\Rightarrow$ QAS/OWF condition.}
Its proof is similar to that of step 1.
Let $(\Samp,\Ver)$ be a classically-secure OWPuzz.
Assume that the QAS/OWF condition is not satisfied. This roughly means that both of the following two conditions are satisfied.
\begin{itemize}
 \item[(a)]
QASs do not exist.
\item[(b)]
Classically-secure OWFs do not exist.
\end{itemize}
From (a), there exists a PPT algorithm $\cS$ such that the output probability distribution of $\cS(1^\secp)$ is
close to that of $\Samp(1^\secp)$ in the statistical distance.
From such $\cS$, we construct a function $f$ as follows.
\begin{enumerate}
    \item 
    Get a bit string $r$ as input.
    \item 
    Compute $(\puzz,\ans)=\cS(1^\secp;r)$.
    \item 
    Output $\puzz$.
\end{enumerate}
From (b), OWFs do not exist. This means that distributional OWFs do not exist as well.
Therefore, there exists a PPT algorithm $\cR$ such that the statistical distance
between $(x,f(x))$ and $(\cR(f(x)),f(x))$ is small for random $x$. 
From $\cS$ and $\cR$, we can construct a PPT adversary $\cA$ that breaks the security of the OWPuzz as follows:
\begin{enumerate}
    \item 
    Take $\puzz$ as input.
    \item 
    Run $r\gets\cR(\puzz)$.
    \item 
    Run $(\puzz',\ans')\gets\cS(1^\secp;r)$.
    \item 
    Output $\ans'$.
\end{enumerate}
As in step 1, we actually need 
a OWF that is independent of $\cS$, and therefore we have to use the universal construction \cite{STOC:Levin85,EC:HKNRR05}.

\paragraph{Step 3: QAS/OWF condition $\Rightarrow$ IV-PoQ.}
Assume that the QAS/OWF condition is satisfied. Then, roughly speaking, a QAS $\cQ$ exists 
or a classically-secure OWF $f$ exists.
From $f$, we can construct an IV-PoQ by using the result of \cite{C:MorYam24}.
The non-trivial part is to construct an IV-PoQ from $\cQ$. 
For that goal, we
use the idea of \cite{Aar14}. However, as we will explain later, a direct application of \cite{Aar14}
does not work for our goal, and some new technical contributions were needed.

We construct a non-interactive IV-PoQ from $\cQ$ as follows:
\begin{enumerate}
    \item 
    The QPT prover runs $\cQ(1^\secp)$ $N$ times, where $N$ is a certain polynomial,
    and sends the result $(y_1,...,y_N)$ to the verifier, where $y_i$ is the output of the $i$-th run of $\cQ(1^\secp)$.
    \item 
    The unbounded verifier computes the Kolmogorov complexity $K(y_1,...,y_N)$\footnote{More precisely, this is 
   time-bounded prefix Kolmogorov complexity $K_U^T(y_1,...,y_N)$ with time bound $T(n)=2^{2^{n}}$
   and the universal self-delimiting machine $U$.}, and accepts if
    it is larger than a certain value.
\end{enumerate}
We can show that thus constructed non-interactive IV-PoQ satisfies completeness and soundness.
For completeness, we 
use Markov's inequality to show that $K(y_1,...,y_N)$, where $y_i\gets\cQ(1^\secp)$ for each $i\in[N]$, is large with high probability and therefore the verifier accepts.
To evaluate the bound of Markov's inequality, we
use Kraft's inequality for the prefix Kolmogorov complexity, which says that $\sum_x2^{-K(x)}\le 1$. 
For soundness, assume that there exists a PPT adversary $\cP^*$ that outputs
$(y_1',...,y_N')$ that is accepted by the verifier with high probability,
which means that $\log\frac{1}{\Pr[(y_1',...,y_N')\gets\cQ(1^\secp)^{\otimes N}]}\lesssim K(y_1',...,y_N')$.
Because of the property of $K$, we have
$K(y_1',...,y_N')\lesssim \log \frac{1}{\Pr[(y_1',...,y_N')\gets \cP^*(1^\secp)]}$.
By combining them, we have $\log\frac{\Pr[(y_1'..,y_N')\gets\cP^*(1^\secp)]}{\Pr[(y_1',...,y_N')\gets\cQ(1^\secp)^{\otimes N}]}\approx 0$,
which roughly means that the output probability distribution of $\cQ(1^\secp)^{\otimes N}$ is close to that of $\cP^*(1^\secp)$.

From such $\cP^*$,
we can construct a PPT algorithm 
whose output probability distribution is close to that of $\cQ(1^\secp)$ in the statistical distance
as follows:
\begin{enumerate}
    \item 
Run $(y_1',...,y_N')\gets\cP^*(1^\secp)$.
\item 
Choose a random $i\in[N]$, and output $y_i'$.
\end{enumerate}
However, this means that $\cQ$ is not a QAS, which contradicts the assumption.
In this way, we can show the QAS/OWF condition $\Rightarrow$ IV-PoQ.

The idea underlying this proof is similar to that of \cite{Aar14}. In fact, the completeness part is exactly the same.
However, for the soundness part, the direct application of \cite{Aar14} does not work, because of several reasons. 
Here we explain main two issues.
First, the search problem constructed in \cite{Aar14} was not necessarily verifiable even in unbounded time since Kolmogorov complexity is uncomputable in general.
This is problematic for our goal, because what we want to construct is a non-interactive IV-PoQ where the prover's message should be verified {\it at least inefficiently}. 
\cite{Aar14} slightly mentioned an extension of the result to the time-bounded case, but there was no proof.
Second, \cite{Aar14} constructed a search advantange from $\mathbf{SampBQP}\neq\mathbf{SampBPP}$, which is a worst-case notion.
However, what we need is a search advantage from the existence of QASs, namely, the {\it average-case version} of sampling advantage.
Hence the proof of \cite{Aar14} cannot be directly used in our setting.

\paragraph{Step 4: QAS/OWF condition $\Rightarrow$ classically-secure OWPuzzs.}
The proof uses a similar technique as that of step 3. 
If the QAS/OWF condition is satisfied, then, roughly speaking, a classically-secure OWF $f$ exists or a QAS $\cQ$ exists.
From the OWF $f$, we can construct a classically-secure OWPuzz easily as follows:
\begin{itemize}
    \item 
$\Samp(1^\secp)\to(\puzz,\ans):$ Choose $x\gets\bit^\secp$, and output $\puzz\coloneqq f(x)$ and $\ans\coloneqq x$.
\item 
$\Ver(\puzz,\ans')\to\top/\bot:$ Accept if and only if $f(\ans')=\puzz$.
\end{itemize}

From $\cQ$, we can construct a non-interactive IV-PoQ as in step 3.
From such a non-interactive IV-PoQ, we can easily construct a classically-secure OWPuzz as follows.
\begin{itemize}
    \item 
$\Samp(1^\secp)\to(\puzz,\ans):$ Run $\tau\gets\cP(1^\secp)$, and output $\puzz\coloneqq 1^\secp$ and $\ans\coloneqq \tau$.
\item 
$\Ver(\puzz,\ans')\to\top/\bot:$ Accept if and only if $\top\gets\cV_2(1^\secp,\ans')$.
\end{itemize}

\subsection{Related Work}
Khurana and Tomer~\cite{cryptoeprint:2024/1490} have recently shown that {\it quantumly-secure} OWPuzzs can be constructed from some assumptions
that imply sampling-based quantum advantage (if a mild complexity assumption, $\mathbf{P}^{\sharp \mathbf{P}}\not\subseteq (io)\BQP/\mathbf{qpoly}$,
is additionally introduced).
There is no technical overlap between their paper and the present paper. 
However, we here clarify relations and differences, because
in a broad perspective,
their paper and the present paper share several important motivations,
including the goal of connecting quantum advantage and ``Microcrypt'' primitives.

Firstly, 
what they actually show is not that quantum advantage implies OWPuzzs, but
that some assumptions that imply quantum advantage also imply OWPuzzs if the additional assumption,
$\mathbf{P}^{\sharp \mathbf{P}}\not\subseteq (io)\BQP/\mathbf{qpoly}$,
is introduced.
On the other hand, we show that quantum advantage (in the sense of the existence of IV-PoQ)
implies OWPuzzs.
Secondly,
they construct quantumly-secure OWPuzzs, while we construct only classically-secure ones.

These differences comes from the difference of main goals. 
Their goal is to construct quantum cryptographic primitives from some well-founded assumptions that will not imply OWFs.
Therefore, the constructed primitives should be quantumly-secure.
However, in that case,
as they also mention in their paper, 
some additional assumptions that limit quantum power should be introduced,
because quantum advantage limits only classical power.
On the other hand, the goal of the present paper is to characterize quantum advantage from cryptographic assumptions,
and therefore we have to consider quantum advantage itself, not assumptions that imply quantum advantage.
Moreover, we want to avoid introducing any additional assumptions that are not related to quantum advantage. In that case,
it is likely that we have to be satisfied with classically-secure OWPuzzs.

It is an interesting open problem whether several notions of quantum advantage
studied in this paper imply quantumly-secure OWPuzzs (possibly introducing some additional assumptions
that limit quantum power).
\section{Preliminaries}\label{sec:preliminaries_2}
\subsection{Basic Notations} 
$\log x$ means $\log_2 x$ and $\ln x$ means $\log_e x$.
We use standard notations of quantum computing and cryptography.
For a bit string $x$, $|x|$ is its length.
$\mathbb{N}$ is the set of natural numbers.
We use $\secp$ as the security parameter.
$[n]$ means the set $\{1,2,...,n\}$.
For a finite set $S$, $x\gets S$ means that an element $x$ is sampled uniformly at random from the set $S$.
$\negl$ is a negligible function, and $\poly$ is a polynomial.
All polynomials appear in this paper are positive, but for simplicity we do not explicitly mention it.
PPT stands for (classical) probabilistic polynomial-time and QPT stands for quantum polynomial-time. 
For an algorithm $\cA$, $y\gets \cA(x)$ means that the algorithm $\cA$ outputs $y$ on input $x$.
If $\cA$ is a classical probabilistic or quantum algorithm that takes $x$ as input and outputs bit strings,
we often mean $\cA(x)$ by the output probability distribution of $\cA$ on input $x$.
When $\cA$ is a classical probabilistic algorithm, $y=\cA(x;r)$ means that the output of $\cA$ is $y$ if it runs on input $x$ and with the random seed $r$.
For two interactive algorithms $\cA$ and $\cB$ that interact over a classical channel, $\tau\gets\langle \cA(x),\cB(y)\rangle$
means that the transcript $\tau$ (i.e., the sequence of all messages exchanged) is generated by the interactive protocol between $\cA$ and $\cB$
where $\cA$ takes $x$ as input and $\cB$ takes $y$ as input.
If both $\cA$ and $\cB$ take the same input $x$, we also write
it as $\tau\gets\langle \cA,\cB\rangle(x)$.
For two quantum states $\rho$ and $\sigma$, $\TD(\rho,\sigma)\coloneqq\frac{1}{2}\|\rho-\sigma\|_1$ means their trace distance,  
where $\|X\|_1\coloneqq\Tr\sqrt{X^\dagger X}$ is the trace norm.
For two probability distributions $P\coloneqq\{p_i\}_i$ and $Q\coloneqq\{q_i\}_i$, 
$\SD(Q,P)\coloneqq\frac{1}{2}\sum_i|p_i-q_i|$ is their statistical distance.
If $\rho=\sum_i p_i |\phi_i\rangle\langle \phi_i|$ and $\sigma=\sum_i q_i |\phi_i\rangle\langle \phi_i|$ for some orthonormal basis $\{|\phi_i\rangle\}_i$, we have
$\TD(\rho,\sigma)=\SD(\{p_i\}_i,\{q_i\}_i)$.

\subsection{One-Way Functions and Distributional One-Way Functions}

We first review the definition of one-way functions (OWFs).
\begin{definition}[One-Way Functions (OWFs)]
\label{def:OWFs}
A function $f:\bit^*\to\bit^*$ 
that is computable in classical deterministic polynomial-time
is a classically-secure (resp. quantumly-secure) one-way function (OWF) if
for any PPT (resp. QPT) adversary $\cA$ and any polynomial $p$, 
\begin{equation} \label{eq:OWF_condition}
\Pr[f(x')=f(x):
x\gets\bit^\secp,
x'\gets\cA(1^\secp,f(x))
]
\le\frac{1}{p(\secp)}
\end{equation} 
holds
for all sufficiently large $\secp\in\mathbb{N}$.
\end{definition}

We define a variant of OWFs, which we call OWFs on a subset $\Sigma\subseteq\mathbb{N}$. 
The difference from the standard OWFs is that security holds only when the security parameter belongs to the subset $\Sigma$ of $\mathbb{N}$.  
\begin{definition}[OWFs on $\Sigma$]\label{def:OWFsSigma}
    Let $\Sigma\subseteq\mathbb{N}$ be a set.
    A function $f:\bit^*\to\bit^*$ 
    that is computable in classical deterministic polynomial-time
    is a classically-secure (resp. quantumly-secure) OWF on $\Sigma$ if
    there exists an efficiently-computable polynomial $n$ such that
    for any PPT (resp. QPT) adversary $\cA$ and any 
    polynomial $p$
    there exists $\secp^*\in\mathbb{N}$ such that
    \begin{equation}
    \Pr[f(x')=f(x): x\gets\bit^{n(\secp)}, x'\gets\cA(1^{n(\secp)},f(x))] \le\frac{1}{p(\secp)}
    \end{equation} 
    holds
    for all $\secp\ge\secp^*$ in $\Sigma$. 
\end{definition}
\begin{remark}
    In the definition of OWFs (\Cref{def:OWFs}) the input length is treated as the security parameter, but in OWFs on $\Sigma$ (\Cref{def:OWFsSigma}), we allow the input length to be an arbitrary polynomial in the security parameter.  
\end{remark}
\begin{remark}
For any finite $\Sigma$, OWFs on $\Sigma$ always exist    
because the definition is trivially satisfied. (We have only to take $\secp^*=\secp_{max}+1$, where
$\secp_{max}$ is the largest element of $\Sigma$.)
However, we include the case when $\Sigma$ is finite in the definition for convenience. 
\end{remark}

The existence of OWFs on $\mathbb{N}\setminus \Sigma$ for a finite subset $\Sigma$ is actually equivalent to that of the standard OWFs.
Its proof is given in \cref{sec:OWF_on_N}.
\begin{lemma}
\label{lem:OWF_on_N}
Let $\Sigma\subseteq \mathbb{N}$ be a finite set. 
Classically-secure (resp. quantumly-secure) OWFs exist if and only if classically-secure (resp. quantumly-secure) OWFs on $\mathbb{N}\setminus \Sigma$ exist. 
\end{lemma}

We also review the definition of distributional one-way functions (DistOWFs).

\begin{definition}[Distributional One-Way Functions (DistOWFs) \cite{FOCS:ImpLub89}]
\label{def:DistOWFs}
A function $f:\bit^*\to\bit^*$ 
that is computable in classical deterministic polynomial-time
is a classically-secure (resp. quantumly-secure) distributional one-way function (DistOWF) if
there exists a polynomial $p$ such that
\begin{align}
\SD(
\{(x,f(x))\}_{x\gets\bit^\secp},
\{(\cA(1^\secp,f(x)),f(x))\}_{x\gets\bit^\secp})
\ge\frac{1}{p(\secp)}
\end{align}
holds
for any PPT (resp. QPT) adversary $\cA$ and for all sufficiently large $\secp\in\mathbb{N}$.  
\end{definition}

Similarly to OWFs on $\Sigma$, we define DistOWFs on $\Sigma$ for a subset $\Sigma\subseteq\mathbb{N}$.  

\begin{definition}[DistOWFs on $\Sigma$]
    \label{def:DistOWF}
    Let $\Sigma\subseteq\mathbb{N}$ be a set.
    A function $f:\bit^*\to\bit^*$ 
    that is computable in classical deterministic polynomial-time
    is a classically-secure (resp. quantumly-secure) DistOWF on $\Sigma$ if
    there exist an efficiently-computable polynomial $n$ and a polynomial $p$ such that
    for any PPT (resp. QPT) adversary $\cA$ there exists $\secp^*\in\mathbb{N}$
    such that
    \begin{align}
    \SD(\{(x,f(x))\}_{x\gets\bit^{n(\secp)}},\{(\cA(1^{n(\secp)},f(x)),f(x))\}_{x\gets\bit^{n(\secp)}}) \ge\frac{1}{p(\secp)}
    \end{align}
    holds
    for all $\secp\ge\secp^*$ in $\Sigma$. 
\end{definition}
\begin{remark}
Again, for any finite $\Sigma$, DistOWFs on $\Sigma$ always exist    
because the definition is trivially satisfied. 
However, we include the case when $\Sigma$ is finite in the definition for convenience. 
\end{remark}

It is well-known that OWFs exist if and only if DistOWFs exist \cite[Lemma 1]{FOCS:ImpLub89}.\footnote{\cite{FOCS:ImpLub89} only provides a proof sketch. Its full proof can be found in \cite[Theorem 4.2.2]{Thesis:Impagliazzo92}} By inspecting its proof, one can see that the proof gives a ``security-parameter-wise'' reduction, i.e., if the base DistOWF is secure on inputs of length $\secp$, then the resulting OWF is secure on inputs of length $n(\secp)$ for some efficiently computable polynomial $n$.\footnote{More precisely, 
for any efficiently computable function $f:\bit^*\rightarrow \bit^*$ and polynomial $p$, there are an efficiently computable function $g:\bit^*\rightarrow \bit^*$ and efficiently computable polynomial $n$ that satisfy the following:  
For any PPT adversary $\cA$ and a polynomial $q$,  there is a PPT adversary $\cB$ such that for any $\secp\in \mathbb{N}$, 
if 
  \begin{equation}
    \Pr[g(x')=g(x): x\gets\bit^{n(\secp)}, x'\gets\cA(1^{n(\secp)},g(x))] > \frac{1}{q(\secp)}, 
    \end{equation} 
then 
\begin{align}
\SD(
\{(x,f(x))\}_{x\gets\bit^\secp},
\{(\cB(1^\secp,f(x)),f(x))\}_{x\gets\bit^\secp})
< \frac{1}{p(\secp)}.
\end{align}
}   
This implies equivalence between OWFs on $\Sigma$ and DistOWFs on $\Sigma$ on any subset $\Sigma\subseteq \mathbb{N}$.  
\begin{lemma}[Based on {\cite[Lemma 1]{FOCS:ImpLub89}}]\label{lem:OWF_DistOWF}
    For any subset $\Sigma\subseteq\mathbb{N}$, classically-secure DistOWFs on $\Sigma$ exist if and only if classically-secure OWFs on $\Sigma$ exist. 
\end{lemma}

It is known that there exists a \emph{universal construction} for OWFs~\cite{STOC:Levin85}, a concrete function that is a OWF if and only if OWFs exist. 
We observe that a similar technique yields a universal construction for OWFs on $\Sigma$ as well.  
Its proof is given in \cref{sec:univ_OWF}.
\begin{lemma}\label{lem:univ_OWF}
    There exists a function $g:\bit^*\to\bit^*$ that is computable in classical deterministic polynomial-time 
    such that for any subset $\Sigma\subseteq\mathbb{N}$, if there exist classically-secure OWFs on $\Sigma$, then $g$ is a classically-secure OWF on $\Sigma$.
\end{lemma}

By combining \Cref{lem:OWF_DistOWF,lem:univ_OWF}, we obtain the following corollary,
which will be used later.
\begin{corollary}\label{cor:univ_DistOWF}
    There exists a function $g:\bit^*\to\bit^*$ that is computable in classical deterministic polynomial-time
    such that for any subset $\Sigma\subseteq\mathbb{N}$, if there exist classically-secure DistOWFs on $\Sigma$, then $g$ is a classically-secure OWF on $\Sigma$.
\end{corollary}

\subsection{One-Way Puzzles}
We also define one-way puzzles (OWPuzzs) on a subset $\Sigma\subseteq\mathbb{N}$,
which are a generalization of OWPuzzs defined in \cite{STOC:KhuTom24}.
If $\Sigma=\mathbb{N}$, the definition becomes the standard one~\cite{STOC:KhuTom24}, and
in that case we call them just OWPuzzs. 

\begin{definition}[One-Way Puzzles (OWPuzzs) on $\Sigma$]
Let $\Sigma\subseteq\mathbb{N}$ be a set.
A one-way puzzle (OWPuzz) on $\Sigma$ is a pair $(\Samp, \Ver)$ of algorithms such that
\begin{itemize}
\item 
$\Samp(1^\secp)\to (\puzz,\ans):$
It is a QPT algorithm that, on input the security parameter $\secp$, 
outputs a pair 
$(\puzz,\ans)$
of classical strings.
\item
$\Ver(\puzz,\ans')\to\top/\bot:$
It is an unbounded algorithm that, on input $(\puzz,\ans')$, outputs either $\top/\bot$.
\end{itemize}
They satisfy the following properties for some functions $c$ and $s$ such that $c(\secp)-s(\secp)\ge\frac{1}{\poly(\secp)}$.
\begin{itemize}
    \item 
{\bf $c$-correctness:}
There exists $\secp^*\in\mathbb{N}$ such that 
\begin{align}
\Pr[\top\gets\Ver(\puzz,\ans):
(\puzz,\ans)\gets\Samp(1^\secp)]
\ge c(\secp)
\end{align}
holds
for all $\secp\ge\secp^*$.
\item
{\bf $s$-security on $\Sigma$:}
For any QPT adversary $\cA$ there exists $\secp^{**}\in\mathbb{N}$ such that
\begin{align}
\Pr[\top\gets\Ver(\puzz,\cA(1^\secp,\puzz)):
(\puzz,\ans)\gets \Samp(1^\secp)] \le s(\secp)
\end{align}
holds for all $\secp\ge\secp^{**}$ in $\Sigma$.
\end{itemize}
\end{definition}

\begin{definition}[Classically-Secure OWPuzzs on $\Sigma$]\label{def:OWPuzz}
A OWPuzz on $\Sigma$ is called a classically-secure OWPuzz on $\Sigma$ if the security is required against PPT adversaries.
\end{definition}

\begin{remark}
Again, if $\Sigma$ is a finite set, OWPuzzs on $\Sigma$ trivially exist, but
we include such a case in the definition for the convenience.
\end{remark}

    \begin{remark}
All classically-secure OWPuzzs appearing in this paper are ones with $(1-\negl(\secp))$-correctness and
$(1-1/\poly(\secp))$-security.
    \end{remark}

\begin{remark}
It is known that $c$-correct and $s$-secure OWPuzzs with $c(\secp)-s(\secp)\ge 1/\poly(\secp)$ can be amplified to 
$(1-\negl(\secp))$-correct and $\negl(\secp)$-secure OWPuzzs~\cite{C:ChuGolGra24}. 
On the other hand, we do not know how to amplify the gap of classically-secure OWPuzzs.
\end{remark}

OWPuzzs can be constructed from OWFs.
We can show that this is also the case for the variants on $\Sigma$.
Its proof is given in \Cref{sec:OWF_OWPuzz}.
\begin{lemma}\label{lem:OWF_OWPuzz}
    Let $\Sigma\subseteq\mathbb{N}$ be a subset.
    If classically-secure OWFs on $\Sigma$ exist, then classically-secure OWPuzzs on $\Sigma$ 
    with $1$-correctness and $\negl$-security exist.
\end{lemma}

\subsection{Inefficient-Verifier Proofs of Quantumness}
In this subsection, we define inefficient-verifier proofs of quantumness (IV-PoQ) on a subset $\Sigma\subseteq\mathbb{N}$. 
IV-PoQ defined in \cite{C:MorYam24} are special cases when $\Sigma=\mathbb{N}$.

\begin{definition}[Inefficient-Verifier Proofs of Quantumness (IV-PoQ) on $\Sigma$]\label{def:IVPoQ}
    Let $\Sigma\subseteq\mathbb{N}$ be a set.
    An IV-PoQ on $\Sigma$ is a tuple $(\cP,\cV_1,\cV_2)$ of interactive algorithms. 
    $\cP$ (prover) is QPT, $\cV_1$ (first verifier) is PPT, and $\cV_2$ (second verifier) is unbounded.
    The protocol is divided into two phases.
    In the first phase, $\cP$ and $\cV_1$ take the security parameter $1^\secp$ as input and interact with each other over a classical channel.
    Let $\tau$ be the transcript, i.e., the sequence of all classical messages exchanged between $\cP$ and $\cV_1$.
    In the second phase, $\cV_2$ takes $1^\secp$ and $\tau$ as input and outputs $\top$ (accept) or $\bot$ (reject).
    We require the following two properties for some functions $c$ and $s$ such that $c(\secp)-s(\secp)\ge 1/\poly(\secp)$.
    \begin{itemize}
    \item
    {\bf $c$-completeness:} 
    There exists $\secp^*\in\mathbb{N}$ such that
        \begin{equation}
            \Pr[\top\gets\cV_2(1^\secp,\tau):\tau\gets\langle\cP,\cV_1\rangle(1^\secp)] \ge c(\secp)
        \end{equation}
    holds for all $\secp\ge\secp^*$.
    \item
    {\bf $s$-soundness on $\Sigma$:} For any PPT prover $\cP^*$ there exists $\secp^{**}\in\mathbb{N}$
    such that
        \begin{equation}
            \Pr[\top\gets\cV_2(1^\secp,\tau):\tau\gets\langle\cP^*,\cV_1\rangle(1^\secp)] \le s(\secp)
        \end{equation}
    holds for all $\secp\ge\secp^{**}$ in $\Sigma$.
    \end{itemize}
Moreover, if all the messages sent from $\cV_1$ are uniformly random strings, we say that the IV-PoQ is public-coin.  
\end{definition}

\begin{remark}
IV-PoQ on $\Sigma$ always exist for any finite set $\Sigma$, 
but we include the case in the definition
for the convenience.
\end{remark}

\begin{remark}
In the previous definition of IV-PoQ~\cite{C:MorYam24}, $\cV_2$ does not take $1^\secp$ as input.
However, this does not change the definition for interactive IV-PoQ, because $1^\secp$ can be added to the first $\cV_1$'s message. We explicitly include $1^\secp$ in the input of $\cV_2$ since we also consider non-interactive IV-PoQ in this paper. 
\end{remark}

\cite{C:MorYam24} showed that classically-secure OWFs imply IV-PoQ by constructing IV-PoQ from statistically-hiding and computationally-binding commitment schemes that are implied by OWFs \cite{HNORV09}.
By inspecting its proof, one can see that the proof gives a ``security-parameter-wise'' reduction, i.e., 
for any efficiently computable polynomial $n$, we can construct IV-PoQ from classically-secure OWFs such that that if the base OWF is secure on inputs of length $n(\secp)$, then the resulting IV-PoQ is sound on the security parameter $\secp$.\footnote{More precisely, 
for any efficiently computable function $f:\bit^*\rightarrow \bit^*$ and efficiently computable polynomial $n$, there is an IV-PoQ $(\cP,\cV_1,\cV_2)$  such that 
for any PPT algorithm $\cP^*$ and any polynomial $p$, there are a PPT algorithm $\cA$ and a polynomial $q$ such that for any $\secp\in \mathbb{N}$, if
\begin{align}
\Pr[\top\gets\cV_2(1^\secp,\tau):\tau\gets\langle\cP^*,\cV_1\rangle(1^\secp)]> \frac{1}{p(\secp)},
\end{align}
then
\begin{align}
\Pr[f(x')=f(x): x\gets\bit^{n(\secp)}, x'\gets\cA(1^{n(\secp)},f(x))]> \frac{1}{q(\secp)}.
\end{align}
} Thus, we have the following lemma. 
\begin{lemma}[Based on \cite{HNORV09,C:MorYam24}]\label{lem:OWF_IVPoQ}
    Let $\Sigma\subseteq\mathbb{N}$ be a set.
    If classically-secure OWFs on $\Sigma$ exist,
    then IV-PoQ on $\Sigma$ exist.
    Moreover, the constructed IV-PoQ is public-coin and satisfies $(1-\negl)$-completeness and $\negl$-soundness on $\Sigma$. 
\end{lemma}
\begin{remark} 
In \cite{C:MorYam24}, they do not explicitly state that the protocol is public-coin. To see that it is indeed public-coin,  
    observe that the verifier's messages of the IV-PoQ of \cite{C:MorYam24} consist of 
the receiver's messages of a statistically hiding commitment scheme of \cite{HNORV09}, descriptions of pairwise independent hash functions, and uniformly random strings from the verifier of \cite{NatPhys:KMCVY22}. As mentioned in \cite[Section 8]{HNORV09}, their commitment scheme is public-coin. Moreover, we can assume that a description of a pairwise independent hash function is public-coin without loss of generality since we can treat the randomness for choosing the function as its description. Thus, the IV-PoQ of \cite{C:MorYam24} is public-coin.
\end{remark}

\subsection{Sampling Complexity}

\begin{definition}[Sampling Problems~\cite{Aar14,ITCS:ABK24}]
\label{def:Samplingproblems}
A (polynomially-bounded) sampling problem $S$ is a collection of probability distributions $\{D_x\}_{x\in\bit^*}$, 
where $D_x$ is a distribution over $\bit^{p(|x|)}$, for some fixed polynomial $p$.
\end{definition}

\begin{definition}[{\bf SampBPP} and {\bf SampBQP}~\cite{Aar14,ITCS:ABK24}] 
\label{def:SampBQP}
{\bf SampBPP} is the class of (polynomially-bounded) sampling problems 
$S=\{D_x\}_{x\in\bit^*}$ for which there exists a PPT algorithm $\cB$ such that for all $x$ and all $\epsilon>0$, 
$\SD(\cB(x,1^{\lfloor 1/\epsilon \rfloor}),D_x) \le\epsilon$, 
where $\cB(x,1^{\lfloor 1/\epsilon\rfloor})$ is the output probability distribution of 
$\cB$ on input $(x, 1^{\lfloor 1/\epsilon\rfloor})$. 
{\bf SampBQP} is defined
the same way, except that $\cB$ is a QPT algorithm rather than a PPT one.
\end{definition}

\subsection{Kolmogorov Complexity}
\label{sec:Kolmogorov} 
In this subsection,
we introduce the definition of time-bounded prefix Kolmogorov complexity and prove lemmas we use in the proof of \Cref{lem:QAS_IVPoQ}.


Time-bounded prefix Kolmogorov complexity is defined using a special type of Turing machines called \emph{self-delimiting machines}.\footnote{ 
For the purpose of this paper, their detailed definition is not necessary (for which we refer to \cite[Chapter 3]{LV19}). 
The important point is, as we will explain, that the set of their valid programs is prefex-free.
} 
Let $M$ be a self-delimiting machine, $p$ be a program, and $x$ be a bit string.
We denote $M(p)=x$ (resp. in at most $t$ steps) if
$M$ takes as input a program $p$ in the input tape,
$M$ halts (resp. in at most $t$ steps),
and
the output tape contains $x$.\footnote{Actually, another additional criteria is required: 
when it halts, the head of the input tape is on the rightmost bit of $p$. This in particular means that $M$ never scans the next bit after $p$,
because the input tape of a self-delimiting machine is one-way.} 
For a self-delimiting machine $M$, we say that $p\in\bit^*$ is a valid program for $M$ if there exists $x\in\bit^*$ such that $M(p)=x$.
For any self-delimiting machine $M$, the set of all valid programs $p\in\bit^*$ for $M$ is prefix-free.
Here, a set $S\subseteq\bit^*$ of strings is called prefix-free, if any $x\in S$ is not a prefix of any other $y\in S$.\footnote{$x$ is a prefix of $y$ if there exists a bit string $z$ such that
$y=x\|z$. For example, if $x=101$ and $y=101111$, then $x$ is a prefix of $y$.}
A function $t:\mathbb{N}\to\mathbb{N}$ is called time-constructible if $t(n)\ge n$ for all $n\in\mathbb{N}$ and there exists a Turing machine that, on input $1^n$, halts and outputs the binary representation of $t(n)$ in time $O(t(n))$ for all $n\in\mathbb{N}$.

\begin{definition}[Time-Bounded Prefix Kolmogorov Complexity \cite{JL00,LV19}]
    Let $M$ be a self-delimiting machine.
    Let $t:\mathbb{N}\to\mathbb{N}$ be a time-constructible function and $x\in\bit^*$ be a bit string.
    Let 
    \begin{align}
    G^t_M(x):=\{p\in\bit^* : M(p)=x \text{ in at most $t(|x|)$ steps}\}.
    \end{align}
    The $t$-time bounded prefix Kolmogorov complexity $K^t_M(x)$ is defined as
    \begin{align}
    \label{def_K}
        K^t_M(x) := 
        \begin{cases}
            \min\{|p| : p\in G^t_M(x)\} & \text{If $G^t_M(x)$ is not empty.} \\
            \infty & \text{If $G^t_M(x)$ is empty.}
        \end{cases}
    \end{align}
\end{definition}
\begin{remark}
There is another variant called time-bounded \emph{plain} Kolmogorov complexity where $M$ is a normal Turing machine rather than a self-delimiting one. Though the definition of the plain version is simpler than that of the above prefix version, we choose to use the prefix version since it is unclear if \Cref{lem:Kolm_coding,lem:Kolm_prob} hold for the plain version (though we believe that we can prove variants of \Cref{lem:Kolm_coding,lem:Kolm_prob} for the plain version with some additional error terms, which would be also sufficient for the purpose of this paper). 
\end{remark}

The following lemma guarantees the existence of a special self-delimiting machine $U$.
We call it \emph{the universal self-delimiting machine}.
\begin{lemma}[Example 7.1.1 of \cite{LV19}]\label{lem:Kolm_optimal}
    There exists a self-delimiting machine $U$ such that for any self-delimiting machine $M$, there exists a constant $c>0$ such that
    \begin{align}
        K^{ct\log t}_U(x) \le K^t_M(x) +c
    \end{align}
    for all $x\in\bit^*$ and any time-constructible function $t:\mathbb{N}\to\mathbb{N}$.
\end{lemma}

We will use the following lemma to show \Cref{lem:QAS_IVPoQ}.\footnote{The lemma was shown in \cite{LV19} for the prefix Kolmogorov complexity.
Here we modify its proof for time-bounded prefix Kolmogorov complexity.}
\begin{lemma}\label{lem:Kolm_prob}
   Let $m:\mathbb{N}\to\mathbb{N}$ be a function such that $m(\secp)\ge\secp$ for all $\secp\in\mathbb{N}$. 
   For each $\secp\in\mathbb{N}$, let $\cD_\secp$ be a distribution over $m(\secp)$ bits. 
    For any $k>0$, any time-constructible function $t:\mathbb{N}\to\mathbb{N}$ such that $t(n)=\Omega(n^2)$,
    and for all sufficiently large $\secp\in\mathbb{N}$,
    \begin{align}
        \Pr_{x\gets\cD_\secp} \left[K^t_U(x)\ge \log\frac{1}{\Pr[x\gets \cD_\secp]}-k\right] \ge 1-\frac{1}{2^k},
    \end{align}
    where $U$ is the universal self-delimiting machine.
\end{lemma}
\begin{proof}[Proof of \cref{lem:Kolm_prob}]
Let $m:\mathbb{N}\to\mathbb{N}$ be a function such that $m(\secp)\ge\secp$ for all $\secp\in\mathbb{N}$. 
For each $\secp\in\mathbb{N}$, let $\cD_\secp$ be a distribution over $m(\secp)$ bits. 
Let $t:\mathbb{N}\to\mathbb{N}$ be a time-constructible function such that $t(n)=\Omega(n^2)$.
We will show later that for all sufficiently large $\secp\in\mathbb{N}$ and for all $x\in\bit^{m(\secp)}$, 
there exists a program $P_x$ such that $|P_x|=K^t_U(x)$ and $U(P_x)=x$ in at most $t(|x|)$ steps, where $U$ is the universal self-delimiting machine.
The set $\{P_x\}_{x\in\bit^{m(\secp)}}$ is prefix-free because $U$ is a self-delimiting machine and
$P_x$ is valid for all $x\in\bit^{m(\secp)}$.
It is known that 
    $\sum_{x\in\bit^{m(\secp)}} 2^{-|P_x|} \le 1$
holds for prefix-free $\{P_x\}_{x\in\bit^{m(\secp)}}$. (It is called Kraft's inequality (Section 1.11 of \cite{LV19}).)
By Markov's inequality, for any $k>0$,
\begin{align}
    \Pr_{x\gets\cD_\secp} \left[ \frac{2^{-K^t_U(x)}}{\Pr[x\gets\cD_\secp]} \ge 2^k \right]
    &\le \frac{1}{2^k} \sum_{x\in\bit^{m(\secp)}} \Pr[x\gets\cD_\secp] \frac{2^{-K^t_U(x)}}{\Pr[x\gets\cD_\secp]} \\ 
    &= \frac{1}{2^k} \sum_{x\in\bit^{m(\secp)}} 2^{-K^t_U(x)} 
    = \frac{1}{2^k} \sum_{x\in\bit^{m(\secp)}} 2^{-|P_x|}
    \le \frac{1}{2^k}.
\end{align}
Therefore, we obtain
\begin{align}
    \Pr_{x\gets\cD_\secp} \left[K^t_U(x)\ge \log\frac{1}{\Pr[x\gets\cD_\secp]}-k\right] \ge 1-\frac{1}{2^k}
\end{align}
for any $k>0$ and for all sufficiently large $\secp\in\mathbb{N}$.

In the remaining part, we show that for all sufficiently large $\secp\in\mathbb{N}$ and for all $x\in\bit^{m(\secp)}$, there exists a program $P_x$ such that $|P_x|=K^t_U(x)$ and $U(P_x)=x$ in at most $t(|x|)$ steps.
To show it, it is sufficient to show $K^t_U(x)<\infty$ for all sufficiently large $\secp\in\mathbb{N}$ and for all $x\in\bit^{m(\secp)}$,
because in that case we have only to take $P_x$ as the program $p$ that achieves the minimum in \cref{def_K}.
In fact, in the following we show that for all sufficiently large $\secp\in\mathbb{N}$ and all $x\in\bit^{m(\secp)}$, $K_U^t(x)=O(m(\secp))$.

Later we will show that there exist a self-delimiting machine $M$ and a time-constructible function $T:\mathbb{N}\to\mathbb{N}$ that satisfy the following:
$T(n)=O(n)$ and for any $x\in\bit^*$, $K^T_M(x)\le O(|x|)$.  
Then by \Cref{lem:Kolm_optimal}, there exists a constant $c>0$ such that
\begin{align}
    K^{cT \log T}_U(x) \le K^T_M(x) +c \le O(|x|) 
\end{align}
for all $x\in\bit^*$.
Because $t(n)=\Omega(n^2)$, $t(n)\ge cT(n)\log T(n)$ for all sufficiently large $n$.
Thus, for all sufficiently large $\ell\in\mathbb{N}$ and for all $x\in\bit^\ell$,
\begin{align}
    K^t_U(x) \le K^{cT\log T}_U(x) \le O(\ell).
\end{align} 
Here, in the first inequality, we have used the fact that $K^\alpha_U(x)\le K^{\beta}_U(x)$ for any time-constructible functions $\alpha$ and $\beta$ such that $\alpha(|x|)\ge\beta(|x|)$.
Because $m(\secp)\ge \secp$, this means that 
    $K^t_U(x)\le O(m(\secp))$
for all sufficiently large $\secp\in\mathbb{N}$ and for all $x\in\bit^{m(\secp)}$.

Finally, we show that there exist a self-delimiting machine $M$ and a time-constructible function $T:\mathbb{N}\to\mathbb{N}$ such that $T(n)=O(n)$ and for all $x\in\bit^*$, $K^T_M(x)\le O(|x|)$. 
We construct $M$ as follows:
\begin{enumerate}
    \item $M$ reads the input tape cell one by one. If the input is in the form of
    $1^\ell\|0\|x\|y$, where $\ell\ge1$, $x\in\bit^*$, $|x|\ge1$, and $y$ is the string of all blanks, then
    go to the next step.\footnote{It is possible to check whether the input is in this form by scanning the input tape in the one-way.}
    Otherwise, $M$ continues reading the input tape, which means that $M$ never halts.
    \item 
    Copy the first $\ell$ bits of $x\|y$ into the output tape.
\end{enumerate}
Consider the case where the bit string $1^{|x|}\|0\|x$ with some $x\in\bit^*$ is written from the leftmost cell of the input tape and all cells to the right of $1^{|x|}\|0\|x$ are blanks.
Then, for all $x\in\bit^*$, $M(1^{|x|}\|0\|x)=x$ in at most $O(|x|)$ steps.
This means that there exists a time-constructible function $T:\mathbb{N}\to\mathbb{N}$ such that $T(n)=O(n)$ and
\begin{align}
    K^T_M(x) \le | 1^{|x|}\|0\|x | = O(|x|)
\end{align}
for all $x\in\bit^*$.
Therefore, we complete the proof.
\end{proof}

We will also use the following lemma which will be used to show \Cref{lem:QAS_IVPoQ}.\footnote{This is essentially shown in \cite{LV19}. 
However in its proof, the $O(\log\secp)$ factor in \Cref{eq:Kolm_coding} is not explicitly given.
Here we reconstruct the proof to obtain the $O(\log\secp)$ factor,
because we need it for our purpose.
}
\begin{lemma}\label{lem:Kolm_coding}
    Let $\cA$ be a PPT algorithm that, on input $1^\secp$, outputs an $m(\secp)$-bit string for all $\secp\in\mathbb{N}$. 
    Here $m:\mathbb{N}\to\mathbb{N}$ is a polynomial such that $m(\secp)\ge\secp$ for all $\secp\in\mathbb{N}$.
    Let $p_{\secp,x}:=\Pr[x\gets\cA(1^\secp)]$ for all $\secp\in\mathbb{N}$ and for all $x\in\bit^{m(\secp)}$.
    Then, for any time-constructible function $t:\mathbb{N}\to\mathbb{N}$ such that $t(n)=2^{\omega(\poly(n))}$, 
    \begin{align}\label{eq:Kolm_coding}
        K^t_U(x) \le \log\frac{1}{p_{\secp,x}} + O(\log \secp)
    \end{align}
    for all sufficiently large $\secp\in\mathbb{N}$ and for all $x\in\Supp(\cA(1^\secp))$.
    Here, $\Supp(\cA(1^\secp))$ is the support of the distribution $\cA(1^\secp)$ over the outputs of $\cA(1^\secp)$,
    and $U$ is the universal self-delimiting machine.
\end{lemma}
\begin{proof}[Proof of \Cref{lem:Kolm_coding}]
    For all $\secp\in\mathbb{N}$ and for all $x\in\Supp(\cA(1^\secp))$, let $c(x)$ be the code word of $x$ obtained by Shannon-Fano coding (Section 1.11 of \cite{LV19}).
    By the property of Shannon-Fano coding, the set $\{c(x)\}_{x\in\Supp(\cA(1^\secp))}$ is prefix-free for all $\secp\in\mathbb{N}$, 
    and 
    \begin{align}\label{eq:ShannonFano_bound}
        \log \frac{1}{p_{\secp,x}} < |c(x)| \le \log\frac{1}{p_{\secp,x}}+1
    \end{align}
    is satisfied for all $\secp\in\mathbb{N}$ and for all $x\in\Supp(\cA(1^\secp))$.
 
    Later we will show that there exist a self-delimiting machine $M$ and a time-constructible function $T$ that satisfy the following:
    $T(n)=2^{O(n^a)}$ for some constant $a\ge 1$ and for all $\secp\in\mathbb{N}$ and for all $x\in\Supp(\cA(1^\secp))$, 
    \begin{align}
        K^T_M(x) \le |c(x)| + O(\log\secp).
    \end{align}
    By \Cref{eq:ShannonFano_bound}, we have
    \begin{align}
        K^T_M(x) \le \log\frac{1}{p_{\secp,x}} + O(\log\secp).
    \end{align}
    By \Cref{lem:Kolm_optimal}, there exists a constant $c>0$ such that
    \begin{align}
        K^{cT\log T}_U(x) \le K^T_M(x) +c \le \log\frac{1}{p_{\secp,x}} + O(\log\secp),
    \end{align}
    where $U$ is the self-delimiting machine.

    Let 
    $t:\mathbb{N}\to\mathbb{N}$ 
    be a time-constructible function 
    such that $t(n)=2^{\omega(\poly(n))}$. 
    Then we have $t(n)\ge cT(n)\log T(n)$ for all sufficiently large $n\in\mathbb{N}$.
    Because $m(\secp)\ge\secp$ for all $\secp\in\mathbb{N}$, this means $t(m(\secp))\ge cT(m(\secp))\log T(m(\secp))$ for all sufficiently large $\secp\in\mathbb{N}$.
    Thus, for all sufficiently large $\secp\in\mathbb{N}$ and for all $x\in\Supp(\cA(1^\secp))$,
    \begin{align}
        K^t_U(x) \le K^{cT\log T}_U(x) \le \log\frac{1}{p_{\secp,x}} + O(\log\secp).
    \end{align}
    Here, in the first inequality, we have used the fact that $K^\alpha_U(x)\le K^{\beta}_U(x)$ for any time-constructible functions $\alpha$ and $\beta$ such that $\alpha(|x|)\ge\beta(|x|)$.
    
    In the remaining part, we show that there exist a self-delimiting machine $M$ and a time-constructible function $T:\mathbb{N}\to\mathbb{N}$ such that $T(n)=2^{O(n^a)}$ for some constant $a\ge 1$ and $K^T_M(x)\le|c(x)|+O(\log\secp)$ for all $\secp\in\mathbb{N}$ and for all $x\in\Supp(\cA(1^\secp))$.
    We construct the self-delimiting machine $M$ as follows: 
    \begin{enumerate}
        \item  $M$ reads the input tape cell one by one. If the input is in the form of
        $1^\ell\|0\|\hat{\secp}\|y\|z$, where $\ell\ge 1$, $\hat{\secp}\in\bit^\ell$, $y\in\bit^*$, $|y|\ge1$, and $z$ is the string of all blanks, then
        go to the next step.\footnote{It is possible to check whether the input is in this form by scanning the input tape by the $(2\ell+2)$-th cell in the one-way.}
        Otherwise, $M$ continues reading the input tape, which means that $M$ never halts.
        \item Let $\secp\in\mathbb{N}$ be the integer whose binary representation is $\hat{\secp}\in\bit^\ell$.
        \item By the brute-force computation, $M$ computes the binary representation of $p_{\secp,x}=\Pr[x\gets\cA(1^\secp)]$ for all $x\in\bit^{m(\secp)}$ and stores them in the work tape.
        \item For all $x\in\Supp(\cA(1^\secp))$, $M$ computes $c(x)$ by using Shannon-Fano coding and stores the list $\{c(x)\}_{x\in\Supp(\cA(1^\secp))}$ in the work tape.
        \item $M$ reads the input tape from the $(2\ell+2)$-th cell one by one. If there exists $x'\in\Supp(\cA(1^\secp))$ such that $y=c(x')\|z'$, where $z'=\{0,1,\square\}$, then go to the next step.
        Otherwise, $M$ continues reading the input tape, which means that $M$ never halts.
        \item $M$ halts and outputs $x'$.
    \end{enumerate}

    We evaluate the running time of $M$.
    Let $r:\mathbb{N}\to\mathbb{N}$ be the running time of $\cA$.
    Because $\cA$ is a PPT algorithm, $r$ is a polynomial and therefore there exists a constant $a\ge 1$ such that $r(n)=O(n^a)$.
    For all $\secp\in\mathbb{N}$ and for all $x\in\bit^{m(\secp)}$, the binary expansion of $p_x$ is computed in time $2^{O(r(\secp))}=2^{O(\secp^a)}$ by the brute-force computation because the running time of $\cA(1^\secp)$ is $r(\secp)$ and the number of random bits used is at most $O(r(\secp))$.
    Thus for all $\secp\in\mathbb{N}$, $M$ can obtain $\{p_x\}_{x\in\bit^{m(\secp)}}$ in time $2^{O(\secp^a+m(\secp))}$.
    Moreover in step 4, $M$ obtains $\{c(x)\}_{x\in\Supp(\cA(1^\secp))}$ as follows:
    \begin{enumerate}
        \item[i.] Sort $\{p_{\secp,x}\}_{x\in\Supp(\cA(1^\secp))}$ in the descending order from left to right. For each $x\in\Supp(\cA(1^\secp))$, let $N_x$ be the integer such that $p_{\secp,x}$ appears in the $N_x$-th position of the sorted list. 
        \item[ii.] Let $P_1:=0$. For all $2\le i \le |\Supp(\cA(1^\secp))|$, let
        \begin{align}
            P_i:=\sum_{x'\in\Supp(\cA(1^\secp)):N_{x'}< i}p_{\secp,{x'}}.
        \end{align}
        \item [iii.] For all $x\in\Supp(\cA(1^\secp))$, the code word $c(x)$ of $x$ is obtained by truncating the binary expansion of $P_{N_x}$ such that its length $|c(x)|$ satisfies
        \begin{align}
            \log\frac{1}{p_x} < |c(x)| \le \log\frac{1}{p_x} + 1.
        \end{align}
    \end{enumerate}
    Then, $M$ can obtain the list $\{c(x)\}_{x\in\Supp(\cA(1^\secp))}$ in time $2^{O(m(\secp))}$. 
    Totally in step 3 and 4, $M$ runs in time $2^{O(\secp^a+m(\secp))}$.
    Because $m(\secp)\ge\secp$ and $a\ge 1$, $\secp^a+m(\secp)=O(m(\secp)^a+m(\secp))=O(m(\secp)^a)$.
    This means $M$ runs in time $2^{O(m(\secp)^a)}$ in step 3 and 4.

    Consider the case where the bit string $1^{|\hat{\secp}|}\|0\|\hat{\secp}\|c(x)$ with some $\secp\in\mathbb{N}$ and some $x\in\Supp(\cA(1^\secp))$ is written from the leftmost cell of the input tape.
    Let all cells to the right of $1^{|\hat{\secp}|}\|0\|\hat{\secp}\|c(x)$ be blanks.
    Then, for all $\secp\in\mathbb{N}$ and for all $x\in\Supp(\cA(1^\secp))$, $M(1^{|\hat{\secp}|}\|0\|\hat{\secp}\|c(x))=x$ in at most $2^{O(|x|^a)}$ steps.
    This means that there exists a time-constructible function $T:\mathbb{N}\to\mathbb{N}$ such that $T(n)=2^{O(n^a)}$ for some constant $a\ge 1$ and for all $\secp\in\mathbb{N}$ and for all $x\in\Supp(\cA(1^\secp))$,
    \begin{align}
        K^T_M(x) \le | 1^{|\hat{\secp}|}\|0\|\hat{\secp}\|c(x) | \le |c(x)| + 2 \lceil \log\secp\rceil +1 = |c(x)| + O(\log\secp).
    \end{align}
    We complete the proof.

\end{proof}

\section{QASs and Int-QASs}\label{sec:quantum_sampling}
In this section, we introduce two new concepts, quantum advantage samplers (QASs) and
interactive quantum advantage samplers (Int-QASs).
We also show some results on them.

\subsection{Definitions of QASs and Int-QASs}
We first define QASs on a set $\Sigma\subseteq\mathbb{N}$.
If $\Sigma=\mathbb{N}$, we call them just QASs.


\begin{definition}[Quantum Advantage Samplers (QASs) on $\Sigma$]
    Let $\Sigma\subseteq\mathbb{N}$ be a set.
    Let $\cA$ be a QPT algorithm that takes $1^\secp$ as input and outputs a classical string.
    $\cA$ is a quantum advantage sampler (QAS) on $\Sigma$ 
    if the following is satisfied: There exists a polynomial $p$ such that 
    for any PPT algorithm $\cB$ (that takes $1^\secp$ as input and outputs a classical string) 
    there exists $\secp^*\in\mathbb{N}$ such that
    \begin{equation}
        \SD(\cA(1^\secp),\cB(1^\secp)) > \frac{1}{p(\secp)}
    \end{equation}
    holds
    for all $\secp\ge\secp^*$ in $\Sigma$.
\end{definition}

\begin{remark}
For any finite set $\Sigma$, QASs on $\Sigma$ always exist, 
but 
we include the case in the definition for the convenience.
\end{remark}

We also define interactive versions of QASs, which we call Int-QASs, as follows.
\begin{definition}[Interactive Quantum Advantage Samplers (Int-QASs)]\label{def:IntQAS}
    Let $(\cA,\cC)$ be a tuple of two interactive QPT algorithms $\cA$ and $\cC$ that communicate over a classical channel. 
    $(\cA,\cC)$ is an interactive quantum advantage sampler (Int-QAS) if
    the following is satisfied:
    There exists a polynomial $p$ such that for any PPT algorithm $\cB$ that interacts with $\cC$,
    \begin{equation}     
    \SD(\langle\cA,\cC\rangle(1^\secp),\langle\cB,\cC\rangle(1^\secp)) > \frac{1}{p(\secp)}
    \end{equation}
    holds for all sufficiently large $\secp\in\mathbb{N}$.
    Here, $\langle\cA,\cC\rangle(1^\secp)$ (resp. $\langle\cB,\cC\rangle(1^\secp)$) is the probability distribution over the transcript of the interaction between $\cA$ (resp. $\cB$) and $\cC$.
\end{definition}

\subsection{Relation Between QASs and Sampling Complexity Classes}
We can show that the existence of QASs implies $\mathbf{SampBPP}\neq\mathbf{SampBQP}$. 
\begin{lemma}\label{lem:QAS_to_SampAd}
Let $\Sigma\subseteq\mathbb{N}$ be an infinite subset.
If QASs on $\Sigma$ exist, then $\mathbf{SampBPP}\neq\mathbf{SampBQP}$. 
\end{lemma}
\begin{proof}[Proof of \cref{lem:QAS_to_SampAd}]
    Let $\cA$ be a QAS on an infinite subset $\Sigma\subseteq\mathbb{N}$ and let $\cA(1^\secp)$ be the output distribution of $\cA$ on input $1^\secp$.
    Consider the collection $\{\cA(1^\secp)\}_{\secp\in\Sigma}$ of distributions.
   Then clearly $\{\cA(1^\secp)\}_{\secp\in\Sigma}\in\mathbf{SampBQP}$.
     What we need to show is that $\{\cA(1^\secp)\}_{\secp\in\Sigma}\notin\mathbf{SampBPP}$.
     Because $\cA$ is a QAS on $\Sigma$, 
    there exists a polynomial $p$
    such that for any PPT algorithm $\cB$, 
    there exists $\secp^*\in\mathbb{N}$ such that
    $\SD(\cA(1^\secp),\cB(1^\secp))>\frac{1}{p(\secp)}$
    for all $\secp\ge\secp^*$ in $\Sigma$.
    Because $\Sigma$ is an infinite set, there exists $x\in\Sigma$
    such that $x\ge\secp^*$.
    Then 
    for any PPT algorithm $\cB$, there exists $x$ and $1/p(x)$
    such that
    $\SD(\cA(1^x),\cB(1^x))>\frac{1}{p(x)}$,
    which means that $\{\cA(1^\secp)\}_{\secp\in\Sigma}\not\in\mathbf{SampBPP}$.
\end{proof}

\begin{remark}
Note that the other direction, namely, $\mathbf{SampBPP}\neq\mathbf{SampBQP}$ implies the existence of QASs, does not seem to hold, because of the following reason:
Assume that $\mathbf{SampBPP}\neq\mathbf{SampBQP}$. Then there exists a
sampling problem $\{D_x\}_x$ that is in $\mathbf{SampBQP}$ but not in $\mathbf{SampBPP}$.
The fact that $\{D_x\}_x\not\in\mathbf{SampBPP}$ means that
for any PPT algorithm $\cB$, there exist $x$ and $\epsilon>0$ such that
\begin{align}
\label{mitasu}
\SD(D_x,\cB(x,1^{\lfloor 1/\epsilon\rfloor}))>\epsilon.
\end{align}
This does not necessarily mean that a QPT algorithm $\cA$ that samples $\{D_x\}_x$
is a QAS. For example, 
$\epsilon$ in \cref{mitasu} could be $2^{-|x|}$.
\end{remark}

\subsection{Non-Interactive IV-PoQ Imply QASs}
In this subsection, we show the following.

\begin{lemma}\label{lem:non-int-IV-PoQ_QAS}
Let $\Sigma\subseteq\mathbb{N}$ be an infinite subset.
    If non-interactive IV-PoQ on $\Sigma$ exist, then QASs on $\Sigma$ exist.
\end{lemma}

\begin{proof}[Proof of \Cref{lem:non-int-IV-PoQ_QAS}]
Let $(\cP,\cV_1,\cV_2)$ be a non-interactive IV-PoQ on $\Sigma$ with $c$-completeness and $s$-soundness such that
$c(\secp)-s(\secp)\ge\frac{1}{q(\secp)}$ for some polynomial $q$. 
$\cP$ takes $1^\secp$ as input and outputs $\tau$. $\cV_1$ does nothing. 
$\cV_2$ takes $1^\secp$ and $\tau$ as input and outputs $\top/\bot$.
Then, we can show that $\cP$ is a QAS on $\Sigma$.
For the sake of contradiction, assume that 
$\cP$ is not a QAS on $\Sigma$.
Then, for any polynomial $p$, there exists a PPT algorithm $\cP^*_p$ such that 
\begin{align}
    \SD(\cP(1^\secp),\cP^*_p(1^\secp)) \le \frac{1}{p(\secp)}
\end{align}
holds for infinitely many $\secp\in\Sigma$.
$\cP^*_p$ can break the $s$-soundness of the non-interactive IV-PoQ.
In fact, if we take $p=2q$,
\begin{align}
    \Pr[\top\gets\cV_2(1^\secp,\tau):\tau\gets\cP^*_{2q}(1^\secp)] 
    &\ge \Pr[\top\gets\cV_2(1^\secp,\tau):\tau\gets\cP(1^\secp)] - \SD(\cP(1^\secp),\cP^*_{2q}(1^\secp)) \\ 
    &\ge c(\secp) - \frac{1}{2q(\secp)}\\
    &>s(\secp)
\end{align}
holds for infinitely many $\secp\in\Sigma$, 
which breaks $s$-soundness.
\end{proof}

\subsection{QASs Imply Non-Interactive IV-PoQ}
The goal of this subsection is to show the following.
\begin{lemma}\label{lem:QAS_to_IV-PoQ}
    Let $\Sigma\subseteq\mathbb{N}$ be an infinite subset.
    If QASs on $\Sigma$ exist, then non-interactive IV-PoQ on $\Sigma$ exist.
\end{lemma}

We can directly show \Cref{lem:QAS_to_IV-PoQ}. However, we first show the following lemma,
\cref{lem:QAS_IVPoQ}, because it is more convenient for our later proofs. 
We then show that if $\cA$ in 
\cref{lem:QAS_IVPoQ} is a QAS, then
$(\cP,\cV)$ in \cref{lem:QAS_IVPoQ} is actually a non-interactive IV-PoQ, which shows \cref{lem:QAS_to_IV-PoQ}.
\begin{lemma}\label{lem:QAS_IVPoQ}
    Let $\cA$ be a QPT algorithm that takes $1^\secp$ as input and outputs a classical string.
    For any polynomial $q$, there exists the following non-interactive protocol $(\cP,\cV)$:
    \begin{itemize}
    \item 
    $\cP$ is a QPT algorithm that, on input $1^\secp$, outputs a bit string $\tau$.
    $\cV$ is an unbounded algorithm that, on input $1^\secp$ and a bit string $\tau$, outputs $\top/\bot$.
        \item For all sufficiently large $\secp\in\mathbb{N}$,
        \begin{align}
            \Pr[\top\gets\cV(1^\secp,\tau):\tau\gets\cP(1^\secp)] \ge 1-\frac{1}{\secp^{\log\secp}}.
        \end{align}
        \item If there exist a PPT algorithm $\cP^*$ and an infinite subset $\Lambda\subseteq\mathbb{N}$ such that
        \begin{align}
            \Pr[\top\gets\cV(1^\secp,\tau):\tau\gets\cP^*(1^\secp)] \ge 1-\frac{1}{\secp^{\log\secp}}-\frac{1}{q(\secp)^2}
        \end{align}
        holds for all $\secp\in\Lambda$, then there exists a PPT algorithm $\cB$ such that
        \begin{align}
            \SD(\cA(1^\secp),\cB(1^\secp)) \le \frac{1}{q(\secp)}
        \end{align}
        holds for all sufficiently large $\secp\in\Lambda$.
    \end{itemize}
\end{lemma}

\begin{proof}[Proof of \Cref{lem:QAS_IVPoQ}]
    We use the technique~\cite{Aar14} of constructing a search problem from a sampling problem.
    Let $\cA(1^\secp)\to\bit^{m(\secp)}$ be a QPT algorithm and $p_y:=\Pr[y\leftarrow\cA(1^\secp)]$ be the probability that $\cA(1^\secp)$ outputs $y\in\bit^{m(\secp)}$, where $m$ is a polynomial.
    
    For any polynomial $q$, we construct a non-interactive protocol $(\cP,\cV)$ as follows.
    \begin{enumerate}
        \item  
        If $q(\secp)^4\ge\secp$, set $N(\secp):=q(\secp)^4$. Otherwise, set $N(\secp):=\secp$.
        $\cP$ runs $y_i\gets\cA(1^\secp)$ for each $i\in[N]$, 
        and sends $(y_1,...,y_N)$ to $\cV$.
        \item $\cV$ outputs $\top$ if 
    \begin{align}\label{eq:answer_set}
        K^T_U(y_1,...,y_N) \ge \log\frac{1}{p_{y_1}\cdots p_{y_N}} - \log^2\secp,
     \end{align}
    where $K_U^T$ denotes time-bounded prefix Kolmogorov complexity with time-bound $T(mN):=2^{2^{mN}}$ and the universal self-delimiting machine $U$.
    Otherwise $\cV$ outputs $\bot$.
    Note that given $(y_1,...,y_N)$, whether $(y_1,...,y_N)$ satisfies \cref{eq:answer_set} or not can be checked in an unbounded time, because $K_U^T$ is computable.
    \end{enumerate}

    We first show the second item of the lemma. 
    From \Cref{lem:Kolm_prob},
    \begin{equation}
        \Pr_{(y_1,...,y_N)\gets\cA(1^\secp)^{\otimes N}} \left[ K^T_U(y_1,...,y_N) \ge \log\frac{1}{p_{y_1}\cdots p_{y_N}} - \log^2\secp \right] \ge 1-\frac{1}{\secp^{\log\secp}}
    \end{equation}
    for all sufficiently large $\secp\in\mathbb{N}$,
    where $\cA(1^\secp)^{\otimes N}$ means that $\cA(1^\secp)$ is executed $N$ times.
    Here we have used the facts that $m(\secp)N(\secp)\ge\secp$ for all $\secp\in\mathbb{N}$ and $T(n)=2^{2^n}=\Omega(n^2)$ to apply \Cref{lem:Kolm_prob}.
    Thus,
    \begin{align}
        &\Pr[ \top\gets\cV(1^\secp,y_1,...,y_N) : (y_1,...,y_N)\gets\cP(1^\secp) ] \\
        &= \Pr_{(y_1,...,y_N)\gets\cA(1^\secp)^{\otimes N}} \left[ K^T_U(y_1,...,y_N) \ge \log\frac{1}{p_{y_1}\cdots p_{y_N}} - \log^2\secp \right] \\
        &\ge 1-\frac{1}{\secp^{\log\secp}} 
    \end{align}
    for all sufficiently large $\secp\in\mathbb{N}$. 
    
    We next show the third item of the lemma.
    We assume that there exists a PPT algorithm $\cP^*$ and an infinite subset $\Lambda\subseteq\mathbb{N}$ such that 
    \begin{align}
        \label{assump}
        \Pr[ \top\gets\cV(1^\secp,y_1,...,y_N) : (y_1,...,y_N)\gets\cP^*(1^\secp) ] \ge 1-\frac{1}{\secp^{\log\secp}}-\frac{1}{q(\secp)^2}
    \end{align}
    holds for all $\secp\in\Lambda$.
    From this $\cP^*$, we construct the following PPT algorithm $\cB$.
     \begin{itemize}
        \item $\cB(1^\secp)\to y$:
        \begin{enumerate}
            \item Take $1^\secp$ as input.
            \item Run $(y_1,...,y_N)\gets\cP^*(1^\secp)$.
            \item Sample $i\gets[N]$.
            \item Output $y:=y_i$.
        \end{enumerate}
    \end{itemize}
    Our goal is to show    
    \begin{align}
    \label{eq:N1_4}
        \SD(\cA(1^\secp),\cB(1^\secp)) \le \frac{1}{q(\secp)} 
    \end{align}
    for all sufficiently large $\secp\in\Lambda$.
   We introduce the following projection 
    \begin{align}
        \Pi := \sum_{(y_1,...,y_N): \top\gets\cV(1^\secp,y_1,...,y_N)} |y_1,...,y_N\rangle\langle y_1,...,y_N|    
    \end{align}
    and we define the following quantum states:
\begin{align} 
C &\coloneqq \sum_{y_1,...,y_N} \Pr[(y_1,...,y_N)\gets\cP^*(1^\secp)] |y_1,...,y_N\rangle\langle y_1,...,y_N| \\ 
C'&\coloneqq \frac{\Pi C}{\Tr(\Pi C)} \\ 
C_i&\coloneqq \Tr_{\lnot i} (C) = \sum_{y_1,...,y_i,...,y_N}\Pr[(y_1,...,y_i,...,y_N)\gets\cP^*(1^\secp)] |y_i\rangle\langle y_i|\\ 
C'_i&\coloneqq \Tr_{\lnot i} (C') \\ 
B&\coloneqq \mathbb{E}_{i\gets[N]} C_i \\ 
B'&\coloneqq \mathbb{E}_{i\gets[N]} C'_i \\ 
A&\coloneqq \sum_y \Pr[y\gets\cA(1^\secp)] |y\rangle\langle y|.
\end{align}
    Here, $\Tr_{\lnot i}$ refers to the operation of tracing out all coordinates except for the $i$-th coordinate.
    Then, by the triangle inequality,
    \begin{align}
        \SD(\cA(1^\secp),\cB(1^\secp)) &= \TD(A,B) 
        \le \TD(B,B') + \TD(A,B').    
    \end{align}
    In the following, we bound $\TD(B,B')$ and $\TD(A,B')$, respectively.
    
    First, we can bound $\TD(B,B')$ as follows:
    \begin{align}
        \TD(B,B') &= \TD (\mathbb{E}_i C_i , \mathbb{E}_i C'_i) \\ 
        &\le \mathbb{E}_i \TD(C_i,C'_i) \quad (\text{By the strong convexity of $\TD$.}) \\ 
        &= \mathbb{E}_i \TD (\Tr_{\lnot i}(C),\Tr_{\lnot i}(C')) \\ 
        &\le \mathbb{E}_i \TD(C,C') \\ 
        &= \TD(C,C') \\ 
        &= \frac{1}{2} \left\| C-\frac{\Pi C}{\Tr(\Pi C)} \right\|_1 \\ 
        &\le \frac{1}{2} \|(I-\Pi)C\|_1 + \frac{1}{2} \left\| \Pi C - \frac{\Pi C}{\Tr(\Pi C)} \right\|_1 \quad (\text{By the triangle inequality.}) \\ 
        &= \frac{1}{2} \Tr|(I-\Pi)C| + \frac{1}{2} \left\|(\Tr(\Pi C)-1)\frac{\Pi C}{\Tr(\Pi C)}\right\|_1 \\ 
        &= \frac{1}{2} \Tr|(I-\Pi)C| + \frac{1}{2} |\Tr(\Pi C)-1| \left\|\frac{\Pi C}{\Tr(\Pi C)}\right\|_1 \\ 
        &= \frac{1}{2} \Tr((I-\Pi)C) + \frac{1}{2} (1-\Tr(\Pi C)) \left\|\frac{\Pi C}{\Tr(\Pi C)}\right\|_1 \label{eq:trans_trace} \\
        &= 1-\Tr(\Pi C) \\
        &= 1-\sum_{(y_1,...,y_N):\top\gets\cV(1^\secp,y_1,...,y_N)} \Pr[(y_1,...,y_N)\gets\cP^*(1^\secp)] \\ 
        &\le \frac{1}{\secp^{\log\secp}} + \frac{1}{q(\secp)^2} \quad (\text{By \Cref{assump}}.) \label{eq:TD_B}
    \end{align}
    holds for all $\secp\in\Lambda$.
    To derive \Cref{eq:trans_trace}, we have used the fact that $(I-\Pi)C$ is positive semidefinite.
    
    Next, we evaluate the upper bound of $\TD(A,B')$.
    Let $c'_{i,y_i}$ be the eigenvalue of $C'_i$ with respect to the eigenstate $|y_i\rangle$. 
    Then,
    \begin{align}
        \TD(A,B') &\le \frac{1}{N} \sum_{i=1}^N \TD(C'_i,A) \\ 
        &= \frac{1}{N} \sum_{i=1}^N \SD(\{c'_{i,y_i}\}_{y_i},\{p_y\}_y) \\ 
        &\le \sqrt{ \frac{1}{N} \sum_{i=1}^N \SD(\{c'_{i,y_i}\}_{y_i},\{p_y\}_y)^2 } \quad (\text{By Jensen's inequality.}) \\ 
        &\le \sqrt{ \frac{\ln 2}{2N} \sum_{i=1}^N D_{\text{KL}} (\{c'_{i,y_i}\}_{y_i}\|\{p_y\}_y) }. \quad (\text{By Pinsker's inequality.}) \label{eq:TD_1}
    \end{align}
    Here, $D_{\text{KL}}$ denotes the KL-divergence defined as 
   \begin{align}
    D_{\text{KL}}(\{p_x\}_x\|\{q_x\}_x):=
    \left\{
    \begin{array}{ll}
    \sum_x p_x \log \frac{p_x}{q_x}& \Supp(\{p_x\}_x)\subseteq\Supp(\{q_x\}_x)\\
    \infty&\mbox{otherwise},
    \end{array}
   \right.
   \end{align} 
   where $\Supp$ is the support,
   and Pinsker's inequality is the following one (Chapter 3 of \cite{CK11}):
  \begin{align}
    \SD(\{p_x\}_{x},\{q_x\}_x) \le \sqrt{\frac{\ln 2}{2} D_{\text{KL}} (\{p_x\}_{x}\|\{q_x\}_x)}.
  \end{align} 
   
    Let $Y:=(y_1,...,y_N)$. Let $p_Y\coloneqq \prod_{i=1}^Np_{y_i}$.
    Let $q_Y$ (resp. $q'_Y$) be the eigenvalue of $C$ (resp. $C'$) with respect to the eigenstate $|Y\rangle$. 
    Define $G := \left\{ Y \middle| \top\gets\cV(1^\secp,Y) \wedge \log\frac{q'_Y}{p_Y}>0 \right\}$
    and $G' := \left\{ Y \middle| \top\gets\cV(1^\secp,Y) \right\}$.
    Since $\{c'_{i,y_i}\}_{y_i}$ is a marginal distribution of $\{q'_Y\}_Y$,
    \begin{align}
        \sum_{i=1}^N D_{\text{KL}} (\{c'_{i,y_i}\}_{y_i}\|\{p_y\}_y)
        &\le D_{\text{KL}}(\{q'_Y\}_Y\|\{p_Y\}_Y) \quad (\text{See below.}) \label{eq:KL_marg}\\ 
        &= \sum_Y q'_Y \log \frac{q'_Y}{p_Y} \\
        &= \sum_{Y:\top\gets\cV(1^\secp,Y)} q'_Y \log \frac{q'_Y}{p_Y} \\ 
        &\le \sum_{Y\in G} q'_Y \log \frac{q'_Y}{p_Y} \\ 
        &\le \frac{1}{1-\delta} \sum_{Y\in G} q_Y \left( \log\frac{q_Y}{p_Y} + \log\frac{1}{1-\delta} \right) \quad (\text{See below.}) \label{eq:q_delta} \\
        &\le \frac{1}{1-\delta} \sum_{Y\in G} q_Y \left( \log\frac{q_Y}{p_Y} + 1 \right) \quad \left(\text{If $\delta\le\frac{1}{2}$, then $\log\frac{1}{1-\delta}\le 1$.}\right) \\
        &\le \frac{1}{1-\delta} \left( \max_{Y\in G'}\log\frac{q_Y}{p_Y} + 1 \right) \sum_{Y\in G} q_Y \\
        &\le \frac{1}{1-\delta} \left( \max_{Y\in G'}\log\frac{q_Y}{p_Y} + 1 \right) \label{eq:TD_2}
    \end{align}
    holds for all sufficiently large $\secp\in\Lambda$, where $\delta\coloneqq\frac{1}{\secp^{\log\secp}}+\frac{1}{q(\secp)^2}$, which is less than $1/2$ for sufficiently large $\secp$.
    In \cref{eq:KL_marg}, we have used the following lemma.
\begin{lemma}[\cite{STOC:Rao08}]\label{lem:KL_marg}
    Let $\cR$ be a distribution over $[M]^N$, with the marginal distribution $\cR_i$ on the $i$-th coordinate. Let $\cD$ be a distribution over $[M]$. Then
    \begin{equation}
        \sum_{i=1}^N D_{\mathrm{KL}}(\cR_i\|\cD) \le D_{\mathrm{KL}}(\cR\|\cD^N).
    \end{equation}
\end{lemma}
    In \Cref{eq:q_delta}, we have used the fact that $q_Y'\le\frac{q_Y}{1-\delta}$, which can be shown as follows.
    First, when $Y\not\in G'$, $q_Y'=0$ and the inequality trivially holds.
    Second, when $Y\in G'$,
    \begin{align}
        q_Y &= \langle Y| C |Y\rangle \\ 
        &= \langle Y| \Pi C |Y\rangle \quad (\text{Because~} Y\in G'.) \\
        &= \Tr(\Pi C) \langle Y| C' |Y \rangle \quad (\text{By the definition of~}C'.) \\ 
        &= q'_Y \sum_{Y: \top\gets\cV(1^\secp,Y)} \Pr[Y\gets\cP^*(1^\secp)] \\ 
        &= q'_Y \Pr[\top\gets\cV(1^\secp,Y):(1^\secp,Y)\gets\cP^*(1^\secp)] \\ 
        &\ge q'_Y (1-\delta).
    \end{align}
    Here the last inequality is from \cref{assump}.

    By \Cref{lem:Kolm_coding}, for all sufficiently large $\secp$ and all $Y\in\Supp(\cP^*(1^\secp))$,
    \begin{align}\label{eq:Kolm_upper}
        K^T_U(Y) \le \log\frac{1}{q_Y} + O(\log\secp).
    \end{align}
    Here we have used the facts that $m(\secp)N(\secp)\ge\secp$ for all $\secp\in\mathbb{N}$ and $T(n)=2^{2^n}=2^{\omega(\poly(n))}$ to apply \Cref{lem:Kolm_coding}.
    

    By combining \Cref{eq:Kolm_upper,eq:answer_set}, we obtain
    \begin{equation}
        \log\frac{1}{p_{y_1}\cdots p_{y_N}} \le \log\frac{1}{q_Y} + O(\log\secp) + \log^2\secp
    \end{equation}
    for any $Y\in G'\cap \Supp(\cP^*(1^\secp))$.
    Thus,
    \begin{align}
        \max_{Y\in G'}\log\frac{q_Y}{p_Y} = \max_{Y\in G' \cap \Supp(\cP^*(1^\secp))}\log\frac{q_Y}{p_Y} \le O(\log\secp) + \log^2\secp. \label{eq:TD_3}
    \end{align}
    From \Cref{eq:TD_1,eq:TD_2,eq:TD_3}, 
    \begin{align}
        \TD(A,B')
        &\le \sqrt{ \frac{\ln 2(\log^2\secp+O(\log\secp))}{2(1-\delta)N} } \\ 
        &\le \sqrt{ \frac{\log^2\secp+O(\log\secp)}{N} } \quad \left(\text{If $\delta\le 1-\frac{\ln 2}{2}$, then $\frac{1}{1-\delta}\le \frac{2}{\ln 2}$.}\right) \\ 
        &\le \sqrt{ \frac{2\log^2\secp}{N} } \label{eq:TD_A}
    \end{align}
    holds for all sufficiently large $\secp\in\Lambda$.
    Therefore, from \Cref{eq:TD_B,eq:TD_A}
    \begin{align}
        \SD(\cA(1^\secp),\cB(1^\secp)) 
        \le \frac{1}{\secp^{\log\secp}} + \frac{1}{q(\secp)^2} + \sqrt{\frac{2\log^2\secp}{N}} 
        \le \frac{q(\secp)}{q(\secp)^2}
        = \frac{1}{q(\secp)}
    \end{align}
    holds for all sufficiently large $\secp\in\Lambda$.
\end{proof}

Now we show \Cref{lem:QAS_to_IV-PoQ} by arguing that if $\cA$ 
in \cref{lem:QAS_IVPoQ} is a QAS, then 
the $(\cP,\cV)$ in \cref{lem:QAS_IVPoQ} is actually a non-interactive IV-PoQ.
\begin{proof}[Proof of \Cref{lem:QAS_to_IV-PoQ}]
Let $\Sigma\subseteq \mathbb{N}$ be an infinite subset.
    Let $\cA$ be a QAS on $\Sigma$.
    Then, from the definition of QASs on $\Sigma$, there exists a polynomial $q$ such that for any PPT algorithm $\cB$ there exists $\secp^*\in\mathbb{N}$
    such that
    \begin{align}\label{eq:QAS_condition}
        \SD(\cA(1^\secp),\cB(1^\secp)) > \frac{1}{q(\secp)}
    \end{align}
    holds for all $\secp\ge\secp^*$ in $\Sigma$.
    From \cref{lem:QAS_IVPoQ}, we can construct a non-interactive protocol $(\cP,\cV)$ from the above $\cA$ and $q$. 
    We show that $(\cP,\cV)$ is actually a non-interactive IV-PoQ on $\Sigma$
    with $(1-\frac{1}{\secp^{\log\secp}})$-completeness
    and
    $(1-\frac{1}{\secp^{\log\secp}}-\frac{1}{q(\secp)^2})$-soundness.

    First, the completeness is obtained from the second item of \cref{lem:QAS_IVPoQ}.
    Next, we show the soundness.
    For the sake of contradiction, we assume that the soundness is not satisfied,
    which means that
    there exist a PPT algorithm $\cP^*$ and an infinite subset $\Lambda\subseteq\Sigma$ such that
    \begin{align}
        \Pr[\top\gets\cV(1^\secp,\tau):\tau\gets\cP^*(1^\secp)] > 1-\frac{1}{\secp^{\log\secp}} - \frac{1}{q(\secp)^2}
    \end{align}
    holds for all $\secp\in\Lambda$.
    Then, from the third item of \Cref{lem:QAS_IVPoQ}, there exists a PPT algorithm $\cB^*$ such that
    \begin{align}
        \SD(\cA(1^\secp),\B^*(1^\secp)) \le \frac{1}{q(\secp)}
    \end{align}
    holds for all sufficiently large $\secp\in\Lambda$.
    This contradicts \Cref{eq:QAS_condition} because $\Lambda$ is an infinite subset of $\Sigma$.
    Therefore, the soundness is satisfied.
\end{proof}

\subsection{QASs From Quantum Advantage Assumption}
\begin{definition}[Quantum Advantage Assumption \cite{STOC:AarArk11,BreMonShe16,cryptoeprint:2024/1490}]
\label{def:supremacy}
    We say that quantum advantage assumption holds if the following is satisfied.
    
    \begin{enumerate}
    \item
    There exists a family $\cC=\{\cC_\secp\}_{\secp\in\N}$ of distributions such that
    for each $\secp\in\N$, $\cC_\secp$ is a 
    (uniform) QPT sampleable distribution
   over quantum circuits $C$ that output $\secp$-bit classical bit strings.
 \item
    There exist polynomials $p$ and $\gamma$ such that:
    \begin{enumerate}
        \item 
        For all sufficiently large $\secp\in\N$,
        \begin{align}
            \Pr_{\substack{C\gets\cC_\secp \\ x\gets\bit^\secp}} \left[\Pr[x\gets C] \ge \frac{1}{p(\secp)2^\secp}\right] \ge \frac{1}{\gamma(\secp)}.
        \end{align}
        \item 
        For any oracle $\cO$ satisfying that for all sufficiently large $\secp\in\N$,
        \begin{align}
            \Pr_{\substack{C\gets\cC_\secp \\ x\gets\bit^\secp}} \left[ |\cO(C,x)-\Pr[x\gets C] | \le \frac{\Pr[x\gets C]}{p(\secp)} \right] \ge \frac{1}{\gamma(\secp)}-\frac{1}{p(\secp)},
        \end{align}
        we have that $\mathbf{P}^{\mathbf{\# P}}\subseteq \mathbf{BPP}^\cO$.
    \end{enumerate}
    \end{enumerate}
\end{definition}
\begin{remark}
Traditionally, we required $\mathbf{P}^{\mathbf{\# P}}\subseteq\mathbf{BPP}^\cO$, but instead of it, we could consider
$\mathbf{P}^{\mathbf{\# P}}\subseteq\mathbf{BQP}^\cO$, for example, because
$\mathbf{P}^{\mathbf{\# P}}\subseteq\mathbf{BQP}^{\mathbf{NP}}$ is also unlikely.
\end{remark}

We show that QASs can be derived from the quantum advantage assumption (\cref{def:supremacy}) plus
$\mathbf{P}^{\mathbf{\# P}}\nsubseteq \mathbf{ioBPP}^{\mathbf{NP}}$.
\begin{theorem}
\label{thm:QAAtoQAS}
If the quantum advantage assumption (\cref{def:supremacy}) holds and    
$\mathbf{P}^{\mathbf{\# P}}\nsubseteq \mathbf{ioBPP}^{\mathbf{NP}}$,
then QASs exist.
\end{theorem}

We can directly derive the existence of QASs from the quantum advantage assumption plus 
$\mathbf{P}^{\mathbf{\# P}}\nsubseteq \mathbf{ioBPP}^{\mathbf{NP}}$.
However, we first introduce a useful notion, which we call hardness of quantum probability estimation (QPE),
and show the theorem via hardness of QPE.
\begin{definition}[Hardness of Quantum Probability Estimation (QPE)]\label{def:approximation}
We say that hardness of quantum probability estimation (QPE) holds if the following is satisfied.
\begin{enumerate}
    \item 
    There exists a family $\cD=\{\cD_\secp\}_{\secp\in\mathbb{N}}$ of distributions such that
    for each $\secp\in\N$, $\cD_\secp$ is a (uniform) QPT sampleable distribution
    over classical bit strings.
    \item
    There exists a polynomial $p$ such that for any oracle PPT algorithm $\cA^{\mathbf{NP}}$ and for all sufficiently large $\secp\in\N$,
    \begin{align}
        \Pr_{x\gets\cD_\secp} \left[ |\cA^{\mathbf{NP}}(1^\secp,x)-\Pr[x\gets\cD_\secp]| \le\frac{\Pr[x\gets\cD_\secp]}{p(\secp)} \right] \le 1-\frac{1}{p(\secp)}.
    \end{align}
\end{enumerate}
\end{definition}
\begin{remark}
A similar notion was introduced in \cite{cryptoeprint:2024/1490,cavalarmeta,cryptoeprint:2024/1539}. The main difference here is that
the hardness is for PPT algorithms with $\mathbf{NP}$ oracle.    
\end{remark}

\cref{thm:QAAtoQAS} is shown by combining the following two lemmas.
\begin{lemma}[Based on \cite{cryptoeprint:2024/1490}]\label{lem:native_assumption}
    If the quantum advantage assumption holds and 
    $\mathbf{P}^{\mathbf{\# P}}\nsubseteq \mathbf{ioBPP}^{\mathbf{NP}}$, then
    hardness of QPE holds.
\end{lemma}
    
\begin{lemma}\label{lem:supremacy_QAS}
    If hardness of QPE holds, then QASs exist.
\end{lemma}

The proof of \cref{lem:native_assumption} is almost the same as Theorem 5.1 of \cite{cryptoeprint:2024/1490} and therefore we omit it.

In the proof of \cref{lem:supremacy_QAS}, we use the following lemma.
\begin{lemma}[Approximate Counting \cite{STOC:Stockmeyer83}]
    \label{lem:approx_count}
    For any $f\in\mathbf{\# P}$ and for any polynomial $p$, there exists an oracle PPT algorithm $\cR^{\mathbf{NP}}$ with $\mathbf{NP}$ oracle such that for any $x$,
    \begin{align}
        \Pr \left[ | \cR^{\mathbf{NP}}(f,x) - f(x) | \le \frac{1}{p(|x|)} \right] \ge 1-\negl(|x|).
    \end{align}
\end{lemma}

Now, we are ready to show \cref{lem:supremacy_QAS}.
\begin{proof}[Proof of \cref{lem:supremacy_QAS}]
Let $\cD=\{\cD_\secp\}_{\secp\in\N}$ be a family of distributions that satisfies \cref{def:approximation}.
For each $\secp\in\N$, let $\cD_\secp$ be a QPT samplable distribution over $\bit^{m(\secp)}$, where $m$ is a polynomial. 
Then there exists a polynomial $p$ such that for any oracle PPT algorithm $\cA^{\mathbf{NP}}$ and for all sufficiently large $\secp\in\mathbb{N}$,
\begin{align}\label{eq:native_assumption}
    \Pr_{x\gets\cD_\secp} \left[ |\cA^{\mathbf{NP}}(1^\secp,x)-\Pr[x\gets\cD_\secp]| \le\frac{\Pr[x\gets\cD_\secp]}{p(\secp)} \right] \le 1-\frac{1}{p(\secp)}.
\end{align}
Let us consider the QPT algorithm $\cQ$ that takes $1^\secp$ as input and samples $x\gets\cD_\secp$.
We show that $\cQ$ is a QAS.
For the sake of contradiction, assume that $\cQ$ is not a QAS.
Then, for any polynomial $q$, there exists a PPT algorithm $\cS_q$ such that
\begin{align}
    \SD(\cQ(1^\secp),\cS_q(1^\secp)) \le \frac{1}{q(\secp)}
\end{align}
holds for infinitely many $\secp\in\N$.
By \cref{lem:approx_count}, for any polynomial $p$, there exists an oracle PPT algorithm $\cA^{\mathbf{NP}}_p$ such that for all $x\in\bit^{m(\secp)}$,
\begin{align}
    \Pr \left[  |\cA^{\mathbf{NP}}_p(x)-\Pr[x\gets\cS_q(1^\secp)]| \le \frac{\Pr[x\gets\cS_q(1^\secp)]}{3p(\secp)} \right] \ge 1-\negl(\secp).
\end{align}
In the following, we set $q(\secp):=12p(\secp)^2$.
For each $\secp\in\N$, we define a set $\Good_\secp$ as follows:
\begin{align}
    \Good_\secp := \left\{ x\in\bit^{m(\secp)} : |\Pr[x\gets\cQ(1^\secp)]-\Pr[x\gets\cS_q(1^\secp)]|\le\frac{\Pr[x\gets\cQ(1^\secp)]}{3p(\secp)} \right\}.
\end{align}
Then,
\begin{align}
    \Pr_{x\gets\cQ(1^\secp)} [x\notin\Good_\secp] 
    &= \Pr_{x\gets\cQ(1^\secp)} \left[ |\Pr[x\gets\cQ(1^\secp)]-\Pr[x\gets\cS_q(1^\secp)]|\ge\frac{\Pr[x\gets\cQ(1^\secp)]}{3p(\secp)} \right] \\ 
    &= \Pr_{x\gets\cQ(1^\secp)} \left[ \left|1-\frac{\Pr[x\gets\cS_q(1^\secp)]}{\Pr[x\gets\cQ(1^\secp)]}\right| \ge \frac{1}{3p(\secp)} \right] \\ 
    &\le 3p(\secp) \sum_x \Pr[x\gets\cQ(1^\secp)] \left|1-\frac{\Pr[x\gets\cS_q(1^\secp)]}{\Pr[x\gets\cQ(1^\secp)]}\right| ~(\text{By Markov's inequality.})\\ 
    &= 3p(\secp) \sum_x |\Pr[x\gets\cQ(1^\secp)]-\Pr[x\gets\cS_q(1^\secp)]| \\ 
    &= 6p(\secp) \SD(\cQ(1^\secp),\cS_q(1^\secp)) \\ 
    &\le \frac{6p(\secp)}{q(\secp)}
    = \frac{1}{2p(\secp)}
\end{align}
holds for infinitely many $\secp\in\N$.
Moreover for all sufficiently large $\secp$ and for all $x\in\Good_\secp$,
\begin{align}
    &1-\negl(\secp) \\
    &\le \Pr\left[ |\cA^{\mathbf{NP}}_p(x) - \Pr[x\gets\cS_q(1^\secp)]| \le \frac{\Pr[x\gets\cS_q(1^\secp)]}{3p(\secp)} \right] \\ 
    &\le \Pr\left[ |\cA^{\mathbf{NP}}_p(x) - \Pr[x\gets\cQ(1^\secp)]| \le \frac{\Pr[x\gets\cS_q(1^\secp)]}{3p(\secp)} + |\Pr[x\gets\cQ(1^\secp)]-\Pr[x\gets\cS_q(1^\secp)]| \right] \\ 
    &\le \Pr\left[ |\cA^{\mathbf{NP}}_p(x) - \Pr[x\gets\cQ(1^\secp)]| \le \Pr[x\gets\cQ(1^\secp)]\left\{\frac{1}{3p(\secp)}\left(1+\frac{1}{3p(\secp)}\right) + \frac{1}{3p(\secp)} \right\} \right] \\ 
    &\le \Pr\left[ |\cA^{\mathbf{NP}}_p(x) - \Pr[x\gets\cQ(1^\secp)]| \le \frac{\Pr[x\gets\cQ(1^\secp)]}{p(\secp)} \right].
\end{align}
Therefore, 
\begin{align}
    &\Pr_{x\gets\cQ(1^\secp)} \left[ |\cA^{\mathbf{NP}}_p(x)-\Pr[x\gets\cD_\secp]|\le\frac{\Pr[x\gets\cD_\secp]}{p(\secp)} \right] \\ 
    &\ge \Pr_{x\gets\cQ(1^\secp)} [x\in\Good_\secp] \Pr_{x\gets\cQ(1^\secp)} \left[ |\cA^{\mathbf{NP}}_p(x)-\Pr[x\gets\cQ(1^\secp)]|\le\frac{\Pr[x\gets\cQ(1^\secp)]}{p(\secp)} ~\middle|~ x\in\Good_\secp \right] \\ 
    &\ge \Pr_{x\gets\cQ(1^\secp)} [x\in\Good_\secp] (1-\negl(\secp)) \\ 
    &\ge \left( 1-\frac{1}{2p(\secp)} \right) (1-\negl(\secp)) \\ 
    &\ge 1-\frac{1}{p(\secp)}
\end{align}
holds for infinitely many $\secp\in\N$.
This contradicts \cref{eq:native_assumption}.

\end{proof}

\section{The QAS/OWF Condition}
\label{sec:QASOWFcondition}
We also introduce another new concept, which we call the QAS/OWF condition. 
\begin{definition}[The QAS/OWF Condition]\label{def:QAS/OWF}
    The QAS/OWF condition holds if there exist a polynomial $p$, a QPT algorithm $\cQ$ that takes $1^\secp$ as input and outputs a classical string, 
    and a function $f:\bit^*\to\bit^*$ that is computable in classical deterministic polynomial-time
    such that for any PPT algorithm $\cS$, the following holds:
    if we define
    \begin{align}
        \Sigma_\cS := \left\{ \secp\in\mathbb{N} \ \middle|\  \SD(\cQ(1^\secp),\cS(1^\secp)) \le \frac{1}{p(\secp)} \right\},
    \end{align}
    then $f$ is a classically-secure OWF on $\Sigma_\cS$.
\end{definition}

 

We can show the following result:
\begin{theorem}
\label{thm:QASOWF_OWF_Samp}
If the QAS/OWF condition is satisfied, then quantumly-secure OWFs exist or
$\mathbf{SampBPP}\neq\mathbf{SampBQP}$.
\end{theorem}

\cref{thm:main2} is obtained by combining this theorem and the 
equivalence of IV-PoQ and the QAS/OWF condition, which will be shown in \cref{sec:equivalence}.

\begin{proof}[Proof of \cref{thm:QASOWF_OWF_Samp}]
We first show that the QAS/OWF condition implies the existence of classically-secure OWFs or $\mathbf{SampBPP}\neq\mathbf{SampBQP}$.
Let us assume that the QAS/OWF condition is satisfied.
Then, by the definition of the QAS/OWF condition,
there exist a polynomial $p$, a QPT algorithm $\cQ$, and a function $f$ that is computable in classical deterministic polynomial-time 
such that for any PPT algorithm $\cS$, if we define
\begin{align}
    \Sigma_\cS := \left\{ \secp\in\mathbb{N} \ \middle|\  \SD(\cQ(1^\secp),\cS(1^\secp)) \le \frac{1}{p(\secp)} \right\},
\end{align}
then $f$ is a classically-secure OWF on $\Sigma_\cS$.
We divide the proof into the following two cases:
\paragraph{There exist a PPT algorithm $\cS$ and a finite subset $\Lambda\subseteq\mathbb{N}$ such that $\Sigma_\cS=\mathbb{N}\setminus\Lambda$.}
In this case, from \cref{lem:OWF_on_N}, classically-secure OWFs exist. 
\paragraph{For any PPT algorithm $\cS$ and for any finite subset $\Lambda\subseteq\mathbb{N}$, $\Sigma_\cS\neq\mathbb{N}\setminus\Lambda$.}
In this case,
if we define the sampling problem $\{\cQ(1^\secp)\}_{\secp\in\mathbb{N}}$, 
for any PPT algorithm $\cS$, there exists an $x\in\mathbb{N}\setminus\Sigma_\cS$ such that 
\begin{align}
    \SD(\cQ(1^x),\cS(1^x)) >\frac{1}{p(x)},
\end{align}
which means that $\{\cQ(1^\secp)\}_{\secp\in\mathbb{N}}\notin\mathbf{SampBPP}$.
On the other hand, it is clear that $\{\cQ(1^\secp)\}_{\secp\in\mathbb{N}}\in\mathbf{SampBQP}$ and therefore $\mathbf{SampBPP}\neq\mathbf{SampBQP}$.

We next show that if classically-secure OWFs exist, then quantumly-secure OWFs exist or $\mathbf{SampBPP}\neq\mathbf{SampBQP}$.
Let $f:\bit^*\to\bit^*$ be a classically-secure OWF.
For the sake of contradiction, we assume that quantumly-secure OWFs do not exist and $\mathbf{SampBPP}=\mathbf{SampBQP}$.
Then, there exists a QPT adversary $\cA$ and a polynomial $p$ such that
\begin{align}
    \Pr[f(z)=f(x):x\gets\bit^\secp,z\gets\cA(1^\secp,f(x))] \ge \frac{1}{p(\secp)}
\end{align}
holds for infinitely many $\secp\in\mathbb{N}$.
Let $\cA(1^\secp,f(x))$ be the output distribution of $\cA$ on input $(1^\secp,f(x))$.
Then, $\{\cA(1^\secp,f(x))\}_{\secp,f(x)}\in\mathbf{SampBQP}$ and therefore $\{\cA(1^\secp,f(x))\}_{\secp,f(x)}\in\mathbf{SampBPP}$.
(Remember that we have assumed $\mathbf{SampBQP}=\mathbf{SampBPP}$.)
Thus, by the definition of $\mathbf{SampBPP}$, there exists a PPT algorithm $\cB$ such that $\SD(\cA(1^\secp,f(x)),\cB(1^\secp,f(x)))\le\frac{1}{2p(\secp)}$ for all $\secp$ and all $f(x)$.
Then, for infinitely many $\secp\in\mathbb{N}$,
\begin{align}
    &\Pr[f(z)=f(x):x\gets\bit^\secp,z\gets\cB(1^\secp,f(x))] \\
    &\ge \Pr[f(z)=f(x):x\gets\bit^\secp,z\gets\cA(1^\secp,f(x))] - \SD(\cA(1^\secp,f(x)),\cB(1^\secp,f(x))) \\ 
    &> \frac{1}{p(\secp)} - \frac{1}{2p(\secp)} 
    = \frac{1}{2p(\secp)}.
\end{align}
This means that $f$ is not classically-secure, which is the contradiction.
\end{proof}

We show a lemma that gives a rephrasing of the negation of the QAS/OWF condition. 
Before presenting the lemma, we explain its intuition and motivation.  
A straightforward negation of the QAS/OWF condition gives the following: 
For any polynomial $p$, QPT algorithm $\cQ$, 
    and a polynomial-time computable function $f$, there is a PPT algorithm $\cS$ such that $f$ is not a classically-secure OWF on $\Sigma_\cS$ where  $\Sigma_\cS$ is as defined in \Cref{def:QAS/OWF}. 
    In this statement, $f$ is not allowed to depend on $\cS$. 
    On the other hand, in the proofs of \Cref{thm:Int-QAS_to_QAS/OWF,thm:OWPuzz_to_QAS/OWF}, we need to allow $f$ to depend on $\cS$ due to a technical reason. The following lemma roughly shows that we can change the order of quantifiers of $f$ and $\cS$ in the above statement so that $f$ can depend on $\cS$. Moreover, we require that $f$'s distributional one-wayness (rather than one-wayness) is broken on $\Sigma_\cS$. 
    The formal statement is given below.    
\begin{lemma}\label{lem:quantifier}
    If the QAS/OWF condition is not satisfied, then the following statement is satisfied:
    for any QPT algorthm $\cQ$ that takes $1^\secp$ as input and outputs a classical string and for any polynomial $p$,
    there exists a PPT algorithm $\cS$ such that for any efficiently-computable polynomial $n$ and any family $\{f_\secp:\bit^{n(\secp)}\to\bit^*\}_{\secp\in\mathbb{N}}$ of functions that are computable in classical
    deterministic polynomial-time, there exists a PPT algorithm $\cR$ such that
    \begin{align}
        \SD(\cQ(1^\secp),\cS(1^\secp)) \le \frac{1}{p(\secp)}
    \end{align}
    and 
    \begin{align}
        \SD(\{x,f_\secp(x)\}_{x\gets\bit^{n(\secp)}} , \{\cR(1^{n(\secp)},f_\secp(x)),f_\secp(x)\}_{x\gets\bit^{n(\secp)}}) \le \frac{1}{p(\secp)}
    \end{align}
    hold for infinitely many $\secp\in\mathbb{N}$.
\end{lemma}

\begin{proof}[Proof of \Cref{lem:quantifier}]
For the sake of contradiction, we assume the following:

\begin{screen}
\paragraph{Assumption 1.}
There exist a QPT algorithm $\cQ$ and a polynomial $p$ such that
for any PPT algorithm $\cS$, 
there exist an efficiently-computable polynomial $n$ and a family $\{f_\secp^\cS:\bit^{n(\secp)}\to\bit^*\}_{\secp\in\mathbb{N}}$ of functions that are computable in classical deterministic polynomial-time
such that for any PPT algorithm $\cR$,
\begin{align}\label{eq:quantifier_first}
    \SD(\cQ(1^\secp),\cS(1^\secp)) > \frac{1}{p(\secp)}
\end{align}
or
\begin{align}\label{eq:quantifier_second}
    \SD(\{x,f_\secp^\cS(x)\}_{x\gets\bit^{n(\secp)}},\{\cR(1^{n(\secp)},f_\secp^\cS(x)),f_\secp^\cS(x)\}_{x\gets\bit^{n(\secp)}}) > \frac{1}{p(\secp)}
\end{align}
holds for all sufficiently large $\secp\in\mathbb{N}$.
\end{screen}
It suffices to show that Assumption 1 implies the QAS/OWF condition.
By using the above $\{f_\secp^\cS\}_{\secp\in\mathbb{N}}$, 
we define a function $f_\cS:\bit^*\to\bit^*$ that is computable in classical deterministic polynomial-time as follows:
\begin{itemize}
    \item $f_\cS:\bit^*\to\bit^*$:
    \begin{enumerate}
        \item Take $x\in\bit^\ell$ as input.
        \item 
        Let $\secp^*$ be the maximum $\secp$ such that $n(\secp)\le\ell$. Let $x'$ be the $n(\secp^*)$-bit prefix of $x$.
        \item 
        Output $f_{\secp^*}^\cS(x')$.
    \end{enumerate}
\end{itemize}  
Then, for any $\secp$ that satisfies \Cref{eq:quantifier_second},
\begin{align}
    &\SD(\{x,f_\cS(x)\}_{x\gets\bit^{n(\secp)}},\{\cR(1^{n(\secp)},f_\cS(x)),f_\cS(x)\}_{x\gets\bit^{n(\secp)}}) \\ 
    &= \SD(\{x,f_{\secp}^\cS(x)\}_{x\gets\bit^{n(\secp)}},\{\cR(1^{n(\secp)},f_{\secp}^\cS(x)),f_{\secp}^\cS(x)\}_{x\gets\bit^{n(\secp)}}) 
    > \frac{1}{p(\secp)}.\label{eq:constructing_f_S}
\end{align}
From Assumption 1 and \cref{eq:constructing_f_S}, we obtain the following.

\begin{screen}
    \paragraph{Assumption 1'.}
    There exist a QPT algorithm $\cQ$ and a polynomial $p$ such that
    for any PPT algorithm $\cS$, there exist a function $f_\cS:\bit^*\to\bit^*$ that is computable in classical deterministic polynomial-time
    and an efficiently-computable polynomial $n$ such that for any PPT algorithm $\cR$,
    \begin{align}\label{eq:quantifier_first_2}
        \SD(\cQ(1^\secp),\cS(1^\secp)) > \frac{1}{p(\secp)}
    \end{align}
    or
    \begin{align}\label{eq:quantifier_second_2}
        \SD(\{x,f_\cS(x)\}_{x\gets\bit^{n(\secp)}},\{\cR(1^{n(\secp)},f_\cS(x)),f_\cS(x)\}_{x\gets\bit^{n(\secp)}}) > \frac{1}{p(\secp)}
    \end{align}
    holds for all sufficiently large $\secp\in\mathbb{N}$.
\end{screen}

Assumption 1' implies the following:
\begin{screen}
    \paragraph{Assumption 1''.}
    There exist a QPT algorithm $\cQ$ and a polynomial $p$ such that
    for any PPT algorithm $\cS$, 
    there exists a function $f_\cS:\bit^*\to\bit^*$ that is computable in classical deterministic polynomial-time
    such that the following holds:
    if we let 
    \begin{align}
        \Sigma_\cS := \left\{ \secp\in\mathbb{N} \ \middle|\  \SD(\cQ(1^\secp),\cS(1^\secp)) \le \frac{1}{p(\secp)} \right\},
    \end{align}
    then $f_\cS$ is a classically-secure DistOWF on $\Sigma_\cS$.
\end{screen}

By \Cref{cor:univ_DistOWF}, Assumption 1'' implies the QAS/OWF condition.
\end{proof}

\section{Equivalence of IV-PoQ and Classically-Secure OWPuzzs} \label{sec:equivalence}
Our main result, \cref{thm:main}, that IV-PoQ exist if and only if classically-secure OWPuzzs exist
is obtained by combining the following theorems. (For relations among these theorems, see \cref{fig:proof_strategy}.)
\begin{theorem}\label{lem:IV-PoQ_to_Int-SampQA}
    If IV-PoQ exist, then Int-QASs exist.
\end{theorem}
\begin{theorem}\label{thm:Int-QAS_to_QAS/OWF}
    If Int-QASs exist, then the QAS/OWF condition is satisfied.
\end{theorem}
\begin{theorem}\label{thm:QAS/OWF_to_IV-PoQ}
    If the QAS/OWF condition is satisfied, then IV-PoQ exist.
\end{theorem}
\begin{theorem}\label{thm:QAS/OWF_to_OWPuzz}
    If the QAS/OWF condition is satisfied, then classically-secure OWPuzzs exist.
\end{theorem} 
\begin{theorem}\label{thm:OWPuzz_to_QAS/OWF}
    If classically-secure OWPuzzs exist, then the QAS/OWF condition is satisfied.
\end{theorem}

Proofs of these theorems are given in the following subsections.

\begin{figure}[H]
    \centering
    \begin{tikzpicture}[scale=0.8]
        \node[draw, align=center, outer sep=5pt] (a) at (-5,2) {IV-PoQ \\ (\cref{def:IVPoQ})} ;
        \node[draw, align=center, outer sep=5pt] (b) at (-5,-2) {Int-QASs \\ (\cref{def:IntQAS})} ;
        \node[draw, align=center, outer sep=5pt] (c) at (0,0) {The QAS/OWF condition \\ (\cref{def:QAS/OWF})} ;
        \node[draw, align=center, outer sep=5pt] (d) at (6,0) {Classically-secure \\ OWPuzzs \\ (\cref{def:OWPuzz})} ;
        \draw[-latex] (a) -- node[auto=right] {\cref{lem:IV-PoQ_to_Int-SampQA}} (b) ;
        \draw[-latex] (b) -- node[auto=right] {\cref{thm:Int-QAS_to_QAS/OWF}} (c) ;
        \draw[-latex] (c) -- node[auto=right] {\cref{thm:QAS/OWF_to_IV-PoQ}} (a) ;
        \draw[-latex] (c) to[out=25,in=155] node[above] {\cref{thm:QAS/OWF_to_OWPuzz}} (d) ;
        \draw[-latex] (d) to[out=205,in=335] node[below] {\cref{thm:OWPuzz_to_QAS/OWF}} (c) ;
    \end{tikzpicture}
    \caption{Relations among theorems.}\label{fig:proof_strategy}
\end{figure}

\subsection{Proof of \cref{lem:IV-PoQ_to_Int-SampQA}}
\label{subsec:IV-PoQ_Int-SampQA}
In this subsection, we show \cref{lem:IV-PoQ_to_Int-SampQA}, namely, IV-PoQ $\Rightarrow$ Int-QASs.

\begin{proof}[Proof of \cref{lem:IV-PoQ_to_Int-SampQA}] 
A proof is similar to that of \cref{lem:non-int-IV-PoQ_QAS}.
Let $(\cP,\cV_1,\cV_2)$ be an IV-PoQ with $c$-completeness and $s$-soundness such that $c(\secp)-s(\secp)\ge\frac{1}{q(\secp)}$ for a polynomial $q$.
Let us consider the interactive protocol $(\cP,\cV_1)$ (i.e., the first phase of the IV-PoQ) which takes $1^\secp$ as input and outputs a transcript $\tau$.
We claim that $(\cP,\cV_1)$ is an Int-QAS.
For the sake of contradiction, assume that it is not.
Then for any polynomial $p$, there exists a PPT algorithm $\cP^*_p$ such that
\begin{equation}
    \SD \left( \langle\cP^*_p,\cV_1\rangle(1^\secp),\langle\cP,\cV_1\rangle(1^\secp) \right) \le\frac{1}{p(\secp)}
\end{equation} 
holds for infinitely many $\secp\in\mathbb{N}$.
We can prove that $\cP^*_{2q}$ breaks the $s$-soundness of the IV-PoQ as follows.
\begin{align}
    &\Pr[\top\gets\cV_2(1^\secp,\tau):\tau\gets\langle\cP^*_{2q},\cV_1\rangle(1^\secp)] \\
    &\ge \Pr[\top\gets\cV_2(1^\secp,\tau):\tau\gets\langle\cP,\cV_1\rangle(1^\secp)] - \SD \left( \langle\cP^*_{2q},\cV_1\rangle(1^\secp),\langle\cP,\cV_1\rangle(1^\secp) \right) \\
    &\ge c(\secp) - \frac{1}{2q(\secp)}\\
    &>s(\secp)
\end{align}
holds for infinitely many $\secp\in\mathbb{N}$, which breaks the $s$-soundness of the IV-PoQ. 
\end{proof}

\subsection{Proof of \cref{thm:Int-QAS_to_QAS/OWF}}
\label{subsec:Int-SampQA_OWF-SampQA}
In this subsection, we show \cref{thm:Int-QAS_to_QAS/OWF}, namely,
$\mbox{Int-QASs}\Rightarrow\mbox{the QAS/OWF condition}$.

\begin{proof}[Proof of \cref{thm:Int-QAS_to_QAS/OWF}]
Let $(\cA,\cC)$ be an $\ell$-round Int-QAS, where $\ell$ is a polynomial.
   Without loss of generality, we can assume that in each round, $\cC$ first sends a message to $\cA$, and $\cA$ returns a message to $\cC$. 
   This is always possible by adding some dummy messages.
   Let $c_i$ be the $i$-th message from $\cC$ to $\cA$, and $a_i$ be the $i$-th message from $\cA$ to $\cC$.
   We denote $(c_1,a_1,...,c_\ell,a_\ell)\gets\langle\cA,\cC\rangle(1^\secp)$ to mean that the transcript $(c_1,a_1,...,c_\ell,a_\ell)$ is obtained by executing the interactive protocol 
   $(\cA,\cC)$ on input $1^\secp$. 
   For the notational simplicity, we denote $\tau_k:=(c_1,a_1,...,c_k,a_k)$ for $k\in [\ell]$.  
   
    For the sake of contradiction, assume that the QAS/OWF condition is not satisfied.
    Then, by \Cref{lem:quantifier} where we set $\mathcal{Q}=\langle\cA,\cC\rangle$, we obtain the following statement:
    for any polynomial $p$, there exists a PPT algorithm $\cS_p$ such that for any efficiently-computable polynomial $n$
    and any family $\{f_\secp:\bit^{n(\secp)}\to\bit^*\}_{\secp\in\mathbb{N}}$ of functions that are computable in classical deterministic polynomial-time, 
    there exists a PPT algorithm $\cR_p$ such that
    \begin{align}\label{eq:S_p_close}
        \SD(\langle\cA,\cC\rangle(1^\secp),\cS_p(1^\secp)) \le \frac{1}{p(\secp)}
    \end{align}
    and 
    \begin{align}\label{eq:IntQAS_cond2}
        \SD(\{x,f_\secp(x)\}_{x\gets\bit^{n(\secp)}} , \{\cR_p(1^{n(\secp)},f_\secp(x)),f_\secp(x)\}_{x\gets\bit^{n(\secp)}}) \le \frac{1}{p(\secp)}
    \end{align}
    hold for infinitely many $\secp\in\mathbb{N}$. 
    Let $u(\secp)$ be the length of the randomness used by $\cS_p(1^\secp)$.  
    By using $\cS_p$, we define a family $\{f_\secp\}_{\secp\in\mathbb{N}}$ of functions that are computable in classical deterministic polynomial-time as follows:
    \begin{itemize}
        \item $f_{\secp}:\bit^{\lceil\log\ell(\secp)\rceil+u(\secp)}\to\bit^*$: 
        \begin{enumerate}
            \item Take $(k,r)$ as input, where $k\in \bit^{\lceil\log\ell(\secp)\rceil}$ and $r\in\bit^{u(\secp)}$. 
            \item Regard $k$ as an encoding of an integer in $[2^{\lceil\log\ell(\secp)\rceil}]$.\footnote{For example, regard $k$ as a binary encoding of an integer and then add $1$.} We use the same notation $k$ to mean the corresponding integer. If $k\notin [\ell(\secp)]$, output an arbitrary fixed value, say, $0$. 
            \item Run $\cS_p(1^\secp;r)=(c_1,a_1,...,c_\ell,a_\ell)$.
            \item Output $(k,\tau_{k-1},c_k)$. 
       \end{enumerate}
    \end{itemize}
      For this specific $\{f_\secp\}_{\secp\in\mathbb{N}}$, there exists a PPT algorithm $\cR_p$ such that \Cref{eq:S_p_close,eq:IntQAS_cond2} hold for infinitely many $\secp\in \mathbb{N}$. We write $\Lambda\subseteq \mathbb{N}$ to mean the set of such $\secp$.
      Rewriting \Cref{eq:IntQAS_cond2}, we have 
        \begin{align}\label{eq:IntQAS_cond2_rewrite}
        \SD(\{(k,r),f_\secp(k,r)\}_{
       (k,r)\gets \bit^{\lceil\log\ell(\secp)\rceil+u(\secp)}} , \{\cR_p(1^{n(\secp)},f_\secp(k,r)),f_\secp(k,r)\}_{
        (k,r)\gets \bit^{\lceil\log\ell(\secp)\rceil+u(\secp)}}) \le \frac{1}{p(\secp)}
    \end{align}
    for all $\secp\in \Lambda$.   
    For any fixed $k^*\in [\ell(\secp)]$, we have 
          \begin{align}
        &\SD(\{(k,r),f_\secp(k,r)\}_{k,r} , \{\cR_p(1^{n(\secp)},f_\secp(k,r)),f_\secp(k,r)\}_{k,r})\\
&=\frac{1}{2}\sum_{k',r',\tau_{k'-1},c_{k'}}\left|
\begin{array}{ll}
\Pr_{k,r}[((k,r),f_\secp(k,r))=((k',r'),(k',\tau_{k'-1},c_{k'}))]\\
       -
       \Pr_{k,r}[(\cR_p(1^{n(\secp)},f_\secp(k,r)),f_\secp(k,r))=((k',r'),(k',\tau_{k'-1},c_{k'}))]
   \end{array}
       \right|\\
&\ge\frac{1}{2}\sum_{r',\tau_{k^*-1},c_{k^*}}\left|
\begin{array}{ll}
\Pr_{k,r}[((k,r),f_\secp(k,r))=((k^*,r'),(k^*,\tau_{k^*-1},c_{k^*}))]\\
       -
       \Pr_{k,r}[(\cR_p(1^{n(\secp)},f_\secp(k,r)),f_\secp(k,r))=((k^*,r'),(k^*,\tau_{k^*-1},c_{k^*}))]
   \end{array}
       \right|\\
       &=\frac{1}{2}\Pr_k[k=k^*]\sum_{r',\tau_{k^*-1},c_{k^*}}\left|
\begin{array}{ll}
\Pr_{r}[((k^*,r),f_\secp(k^*,r))=((k^*,r'),(k^*,\tau_{k^*-1},c_{k^*}))]\\
       -
       \Pr_{r}[(\cR_p(1^{n(\secp)},f_\secp(k^*,r)),f_\secp(k^*,r))=((k^*,r'),(k^*,\tau_{k^*-1},c_{k^*}))]
   \end{array}
       \right|\\
&=
     \frac{1}{2^{\lceil\log\ell(\secp)\rceil}}
     \SD(\{(k^*,r),f_\secp(k^*,r)\}_{r\gets\bit^{u(\secp)}} , \{\cR_p(1^{n(\secp)},f_\secp(k^*,r)),f_\secp(k^*,r)\}_{r\gets\bit^{u(\secp)}})\\
     &\ge 
     \frac{1}{2\ell(\secp)}
     \SD(\{(k^*,r),f_\secp(k^*,r)\}_{r\gets\bit^{u(\secp)}} , \{\cR_p(1^{n(\secp)},f_\secp(k^*,r)),f_\secp(k^*,r)\}_{r\gets\bit^{u(\secp)}}),
\end{align}
where $(k,r)\gets \bit^{\lceil\log\ell(\secp)\rceil+u(\secp)}$.
Combining the above  and \Cref{eq:IntQAS_cond2_rewrite}, we have 
    \begin{align}
        \SD(\{(k,r),f_\secp(k,r)\}_{r\gets\bit^{u(\secp)}} , \{\cR_p(1^{n(\secp)},f_\secp(k,r)),f_\secp(k,r)\}_{r\gets\bit^{u(\secp)}}) \le \frac{2\ell(\secp)}{p(\secp)}
    \end{align}
    for all $\secp\in \Lambda$
    and all $k\in [\ell(\secp)]$.   
    Equivalently, we have 
    \begin{align}
        \SD(\{(k,r),(k,\tau_{k-1},c_k)\} , \{(k',r'),(k,\tau_{k-1},c_k)\}) \le \frac{2\ell(\secp)}{p(\secp)}
    \end{align}
    for all $\secp\in \Lambda$ 
    where $r\gets \bit^{u(\secp)}$, $\cS_p(1^\secp;r)=(c_1,a_1,...,c_\ell,a_\ell)$, and 
    $(k',r')\gets\cR_p(1^{n(\secp)},(k,\tau_{k-1},c_k))$. 
    By the monotonicity of the statistical distance,\footnote{For any random variables $X$ and $Y$ and any algorithm $\cA$, $\SD(\cA(X),\cA(Y))\le \SD(X,Y)$.
    } the above implies  
   \begin{align} \label{eq:next_message_close_sim}
        \SD(\{a_k,(k,\tau_{k-1},c_k)\} , \{a'_{k'},(k,\tau_{k-1},c_k)\}) \le \frac{2\ell(\secp)}{p(\secp)}
    \end{align}
   for all $\secp\in \Lambda$ 
    where $r\gets \bit^{u(\secp)}$, $\cS_p(1^\secp;r)=(c_1,a_1,...,c_\ell,a_\ell)$,  
    $(k',r')\gets\cR_p(1^{n(\secp)},(k,\tau_{k-1},c_k))$,  and
  $\cS_p(1^\secp;r')=(c'_1,a'_1,...,c'_\ell,a'_\ell)$.
  By \Cref{eq:S_p_close}, the distributions of $\{a_k,(k,\tau_{k-1},c_k)\}$ and $\{a'_{k'},(k,\tau_{k-1},c_k)\}$ in \Cref{eq:next_message_close_sim} change only by at most $\frac{1}{p(\secp)}$ in terms of statistical distance if we generate $(c_1,a_1,...,c_\ell,a_\ell)\gets \langle\cA,\cC\rangle(1^\secp)$ instead of $r\gets \bit^{u(\secp)}$ and $\cS_p(1^\secp;r)=(c_1,a_1,...,c_\ell,a_\ell)$. 
Thus, we have 
   \begin{align} \label{eq:next_message_close}
        \SD(\{a_k,(k,\tau_{k-1},c_k)\} , \{a'_{k'},(k,\tau_{k-1},c_k)\}) \le \frac{2\ell(\secp)+2}{p(\secp)}
    \end{align}
   for all $\secp\in \Lambda$ 
    where $(c_1,a_1,...,c_\ell,a_\ell)\gets \langle\cA,\cC\rangle(1^\secp)$,  
    $(k',r')\gets\cR_p(1^{n(\secp)},(k,\tau_{k-1},c_k))$,  and
  $\cS_p(1^\secp;r')=(c'_1,a'_1,...,c'_\ell,a'_\ell)$.

For any polynomial $q$,  
we construct a PPT algorithm $\cB_q$ 
that satisfies
   \begin{align}\label{eq:B_break_Int_QAS}
       \SD(\langle\cA,\cC\rangle(1^\secp),\langle\cB_q,\cC\rangle(1^\secp)) \le \frac{1}{q(\secp)}
   \end{align}
for all $\secp\in \Lambda$.
Since $\Lambda$ is an infinite set, 
this means that $\cB_q$ breaks the Int-QAS, whih constradicts the assumption. 
Thus, it suffices to prove that \Cref{eq:B_break_Int_QAS} holds for  all $\secp\in \Lambda$.

Below, we give the construction of  $\cB_q$. 
Let $p(\secp):=(2\ell(\secp)+2)\ell(\secp) q(\secp)$. 
In the $k$-th round
    $\cB_q$ interacts with $\cC$  as follows:
   \begin{enumerate}
       \item Receive $c_k$ from $\cC$. 
       \item Run $(k',r')\gets\cR_p(1^{n(\secp)},(k,\tau_{k-1},c_k))$.
       \item Compute $\cS_p(1^\secp;r')=(c'_1,a'_1,...,c'_\ell,a'_\ell)$. 
       \item Return $a_{k}:=a'_{k'}$ to $\cC$.
   \end{enumerate}
For $k\in [\ell(\secp)+1]$, let $D_k$ be the distribution sampled by the following procedure: 
\begin{enumerate}
\item Run $\langle\cA,\cC\rangle(1^\secp)$ until the $(k-1)$-th round where we write $\tau_{k-1}=(c_1,a_1,\ldots,c_{k-1},a_{k-1})$ to mean the partial transcript until the $(k-1)$-th round. 
\item Use $\cB_q$ instead of $\cA$ to complete the protocol. That is, for $i=k,k+1,\ldots,\ell$, do the following:  
\begin{enumerate}
    \item Receive $c_i$ from $\cC$. 
       \item Run $(i',r')\gets\cR_p(1^{n(\secp)},(i,\tau_{i-1},c_i))$.
       \item Compute $\cS_p(1^\secp;r')=(c'_1,a'_1,...,c'_\ell,a'_\ell)$. 
       \item Return $a_{i}:=a'_{i'}$ to $\cC$.
\end{enumerate}
\item Output the full transcript $\tau_\ell=(c_1,a_1,\ldots,c_{\ell},a_{\ell})$. 
\end{enumerate}
Clearly, we have $D_1=\langle\cB_q,\cC\rangle(1^\secp)$ 
and
$D_{\ell+1}=\langle\cA,\cC\rangle(1^\secp)$ 
where the equality means equivalence as distributions. 
Thus, it suffices to prove  
\begin{align}
    \SD(D_1,D_{\ell+1})\le \frac{1}{q(\secp)}
\end{align}
for all $\secp \in \Lambda$. 
By the triangle inequality, it suffices to prove 
\begin{align}\label{eq:hybrd_k_k_plus_1}
    \SD(D_k,D_{k+1})\le \frac{1}{\ell(\secp)q(\secp)}
\end{align}
for all 
$\secp \in \Lambda$
and 
$k\in [\ell(\secp)]$.
Note that the only difference between $D_{k}$ and $D_{k+1}$ is how $a_k$ is generated. 
In $D_k$, $a_k$ is generated by the interaction between $\cA$ and $\cC$. On the other hand, in $D_{k+1}$, $a_k$ is generated by running $(k',r')\gets \cR_p(k,\tau_{k-1},c_k)$ and $\cS_p(1^\secp;r')=(c'_1,a'_1,...,c'_\ell,a'_\ell)$  and then setting $a_k:=a'_{k'}$ where  $(\tau_{k-1},c_k)$ is the partial transcript generated by the interaction between $\cA$ and $\cC$. Thus, by a straightforward reduction to \Cref{eq:next_message_close}, the distributions of $\tau_k$ in $D_{k}$ and $D_{k+1}$ differ by at most $\frac{2\ell(\secp)+2}{p(\secp)}=\frac{1}{\ell(\secp)q(\secp)}$ in terms of statistical distance. 
Moreover, both in $D_{k}$ and $D_{k+1}$, the partial transcript $\tau_k$ is extended to the full transcript $\tau_{\ell}$ in the same manner using the interaction between $\cB_q$ and $\cC$. Thus, by the monotonicity of the statistical distance,  
\Cref{eq:hybrd_k_k_plus_1} holds for all $\secp\in \Lambda$.
This completes the proof of \cref{thm:Int-QAS_to_QAS/OWF}. 
\end{proof}

\subsection{Proof of \cref{thm:QAS/OWF_to_IV-PoQ}}
\label{subsec:OWF-SampQA_IV-PoQ}

In this subsection, we show \cref{thm:QAS/OWF_to_IV-PoQ}, namely, the QAS/OWF condition $\Rightarrow$ IV-PoQ.

\begin{proof}[Proof of \Cref{thm:QAS/OWF_to_IV-PoQ}]
Assume that the QAS/OWF condition is satisfied. Then, 
from the definition of the QAS/OWF condition,
there exist a polynomial $p$, a QPT algorithm $\cQ$ that takes $1^\secp$ as input and outputs a classical string, 
a function $f:\bit^*\to\bit^*$ that is computable in classical deterministic polynomial-time such that
for any PPT algorithm $\cS$, the following holds: Let 
\begin{align}\label{eq:def_Sigma_S}
\Sigma_{\cS}\coloneqq\left\{\secp\in\mathbb{N}:\SD(\cQ(1^\secp),\cS(1^\secp))\le\frac{1}{p(\secp)}\right\},
\end{align}
then $f$ is a classically-secure OWF on $\Sigma_\cS$.
 
By using such $p$ and $\cQ$, from \cref{lem:QAS_IVPoQ}, we obtain a non-interactive protocol $(\cP,\cV)$ that satisfies the
properties of \cref{lem:QAS_IVPoQ}. 
Moreover by using $f$, from \cref{lem:OWF_IVPoQ}, we obtain an IV-PoQ $(\cP',\cV'_1,\cV'_2)$ on $\Sigma_\cS$
with $(1-\negl(\secp))$-completeness and $\negl(\secp)$-soundness.
By combining $(\cP,\cV)$ and $(\cP',\cV'_1,\cV'_2)$,
we construct an IV-PoQ $(\cP'',\cV''_1,\cV''_2)$ with $(1-\negl(\secp))$-completeness and $(1-\frac{1}{\secp^{\log\secp}}-\frac{1}{p(\secp)^2})$-soundness 
as follows:
\begin{itemize}
    \item 1st phase. $(\cP'',\cV''_1)(1^\secp)\to \tau''$:
    \begin{enumerate}
        \item Take $1^\secp$ as input.
        \item Run $\tau\gets\cP(1^\secp)$. $\tau$ is sent to $\cV''_1$.
        \item Run $\tau'\gets\langle\cP',\cV'_1\rangle(1^\secp)$.
        \item Output $\tau'':=(\tau,\tau')$.
    \end{enumerate}
    \item 2nd phase. $\cV''_2(1^\secp,\tau'')\to\top/\bot$:
    \begin{enumerate}
        \item Take $1^\secp$ and $\tau''=(\tau,\tau')$ as input.
        \item Run $\cV(1^\secp,\tau)$ and $\cV'_2(1^\secp,\tau')$.
        \item If both $\cV$ and $\cV'_2$ output $\top$, then output $\top$. Otherwise, output $\bot$.
    \end{enumerate}
\end{itemize} 
First, we prove $(1-\negl)$-completeness of $(\cP'',\cV_1'',\cV_2'')$.
We have
\begin{align}
    \Pr[\top\gets\cV''_2(1^\secp,\tau''):\tau''\gets\langle\cP'',\cV''_1\rangle(1^\secp)] 
    &=\Pr\left[\top\gets\cV(1^\secp,\tau) \land \top\gets\cV'_2(1^\secp,\tau') : \begin{gathered} \tau\gets\cP(1^\secp) \\ \tau'\gets\langle\cP',\cV'_1\rangle(1^\secp)\end{gathered} \right] \\ 
    &=
    \Pr[\top\gets\cV(1^\secp,\tau):\tau\gets\cP(1^\secp)] \\
    &\times\Pr[\top\gets\cV'_2(1^\secp,\tau'):\tau'\gets\langle\cP',\cV'_1\rangle(1^\secp)] \\
    &\ge 1-2\negl(\secp)
\end{align}
for all sufficiently large $\secp\in\mathbb{N}$.
Here, we have used the fact that
$(\cP',\cV_1',\cV_2')$ is $(1-\negl(\secp))$-complete,
and 
the second item of \Cref{lem:QAS_IVPoQ}.

Next, we prove the soundness.
For the sake of contradiction, we assume that $(\cP'',\cV_1'',\cV_2'')$ does not satisfy $(1-\frac{1}{\secp^{\log\secp}}-\frac{1}{p(\secp)^2})$-soundness on $\mathbb{N}$.
Then, there exist a PPT algorithm $\cP^*=(\cP^*_1,\cP^*_2)$ and an infinite subset $\Lambda\subseteq\mathbb{N}$ such that
\begin{align}
    1-\frac{1}{\secp^{\log\secp}}-\frac{1}{p(\secp)^2}
    &<\Pr[\top\gets\cV''_2(1^\secp,\tau''):\tau''\gets\langle\cP^*,\cV''_1\rangle(1^\secp)] \\
    &=\Pr\left[\top\gets\cV(1^\secp,\tau) \land \top\gets\cV'_2(1^\secp,\tau') : 
    \begin{gathered} (\tau,\st)\gets\cP^*_1(1^\secp) \\ \tau'\gets\langle\cP^*_2(\st),\cV'_1(1^\secp)\rangle\end{gathered} \right] 
\end{align}
holds for all $\secp\in\Lambda$.
Define a PPT adversary $\cP^*_3$ against $(\cP',\cV_1',\cV_2')$ as follows.
\begin{enumerate}
    \item 
    Before interacting with $\cV_1'$, run $(\tau,\st)\gets\cP^*_1(1^\secp)$.
    \item 
    Run $\tau'\gets\langle \cP^*_2(\st),\cV_1'(1^\secp)\rangle$.
\end{enumerate}

Then, we have that both of 
\begin{align}
    \Pr[\top\gets\cV(1^\secp,\tau):\tau\gets\cP^*_1(1^\secp)] 
    &= \Pr\left[\top\gets\cV(1^\secp,\tau) : 
    \begin{gathered} (\tau,\st)\gets\cP^*_1(1^\secp) \\ \tau'\gets\langle\cP^*_2(\st),\cV'_1(1^\secp)\rangle\end{gathered} \right] \\ 
    &\ge \Pr\left[\top\gets\cV(1^\secp,\tau) \land \top\gets\cV'_2(1^\secp,\tau') : \begin{gathered} (\tau,\st)\gets\cP^*_1(1^\secp) \\ \tau'\gets\langle\cP^*_2(\st),\cV'_1(1^\secp)\rangle\end{gathered} \right] \\
    &> 1-\frac{1}{\secp^{\log\secp}}-\frac{1}{p(\secp)^2}
    \label{soundness1}
\end{align}
and
\begin{align}
    \Pr[\top\gets\cV'_2(1^\secp,\tau'):\tau'\gets\langle\cP^*_3,\cV'_1\rangle(1^\secp)] 
    &= \Pr\left[ \top\gets\cV'_2(1^\secp,\tau') : 
    \begin{gathered} (\tau,\st)\gets\cP^*_1(1^\secp) \\ \tau'\gets\langle\cP^*_2(\st),\cV'_1(1^\secp)\rangle\end{gathered} \right] \\ 
    &\ge \Pr\left[\top\gets\cV(1^\secp,\tau) \land \top\gets\cV'_2(1^\secp,\tau') : 
    \begin{gathered} (\tau,\st)\gets\cP^*_1(1^\secp) \\ \tau'\gets\langle\cP^*_2(\st),\cV'_1(1^\secp)\rangle\end{gathered} \right] \\
    &> 1-\frac{1}{\secp^{\log\secp}}-\frac{1}{p(\secp)^2}
    \label{soundness2}
\end{align}
hold for all $\secp\in\Lambda$.

By \Cref{lem:QAS_IVPoQ} and \cref{soundness1}, there exists a PPT algorithm $\cS^*$ such that
\begin{align}
    \SD(\cQ(1^\secp),\cS^*(1^\secp)) \le \frac{1}{p(\secp)}
\end{align}
holds for all sufficiently large $\secp\in\Lambda$. 
Let $\Theta_{\cS^*}$ be the set of such $\secp$. 
Define
   \begin{align}
        \Sigma_{\cS^*}\coloneqq\left\{\secp\in\mathbb{N}:\SD(\cQ(1^\secp),\cS^*(1^\secp))\le\frac{1}{p(\secp)}\right\}.
    \end{align}
By the QAS/OWF condition, $f$ is a classically-secure OWF on $\Sigma_{\cS^*}$.
Because $\Theta_{\cS^*}\subseteq\Sigma_{\cS^*}$, this means that $f$ is a classically-secure OWF on $\Theta_{\cS^*}$.
Then $f$ is a classically-secure OWF on $\Lambda$.
From \cref{lem:OWF_IVPoQ},
$(\cP',\cV_1',\cV_2')$ is an IV-PoQ with $(1-\negl(\secp))$-completeness and
$(1-\frac{1}{\secp^{\log\secp}}-\frac{1}{p(\secp)^2})$-soundness on $\Lambda$,
but it contradict \cref{soundness2}. 
\end{proof}

\subsection{Proof of \cref{thm:QAS/OWF_to_OWPuzz}}
\label{subsec:OWF-SampQA_OWPuzzle}
In this subsection, we show \cref{thm:QAS/OWF_to_OWPuzz}, namely, the QAS/OWF condition $\Rightarrow$ classically-secure OWPuzzs.
\begin{proof}[Proof of \cref{thm:QAS/OWF_to_OWPuzz}]
    Assume that the QAS/OWF condition is satisfied. 
    Then, from the definition of the QAS/OWF condition, 
    there exist a polynomial $p$, a QPT algorithm $\cQ$ that takes $1^\secp$ as input and outputs a classical string, 
    a function $f:\bit^*\to\bit^*$ that is computable in classical deterministic polynomial-time such that
    for any PPT algorithm $\cS$, the following holds: Let 
    \begin{align}
        \Sigma_{\cS}\coloneqq\left\{\secp\in\mathbb{N}:\SD(\cQ(1^\secp),\cS(1^\secp))\le\frac{1}{p(\secp)}\right\},
    \end{align}
    then $f$ is a classically-secure OWF on $\Sigma_\cS$.

    By using such $p$ and $\cQ$, we obtain a non-interactive protocol $(\cP,\cV)$ from \Cref{lem:QAS_IVPoQ}.
    By using $(\cP,\cV)$, we construct a pair $(\Samp,\Ver)$ of algorithms as follows:
    \begin{itemize}
        \item $\Samp(1^\secp)\to(\ans,\puzz)$:
        \begin{enumerate}
            \item Run $\tau\gets\cP(1^\secp)$. 
            \item Output $\puzz\coloneqq 1^\secp$ and $\ans\coloneq \tau$.
        \end{enumerate}
        \item $\Ver(\ans^*,\puzz)\to\top/\bot$:
        \begin{enumerate}
            \item Parse $\puzz=1^\secp$ and $\ans^*=\tau^*$.
            \item Run $\cV(1^\secp,\tau^*)$.
            \item Output $\top$ if $\cV$ outputs $\top$. Output $\bot$ otherwise.
        \end{enumerate}
    \end{itemize}
    Moreover by using the classically-secure OWF $f$ on $\Sigma_\cS$, from \cref{lem:OWF_OWPuzz}, we obtain a classically-secure OWPuzz $(\Samp',\Ver')$ on $\Sigma_\cS$ with $1$-correctness and $\negl(\secp)$-security.
    By combining $(\Samp,\Ver)$ and $(\Samp',\Ver')$, we construct a classically-secure OWPuzz $(\Samp'',\Ver'')$ on $\mathbb{N}$ with $(1-\frac{1}{\secp^{\log\secp}})$-correctness and $(1-\frac{1}{\secp^{\log\secp}}-\frac{1}{p(\secp)^2})$-security as follows:
    \begin{itemize}
        \item $\Samp''(1^\secp)\to(\ans'',\puzz'')$:
        \begin{enumerate}
            \item Take $1^\secp$ as input.
            \item Run $(\ans,\puzz)\gets\Samp(1^\secp)$ and $(\ans',\puzz')\gets\Samp'(1^\secp)$.
            \item Set $\ans'':=(\ans,\ans')$ and $\puzz'':=(\puzz,\puzz')$.
            \item Output $(\ans'',\puzz'')$.
        \end{enumerate}
        \item $\Ver''(\ans''^*,\puzz'')\to\top/\bot$:
        \begin{enumerate}
            \item Take $\ans''^*=(\ans^*,\ans'^*)$ and $\puzz''=(\puzz,\puzz')$ as input.
            \item Run $\Ver(\ans^*,\puzz)$ and $\Ver'(\ans'^*,\puzz')$.
            \item If both $\Ver$ and $\Ver'$ output $\top$, then output $\top$. Otherwise, output $\bot$.
        \end{enumerate}
    \end{itemize} 

    First, we prove $(1-\frac{1}{\secp^{\log\secp}})$-correctness of $(\Samp'',\Ver'')$.
    By the second item of \Cref{lem:QAS_IVPoQ},
    \begin{align}
        \Pr[\top\gets\Ver(\ans,\puzz):(\ans,\puzz)\gets\Samp(1^\secp)] 
        &= \Pr[\top\gets\cV(1^\secp,\tau):\tau\gets\cP(1^\secp)] \\ 
        &\ge 1-\frac{1}{\secp^{\log\secp}}
    \end{align}
    holds for all sufficiently large
    $\secp\in\mathbb{N}$.
    Thus,
    by using the fact that $(\Samp',\Ver')$ is $1$-correct,
    we obtain
    \begin{align}
        &\Pr[\top\gets\Ver''(\ans'',\puzz''):(\ans'',\puzz'')\gets\Samp''(1^\secp)] \\ 
        &= \Pr\left[\top\gets\Ver(\ans,\puzz) \land \top\gets\Ver'(\ans',\puzz') : \begin{gathered} (\ans,\puzz)\gets\Samp(1^\secp) \\ (\ans',\puzz')\gets\Samp'(1^\secp)\end{gathered} \right] \\ 
        &= \Pr[\top\gets\Ver(\ans,\puzz):(\ans,\puzz)\gets\Samp(1^\secp)]\\
        &\times \Pr[\top\gets\Ver'(\ans',\puzz'):(\ans',\puzz')\gets\Samp'(1^\secp)] \\
        &\ge 1-\frac{1}{\secp^{\log\secp}}
    \end{align}
    for all sufficiently large
    $\secp\in\mathbb{N}$. 

    Next, we prove the security. 
    For the sake of contradiction, we assume that $(\Samp'',\Ver'')$ does not satisfy $(1-\frac{1}{\secp^{\log\secp}}-\frac{1}{p(\secp)^2})$-security on $\mathbb{N}$.
    Then, there exist a PPT adversary $\cA$ and an infinite subset $\Lambda\subseteq\mathbb{N}$ such that
    \begin{align}
        1-\frac{1}{\secp^{\log\secp}}-\frac{1}{p(\secp)^2}
        &<\Pr[\top\gets\Ver''(\cA(1^\secp,\puzz''),\puzz''):(\ans'',\puzz'')\gets\Samp''(1^\secp)] \\
        &= \Pr\left[
            \begin{array}{c}
                \top\gets\Ver(\ans^*,\puzz) \\ 
                \land \\ 
                \top\gets\Ver'(\ans'^*,\puzz') \\ 
            \end{array} :
            \begin{array}{l}
                (\ans,\puzz)\gets\Samp(1^\secp) \\ 
                (\ans',\puzz')\gets\Samp'(1^\secp) \\ 
                (\ans^*,\ans'^*)\gets\cA(1^\secp,\puzz,\puzz')
            \end{array}
        \right] 
   \end{align}
    holds for all $\secp\in\Lambda$.
    We define the following two PPT adversaries $\cA_1$ and $\cA_2$ against $(\Samp,\Ver)$ and $(\Samp',\Ver')$, respectively:
    \begin{itemize}
        \item $\cA_1(1^\secp,\puzz)\to\ans^*$:
        \begin{enumerate}
            \item Take $1^\secp$ and $\puzz$ as input. 
            \item Run $(\ans',\puzz')\gets\Samp'(1^\secp)$.
            \item Run $(\ans^*,\ans'^*)\gets\cA(1^\secp,\puzz,\puzz')$.
            \item Output $\ans^*$.
        \end{enumerate}
        \item $\cA_2(1^\secp,\puzz')\to\ans'^*$:
        \begin{enumerate}
            \item Take $1^\secp$ and $\puzz'$ as input. 
            \item Run $(\ans,\puzz)\gets\Samp(1^\secp)$.
            \item Run $(\ans^*,\ans'^*)\gets\cA(1^\secp,\puzz,\puzz')$.
            \item Output $\ans'^*$.
        \end{enumerate}
    \end{itemize}
    Then, both of 
    \begin{align}
        &\Pr[\top\gets\Ver(\cA_1(1^\secp,\puzz),\puzz):(\ans,\puzz)\gets\Samp(1^\secp)] \\
        &= \Pr\left[
            \top\gets\Ver(\ans^*,\puzz):
            \begin{array}{l}
                (\ans,\puzz)\gets\Samp(1^\secp) \\ 
                (\ans',\puzz')\gets\Samp'(1^\secp) \\ 
                (\ans^*,\ans'^*)\gets\cA(1^\secp,\puzz,\puzz') 
            \end{array}
        \right] \\ 
        &\ge \Pr\left[
            \begin{array}{c}
                \top\gets\Ver(\ans^*,\puzz) \\ 
                \land \\ 
                \top\gets\Ver'(\ans'^*,\puzz') 
            \end{array} : 
            \begin{array}{l}
                (\ans,\puzz)\gets\Samp(1^\secp) \\ 
                (\ans',\puzz')\gets\Samp'(1^\secp) \\ 
                (\ans^*,\ans'^*)\gets\cA(1^\secp,\puzz,\puzz') 
            \end{array}
        \right]
        > 1-\frac{1}{\secp^{\log\secp}}-\frac{1}{p(\secp)^2} \label{eq:OWPuzz_soundness1}
    \end{align}
    and 
    \begin{align}
        &\Pr[\top\gets\Ver'(\cA_2(1^\secp,\puzz'),\puzz'):(\ans',\puzz')\gets\Samp'(1^\secp)] \\
        &= \Pr\left[
            \top\gets\Ver'(\ans'^*,\puzz'):
            \begin{array}{l}
                (\ans',\puzz')\gets\Samp'(1^\secp) \\ 
                (\ans,\puzz)\gets\Samp(1^\secp) \\ 
                (\ans^*,\ans'^*)\gets\cA(1^\secp,\puzz,\puzz') 
            \end{array}
        \right] \\ 
        &\ge \Pr\left[
            \begin{array}{c}
                \top\gets\Ver(\ans^*,\puzz) \\ 
                \land \\ 
                \top\gets\Ver'(\ans'^*,\puzz') 
            \end{array} : 
            \begin{array}{l}
                (\ans,\puzz)\gets\Samp(1^\secp) \\ 
                (\ans',\puzz')\gets\Samp'(1^\secp) \\ 
                (\ans^*,\ans'^*)\gets\cA(\puzz,\puzz') 
            \end{array}
        \right]
        > 1-\frac{1}{\secp^{\log\secp}}-\frac{1}{p(\secp)^2} \label{eq:OWPuzz_soundness2}
    \end{align}
    hold for all $\secp\in\Lambda$.
    From \Cref{eq:OWPuzz_soundness1}, we obtain
    \begin{align}
        1-\frac{1}{\secp^{\log\secp}} -\frac{1}{p(\secp)^2}
        &< \Pr[\top\gets\Ver(\cA_1(1^\secp,\puzz),\puzz):(\ans,\puzz)\gets\Samp(1^\secp)] \\ 
        &= \Pr[\top\gets\Ver(\cA_1(1^\secp,1^\secp),1^\secp):(\ans,1^\secp)\gets\Samp(1^\secp)] \\ 
        &= \Pr[\top\gets\Ver(\ans^*,1^\secp):\ans^*\gets\cA_1(1^\secp)] \\ 
        &= \Pr[\top\gets\cV(1^\secp,\ans^*):\ans^*\gets\cA_1(1^\secp)] \\ 
        &= \Pr[\top\gets\cV(1^\secp,\tau):\tau\gets\cA_1(1^\secp)] 
    \end{align}
    for all $\secp\in\Lambda$.
    Thus, by the third item of \Cref{lem:QAS_IVPoQ}, there exists a PPT algorithm $\cS^*$ such that
    \begin{align}
        \SD(\cQ(1^\secp),\cS^*(1^\secp)) \le \frac{1}{p(\secp)}
    \end{align}
    holds for all sufficiently large $\secp\in\Lambda$. 
    Let $\Theta_{\cS^*}$ be the set of such $\secp$. 
    Define
   \begin{align}
        \Sigma_{\cS^*}\coloneqq\left\{\secp\in\mathbb{N}:\SD(\cQ(1^\secp),\cS^*(1^\secp))\le\frac{1}{p(\secp)}\right\}.
    \end{align}
    By the QAS/OWF condition, $f$ is a classically-secure OWF on $\Sigma_{\cS^*}$.
    Because $\Theta_{\cS^*}\subseteq\Sigma_{\cS^*}$, this means that $f$ is a classically-secure OWF on $\Theta_{\cS^*}$.
    Then $f$ is a classically-secure OWF on $\Lambda$.
    From \Cref{lem:OWF_OWPuzz}, $(\Samp',\Ver')$ is a classically-secure OWPuzz with $1$-correctness and $(1-\frac{1}{\secp^{\log\secp}}-\frac{1}{p(\secp)^2})$-security on $\Lambda$, but it contradict \Cref{eq:OWPuzz_soundness2}.
    
\end{proof}

\subsection{Proof of \cref{thm:OWPuzz_to_QAS/OWF}}
\label{subsec:OWPuzzle_OWF-SampQA}

In this subsection, we show \cref{thm:OWPuzz_to_QAS/OWF}, namely, classically-secure OWPuzzs $\Rightarrow$ the QAS/OWF condition.

\begin{proof}[Proof of \cref{thm:OWPuzz_to_QAS/OWF}]
    Let $(\Samp,\Ver)$ be a classically-secure OWPuzz with $c$-correctness and $s$-security such that
    $c(\secp)-s(\secp)\ge 1/q(\secp)$ for a polynomial $q$.
    For the sake of contradiction, assume that the QAS/OWF condition is not satisfied.
    Then, from \Cref{lem:quantifier} by taking $\cQ$ of the lemma as $\Samp$, we obtain the following statement: 
    For any polynomial $p$, there exists a PPT algorithm $\cS_p$ such that for any 
    efficiently-computable polynomial $n$ and
    any family $\{f_\secp:\bit^{n(\secp)}\to\bit^*\}_{\secp\in\mathbb{N}}$ of functions that are computable in classical deterministic polynomial-time, 
    there exists a PPT algorithm $\cR_p$ such that 
    \begin{align}\label{eq:cond_A_OWPuzz}
        \SD(\Samp(1^\secp),\cS_p(1^\secp)) \le \frac{1}{p(\secp)}   
    \end{align}
    and 
    \begin{align}\label{eq:cond_B_OWPuzz}
        \SD( \{x,f_\secp(x)\}_{x\gets\bit^{n(\secp)}} , \{\cR_p(1^{n(\secp)},f_\secp(x)),f_\secp(x)\}_{x\gets\bit^{n(\secp)}} ) \le \frac{1}{p(\secp)}
    \end{align}
    hold for infinitely many $\secp\in\mathbb{N}$.\footnote{We mean that both \cref{eq:cond_A_OWPuzz} and \cref{eq:cond_B_OWPuzz}
    are satisfied for {\it the same} $\secp$, and there are infinitely many such $\secp$.}
    
    Let $n(\secp)$ be the length of the random seed for $\cS_p(1^\secp)$.
    By using $\cS_p$, we define a function family $\{f_{p,\secp}\}_{\secp\in\mathbb{N}}$ as follows:
    \begin{itemize}
        \item $f_{p,\secp}:\bit^{n(\secp)}\to\bit^*$:
        \begin{enumerate}
            \item Take $r\in\bit^{n(\secp)}$ as input.
            \item Run $\cS_p(1^\secp;r)=(\puzz,\ans)$.
            \item Output $\puzz$.
        \end{enumerate}
    \end{itemize}
    As we have stated above, for this specific $\{f_{p,\secp}\}_{\secp\in\mathbb{N}}$, there exists a PPT algorithm $\cR_p^{f_{p,\secp}}$ such that \cref{eq:cond_A_OWPuzz,eq:cond_B_OWPuzz} hold for infinitely many $\secp\in \mathbb{N}$. 
    From $\cS_p$ and $\cR_p^{f_{p,\secp}}$, we construct a PPT algorithm $\cA$ such that  
    \begin{align}
       \Pr[\top\gets\Ver(\puzz,\cA(1^\secp,\puzz)):(\puzz,\ans)\gets\Samp(1^\secp)] \ge c(\secp)-\frac{1}{2q(\secp)}
    \end{align}
    holds for infinitely many $\secp\in\mathbb{N}$. 
    This means that $\cA$ breaks $s$-security of the classically-secure OWPuzz $(\Samp,\Ver)$.
    We define such $\cA$ as follows:
    \begin{itemize}
        \item $\cA(1^\secp,\puzz)\to\ans'$:
        \begin{enumerate}
            \item Take $1^\secp$ and $\puzz$ as input.
            \item Run $r\gets\cR_{6q}^{f_{6q,\secp}}(1^{n(\secp)},\puzz)$.
            \item Run $\cS_{6q}(1^\secp;r)=(\ans',\puzz')$.
            \item Output $\ans'$.
        \end{enumerate}
    \end{itemize}
    Then,
    \begin{align}
        &\Pr[\top\gets\Ver(\puzz,\cA(1^\secp,\puzz)):(\ans,\puzz)\gets\Samp(1^\secp)]\\
        &\ge \Pr[\top\gets\Ver(\puzz,\cA(1^\secp,\puzz)):(\ans,\puzz)\gets\cS_{6q}(1^\secp)] - \SD(\Samp(1^\secp),\cS_{6q}(1^\secp)) \\
        &=\Pr\left[\top\gets\Ver(\puzz,\ans'):
        \begin{array}{r}
        r\gets \bit^{n(\secp)}, \\
        \puzz=f_{6q,\secp}(r),\\
        r'\gets \cR_{6q}^{f_{6q,\secp}}(1^{n(\secp)},\puzz),\\
        (\puzz',\ans')=
        \cS_{6q}(1^\secp;r')
        \end{array}
        \right] -\frac{1}{6q(\secp)} \label{eq:trans_1} \\
        &\ge \Pr\left[\top\gets\Ver(\puzz,\ans'):
        \begin{array}{r}
        r\gets \bit^{n(\secp)}, \\
        \puzz=f_{6q,\secp}(r),\\
        (\puzz',\ans')
        =\cS_{6q}(1^\secp;r)
        \\
        \end{array}
        \right]-\frac{2}{6q(\secp)} \label{eq:trans_2} \\ 
        &=\Pr[\top\gets\Ver(\puzz,\ans):
        (\puzz,\ans)\gets\cS_{6q}(1^\secp)
        ]-\frac{2}{6q(\secp)} \\ 
        &\ge \Pr[\top\gets\Ver(\puzz,\ans):
        (\puzz,\ans)\gets\Samp(1^\secp)
        ]-\frac{3}{6q(\secp)} \label{eq:trans_3} \\ 
        &\ge c(\secp)-\frac{1}{2q(\secp)}
   \end{align}
    holds for infinitely many $\secp\in\mathbb{N}$.
    Here, we have used \Cref{eq:cond_A_OWPuzz} to obtain \Cref{eq:trans_1,eq:trans_3}, and used \Cref{eq:cond_B_OWPuzz} to obtain \Cref{eq:trans_2}.
    Therefore, $\cA$ breaks the soundness of classically-secure OWPuzz $(\Samp,\Ver)$.
\end{proof}

\newcommand{\val}{\mathsf{val}}
\newcommand{\ExtCom}{\mathsf{ExtCom}}

\section{Variants of IV-PoQ}
In \Cref{sec:equivalence_variants} we show equivalence among variants of IV-PoQ. 
In \Cref{sec:ZKIVPoQ}, we introduce \emph{zero-knowledge} IV-PoQ and show their relationship with OWFs.  
\subsection{Equivalence Among Variants of IV-PoQ}\label{sec:equivalence_variants}
We consider the following variant of IV-PoQ.
\begin{definition}[Quantum-Verifier IV-PoQ]
A quantum-verifier IV-PoQ $(\cP,\cV_1,\cV_2)$ is defined similarly to IV-PoQ (\Cref{def:IVPoQ}) except that $\cV_1$ is QPT instead of PPT but still only sends classical messages. 
\end{definition}
We show that the following equivalence theorem. 

\begin{theorem}\label{thm:equivalence_variants}
The following are equivalent:
\begin{enumerate}
\item Public-coin IV-PoQ exist. \label{item:pc-IV-PoQ}
\item IV-PoQ exist. \label{item:IV-PoQ}
\item Quantum-verifier IV-PoQ exist. \label{item:QV-IV-PoQ}
\end{enumerate}
\end{theorem}
\begin{proof}[Proof of \cref{thm:equivalence_variants}]
It is clear that \Cref{item:pc-IV-PoQ} implies \Cref{item:IV-PoQ}    
and \Cref{item:IV-PoQ} implies \Cref{item:QV-IV-PoQ}. Below, we show that  \Cref{item:QV-IV-PoQ} implies \Cref{item:pc-IV-PoQ}. To show this, we first observe that quantum-verifier IV-PoQ imply Int-QAS since the proof of \Cref{lem:IV-PoQ_to_Int-SampQA} works even if $\cV_1$ is QPT. Thus, combined with \Cref{thm:Int-QAS_to_QAS/OWF}, quantum-verifier IV-PoQ imply the QAS/OWF condition. Next, we observe that the IV-PoQ constructed from the QAS/OWF condition in the proof of \Cref{thm:QAS/OWF_to_IV-PoQ} is public-coin. To see this, recall that the IV-PoQ is obtained by combining the non-interactive protocol obtained by \cref{lem:QAS_IVPoQ} and the IV-PoQ based on OWFs by \Cref{lem:OWF_IVPoQ}. The former is non-interactive and in particular there is no message sent from the verifier, and thus it is public-coin. Moreover, the latter is also public-coin by \Cref{lem:OWF_IVPoQ}.   
Thus, their combination (as in the proof of \Cref{thm:QAS/OWF_to_IV-PoQ}) also results in public-coin IV-PoQ. Thus, the IV-PoQ constructed in the proof of \Cref{thm:QAS/OWF_to_IV-PoQ} is public-coin. Combining the above observations, we complete the proof that quantum-verifier IV-PoQ imply public-coin IV-PoQ.
\end{proof}

\subsection{Zero-Knowledge IV-PoQ}\label{sec:ZKIVPoQ}
We give a definition of zero-knowledge IV-PoQ below. 
\begin{definition}[Zero-Knowledge IV-PoQ]
An IV-PoQ $(\cP,\cV_1,\cV_2)$ satisfies 
 computational (resp. statistical) zero-knowledge if for any PPT malicious verifier $\cV_1^*$, there exists a PPT simulator $\cS$ such that for any PPT (resp. unbounded-time) distinguisher $\cD$,  
\begin{equation}
\left|\Pr[\cD(\mathsf{view}\langle\cP,\cV^*_1\rangle(1^\secp))=1]-\Pr[\cD(\cS(1^\secp))=1]\right|\le \negl(\secp)
\end{equation}
where $\mathsf{view}\langle\cP,\cV^*_1\rangle(1^\secp)$ means the view of $\cV^*_1$ which consists of the transcript and the random coin of $\cV^*_1$. 

We say that an IV-PoQ $(\cP,\cV_1,\cV_2)$ satisfies honest-verifier computational (resp. statistical) zero-knowledge if the above holds for the case of $\cV^*_1=\cV_1$. 
\end{definition}
\begin{remark}
Standard definitions of the zero-knowledge property in the literature usually consider \emph{non-uniform} malicious verifiers and distinguishers. 
On the other hand, since we treat the uniform model as a default notion in this paper, we define the zero-knowledge property in the uniform-style as above. However, we remark that this choice of model of computation is not essential for the results of this subsection, and all the results of this subsection readily extend to the non-uniform setting with essentially the same proofs. 
\end{remark}
We show relationships between zero-knowledge IV-PoQ and OWFs. 
First, we show that honest-verifier statistical zero-knowledge IV-PoQ  imply classically-secure OWFs. 
\begin{theorem}\label{thm:ZKIVPoQ_to_OWF}
    If honest-verifier statistical zero-knowledge IV-PoQ exist, then classically-secure OWFs exist. 
\end{theorem}
\begin{proof}[Proof of \cref{thm:ZKIVPoQ_to_OWF}]
 Let $(\cP,\cV_1,\cV_2)$  be an honest-verifier statistical zero-knowledge IV-PoQ. By the proof of \Cref{lem:IV-PoQ_to_Int-SampQA}, $(\cP,\cV_1)$ is an Int-QAS, and thus by the proof of \Cref{thm:Int-QAS_to_QAS/OWF}, the QAS/OWF condition holds where $\cQ= \langle \cP,\cV_1\rangle$. That is, there exist a polynomial $p$, and a function $f:\bit^*\to\bit^*$ that is computable in classical deterministic polynomial-time
    such that for any PPT algorithm $\cS$, the following holds:
    if we define
    \begin{align}
        \Sigma_\cS := \left\{ \secp\in\mathbb{N} \ \middle|\  \SD( \langle \cP,\cV_1\rangle(1^\secp),\cS(1^\secp)) \le \frac{1}{p(\secp)} \right\},
    \end{align}
    then $f$ is a classically-secure OWF on $\Sigma_\cS$. 
 By the honest-verifier statistical zero-knowledge property of $(\cP,\cV_1,\cV_2)$, there is a PPT simulator $\cS$ such that 
 $\SD(\langle \cP,\cV_1\rangle(1^\secp),\cS(1^\secp)) \le \negl(\secp)$.\footnote{Recall that we write $\langle \cP,\cV_1\rangle(1^\secp)$ to mean the machine that outputs a transcript of interaction between $\cP$ and $\cV_1$. 
 Since the honest-verifier statistical zero-knowledge requires the simulator to simulate both the transcript and the verifier's randomness, it is trivial to simulate only the transcript. 
 } For this $\cS$, $\Sigma_\cS$ consists of all but finite elements of  $\mathbb{N}$ by the definition of negligible functions. 
Since $f$ is a classically-secure OWF on $\Sigma_\cS$, this implies the existence of classically secure OWFs by \cref{lem:OWF_on_N}.
\end{proof}

Next, we show that OWFs imply computational zero-knowledge IV-PoQ. 
\begin{theorem}\label{thm:OWF_to_ZKIVPoQ}
    If classically-secure OWFs exist, then computational zero-knowledge IV-PoQ exist. 
\end{theorem}
To prove \Cref{thm:OWF_to_ZKIVPoQ} we rely on (classically-secure) extractable commitments. The following definition is based on the one in \cite{TCC:PasWee09}. 

\begin{definition}[Classically-Secure Extractable Commitments]\label{definition:ext-com} 
A classically-secure extractable commitment scheme consists of interactive PPT algorithms $\cS$ (sender) and $\cR$ (receiver). An execution of the scheme is divided into two phases, the commit phase and open phase. In the commit phase, $\cS$ takes the security parameter $1^\secp$ and message $m$ as input, $\cR$ takes the security parameter $1^\secp$ as input, and  $\cS$ and $\cR$ interact with each other where $\cR$ may declare rejection and abort at any point of the commit phase. We call a transcript of the commit phase a commitment and denote it by $\com$.  
In the open phase, $\cS$ sends $m$  and a classical string $d$ (decommitment) to $\cR$ and $\cR$ outputs $\top$ (accept) or $\bot$ (reject). We assume that $\cR$ is stateless, i.e., there is no state information of $\cR$ kept from the commit phase.\footnote{We can also define a commitment scheme with a stateful receiver, but we focus on the stateless receiver case for simplicity.} We require the following four properties.
\begin{itemize}
\item{\bf Correctness.} For any $\secp$ and $m$, if $\cS(1^\secp,m)$ and $\cR(1^\secp)$ run the commit and open phases honestly, then $\cR$ always accepts.  
\item{\bf Computational hiding.} 
For any (stateful) PPT malicious receiver $\cR^*$, we have
\begin{equation}
\left|
\Pr\left[
\mathsf{Out}_{\cR^*}\langle \cS(m_b),\cR^*\rangle(1^\secp)=b:
\begin{array}{ll}
(m_0,m_1)\gets \cR^*(1^\secp)\\
b\gets \bit\\
\end{array}
\right]
-\frac{1}{2}
\right|
\le \negl(\secp)
\end{equation}
where $\mathsf{Out}_{\cR^*}\langle \cS(m_b),\cR^*\rangle(1^\secp)$ means the output of $\cR^*$ after interacting with $\cS$ with the common input $1^\secp$ and private input $m_b$ in the commit phase. 
\item {\bf Statistical binding.\footnote{This follows from extractability, but we state it separately for convenience.}}
For any unbounded-time cheating sender $\cS^*$, let $\com$ be a commitment generated by $\cS^*$ and the honest receiver $\cR$. Then 
there is at most one $m$ such that there exists $d$ such that $\cR$ accepts $(m,d)$ as an opening of $\com$ except for a negligible probability. 
\item {\bf Extractability.} 
There is an expected PPT oracle machine (the extractor) $\mathcal{E}$ that takes the security parameter $1^\secp$ as input, 
given oracle access to any deterministic unbounded-time cheating sender $\cS^*$, and outputs a pair $(\com, m^*)$ such that: 
\begin{itemize}
\item
{\bf Simulation:} $\com$ is identically distributed to a commitment generated by $\cS^*$ and $\cR$. 
\item
{\bf Extraction:} the probability that $\com$ is accepting (i.e., the transcript indicates that $\cR$ never aborts in the commit phase) and $m^* =\bot$ is negligible.
\item
{\bf Binding:} 
If $m^*\ne \bot$, then there is no $m\ne m^*$ and $d$ such that $\cR$ accepts $(m,d)$ as an opening of $\com$. 
\end{itemize}
\end{itemize}
\end{definition} 
\begin{theorem}[\cite{TCC:PasWee09}]\label{thm:extcom}
If classically-secure OWFs exist, then classically-secure extractable commitments exist. 
\end{theorem}

We introduce a notation on extractable commitments. For a commitment $\com$, if there is a unique $m$ such that there is $d$ such that $\cR$ accepts $(m,d)$ as an opening of $\com$, then we define $\val(\com)=m$. If such $m$ does not exist or not unique, we define  $\val(\com)= \bot$. The following lemmas are easy to prove. 
\begin{lemma}\label{lem:val_com_is_m}
If $\com$ is generated by $\cS(1^\secp,m)$ and $\cR(1^\secp)$, then $\val(\com)=m$ except for a negligible probability.
\end{lemma}
\begin{proof}
    This immediately follows from correctness and statistical binding. 
\end{proof}
\begin{lemma}\label{lem:val_com_is_mstar_or_bot}
For any deterministic unbounded-time cheating sender $\cS^*$, 
if $\mathcal{E}^{\cS^*}(1^\secp)$ outputs $(\com,m^*)$, then $\val(\com)=m^*$ or $\val(\com)=\bot$ except for a negligible probability.  
\end{lemma}
\begin{proof}
This immediately follows from the second and third items of the extractability. 
\end{proof}

Then we prove \Cref{thm:OWF_to_ZKIVPoQ}. 
\begin{proof}[Proof of \Cref{thm:OWF_to_ZKIVPoQ}]
By \Cref{lem:OWF_IVPoQ}, public-coin IV-PoQ exist if classically-secure OWFs exist. 
By \Cref{thm:extcom}, classically-secure extractable commitments exist if classically-secure OWFs exist. 
Thus, it suffices to construct computational zero-knowledge IV-PoQ 
 from public-coin IV-PoQ and classially-secure extractable commitments. We show how to do it below. 

Let $\Pi=(\cP,\cV_1,\cV_2)$ be a public-coin IV-PoQ that satisfies $c$-completeness and $s$-soundness 
and let $\mathsf{ExtCom}$ be a classically-secure extractable commitment scheme. 
Without loss of generality, we assume that $\Pi$ is an $\ell$-round protocol where the first message is sent from the verifier and verifier's messages in each round are $n$-bit strings for some polynomials $\ell=\ell(\secp)$ and $n=n(\secp)$. 
Then we construct a computational zero-knowledge IV-PoQ $\tilde{\Pi}=(\tilde{\cP},\tilde{\cV}_1,\tilde{\cV}_2)$ that works as follows. 
\begin{itemize}
\item[\bf (The first phase)]
\item Upon receiving the security parameter $1^\secp$,
$\tilde{\cP}$ sets $\st_0\seteq \ket{1^\secp}$. 
\item For $i=1,2,\ldots,\ell$, do the following:
\begin{itemize}
    \item $\tilde{\cV}_1$ chooses $v_i\gets \bit^n$ and sends $v_i$ to $\tilde{\cP}$. 
    \item $\tilde{\cP}$ runs $\cP$ on the state $\st_{i-1}$ and $i$-th verifier's message $v_i$ to generate $i$-th prover's message $p_i$. 
    Let $\st_i$ be the internal state of $\cP$ at this point. 
    \item $\tilde{\cP}$ commits to $p_i$ using $\mathsf{ExtCom}$ where $\tilde{\cP}$ and $\tilde{\cV}_1$ play the roles of the sender and receiver, respectively. Let $\com_i$ be the commitment generated in this step. 
\end{itemize}
\item[\bf (The second phase)]
\item Upon receiving $1^\secp$ and a transcript $\tilde{\tau}=(v_1,\com_1,v_2,\com_2,\ldots,v_\ell,\com_\ell)$ of the first phase, $\tilde{\cV}_2$ computes $p'_i=\val(\com_i)$ by brute-force for all $i\in [\ell]$. If $p'_i=\bot$ for some $i\in [\ell]$, then $\tilde{\cV}_2$ outputs $\bot$. 
Otherwise $\tilde{\cV}_2$ runs $\cV_2$ on transcript  $\tau=(v_1,p'_1,v_2,p'_2,\ldots,v_\ell,p'_\ell)$ and outputs whatever $\cV_2$ outputs. 
\end{itemize}

By \Cref{lem:val_com_is_m}, we have $p'_i=p_i$ except for a negligible probability when we run the protocol honestly. Thus, $\tilde{\Pi}$ satisfies $(c-\negl)$-completeness by $c$-completeness of $\Pi$. 
The computational zero-knowledge property of $\tilde{\Pi}$ immediately follows from the computational hiding property of $\mathsf{ExtCom}$ since a simulator can simply commit to $0...0$ instead of to $p_i$.  

Below, we reduce $(s+\negl)$-soundness of $\tilde{\Pi}$ to $s$-soundness of $\Pi$ using the extractability of  $\mathsf{ExtCom}$ as follows. 
Toward contradiction, suppose that there is a PPT cheating prover ${\tilde{\cP}}^*$ such that 
  \begin{equation} \label{eq:success_prob_assumption}
            \Pr[\top\gets\tilde{\cV}_2(1^\secp,\tau):\tau\gets\langle\tilde{\cP}^*,\tilde{\cV}_1\rangle(1^\secp)] \ge s(\secp)+1/q(\secp)
        \end{equation}
        for some polynomial $q$ and infinitely many $\secp$. 
For each $j\in \{0,1,.\ldots,\ell\}$, we consider a (not necessarily PPT) cheating provers $\cP^*_j$ against $\Pi$ that works as follows:

\begin{itemize}
\item $\cP^*_j$ takes randomness $r$, which has the same length as randomness of ${\tilde{\cP}}^*$.
\item For $i=1,2,\ldots,j$, do the following:
\begin{itemize}
\item Receive $v_i$ from $\tilde{\cV}_1$. 
\item Let $\tilde{\cP}^*[r,\tau_i]$ be the part of $\tilde{\cP}^*$ that runs the $i$-th execution of $\mathsf{ExtCom}$ where we hardwire the randomness $r$ and the partial transcript $\tau_i=(v_1,\com_1,\ldots,v_{i-1},\com_{i-1},v_i)$. 
\item Run $(\com_i,p_i)\gets \mathcal{E}^{\tilde{\cP}^*[r,\tau_i]}(1^\secp)$ and send $p_i$ to $\tilde{\cV}_1$.
\end{itemize}
\item For $i=j+1,j+2,\ldots,\ell$, do the following:
\begin{itemize}
\item Receive $v_i$ from $\tilde{\cV}_1$. 
\item Run $\tilde{\cP}^*$ on the fixed randomness $r$  and the partial transcript $\tau_i=(v_1,\com_1,\ldots,v_{i-1},\com_{i-1},v_i)$ to complete the $i$-th execution of $\mathsf{ExtCom}$. 
Let $\com_i$ be the commitment generated in this step.  
\item Compute $p_i=\val(\com_i)$ by brute-force.
\item If $p_i=\bot$, abort. Otherwise 
send $p_i$ to $\cV_1$. 
\end{itemize}
\end{itemize}
By the definition, it is easy to see that 
\begin{equation} \label{eq:success_prob_initial}
\Pr[\top\gets\cV_2(1^\secp,\tau):\tau\gets\langle\cP^*_0,\cV_1\rangle(1^\secp)] 
=
\Pr[\top\gets\tilde{\cV}_2(1^\secp,\tau):\tau\gets\langle\tilde{\cP}^*,\tilde{\cV}_1\rangle(1^\secp)]. 
\end{equation}
For $j\in [\ell]$, the only difference between $\cP^*_{j-1}$ and $\cP^*_j$ is how to generate $(\com_j,p_j)$: $\cP^*_{j-1}$ runs $\tilde{\cP}^*$ to generate $\com_j$ and then sets $p_j=\val(\com_j)$ by brute-force whereas $\cP^*_{j}$ generates them by running the extractor. By the first item of extractability, the distributions of $\com_j$ are identical for both cases. Moreover, by \Cref{lem:val_com_is_mstar_or_bot}, $p_j$ takes the same value in both cases unless $\val(\com_j)=\bot$ except for a negligible probability. Note that if $\val(\com_j)=\bot$, then  $\cP^*_{j-1}$ aborts and thus never passes the verification of $\cV_2$. Thus, we have 
\begin{equation} \label{eq:success_prob_hybrid}
\Pr[\top\gets\cV_2(1^\secp,\tau):\tau\gets\langle\cP^*_{j},\cV_1\rangle(1^\secp)] 
\ge 
\Pr[\top\gets\cV_2(1^\secp,\tau):\tau\gets\langle\cP^*_{j-1},\cV_1\rangle(1^\secp)]-\negl(\secp). 
\end{equation}
By combining \Cref{eq:success_prob_assumption,eq:success_prob_initial,eq:success_prob_hybrid}, 
\begin{equation} \label{eq:success_prob_ell}
\Pr[\top\gets\cV_2(1^\secp,\tau):\tau\gets\langle\cP^*_{\ell},\cV_1\rangle(1^\secp)] 
\ge 
s(\secp)+1/q(\secp)-\negl(\secp)
\end{equation}
for infinitely many $\secp$. 

We remark that $\cP^*_\ell$ no longer runs brute-force, but it is still not PPT since it runs the extractor $\mathcal{E}$ that runs in \emph{expected} PPT. This can be converted into a PPT machine by a standard truncation technique. That is, let $\cP^*_{\ell+1}$ be a cheating prover that runs similarly to $\cP^*_{\ell}$ except that whenever it runs $\mathcal{E}$, if its running time is $2\ell(\secp) q(\secp)$ times larger than its expectation, then it aborts. 
Clearly, $\cP^*_{\ell+1}$ runs in PPT. Moreover, by Markov's inequality, for each invocation of $\mathcal{E}$, the probability that $\cP^*_{\ell+1}$  aborts is at most $1/(2\ell(\secp) q(\secp))$. Since $\cP^*_{\ell+1}$ runs $\mathcal{E}$ $\ell(\secp)$ times, 
by the union bound, the probability that this occurs is at most $1/(2 q(\secp))$. Thus, \Cref{eq:success_prob_ell}, we have
\begin{equation} 
\Pr[\top\gets\cV_2(1^\secp,\tau):\tau\gets\langle\cP^*_{\ell+1},\cV_1\rangle(1^\secp)] 
\ge 
s+1/(2q(\secp))-\negl(\secp)
\end{equation}
for infinitely many $\secp$. 
This contradicts $s$-soundness of $\Pi$. 
Thus, $\tilde{\Pi}$ satisfies $(s+\negl(\secp))$-soundness. 
\end{proof}

\paragraph{Zero-knowledge PoQ.}
Though our main focus is on IV-PoQ, we briefly discuss zero-knowledge (efficiently-verifiable) PoQ. First, we observe that the conversion in the proof of \Cref{thm:OWF_to_ZKIVPoQ} works in the efficiently-verifiable setting as well if we introduce an additional layer of zero-knowledge proofs where the prover proves that the committed transcript passes the verification. However, the conversion requires the base PoQ to be public-coin while most existing PoQ are not public-coin. 
Fortunately, we observe that we can relax the public-coin property to the ``transcript-independent'' property which means that the distribution of verifier's messages does not depend on the transcript and only depends on the verifier's randomness.
At first glance, one may think that it is problematic if the verifier uses its private randomness to make a decision in which case the statement that ``the committed transcript passes the verification'' is not an $\mathbf{NP}$ statement. However, this issue can be resolved by letting the verifier reveal its randomness after receiving all the commitments from the prover.\footnote{A similar idea is used in \cite{C:BKLMMVVY22}.} Since the randomness is revealed after the commitments are sent, a cheating prover can no longer change the committed transcript by the binding property of the extractable commitment, and thus this does not affect the soundness. 
In summary, we can generically upgrade any PoQ with transcript-independent verifiers into a (computational) zero-knowledge PoQ 
by additionally assuming the existence of OWFs. 
To our knowledge, all existing PoQ \cite{JACM:BCMVV21,NatPhys:KMCVY22,10.1145/3658665,ITCS:MorYam23,STOC:KLVY23,knowledge_assumptions} have  transcript-independent verifiers.

\paragraph{Toward equivalence between OWFs and zero-knowledge IV-PoQ.}
\Cref{thm:ZKIVPoQ_to_OWF,thm:OWF_to_ZKIVPoQ} can be regarded as a loose equivalence between OWFs and zero-knowledge IV-PoQ. However, there is a gap between them as \Cref{thm:ZKIVPoQ_to_OWF} assumes honest-verifier \emph{statistical} zero-knowledge while \Cref{thm:OWF_to_ZKIVPoQ} only gives \emph{computational} zero-knowledge. It is an interesting open question if we can fill the gap. 

There are two approaches toward solving that. One is to show that computational zero-knowledge IV-PoQ imply OWFs and the other is to show that OWFs imply honest-verifier statistical zero-knowledge IV-PoQ. For the former approach, the technique of \cite{OW93,Vad06}, which shows that computational zero-knowledge arguments for average-case-hard languages imply OWFs, might be useful, but it is unclear how to adapt their technique to the setting of IV-PoQ.  

We also do not have solution for the latter approach either, but we have the following observation.   
 We observe that we can construct statistical zero-knowledge IV-PoQ (or even efficiently-verifiable PoQ)  if we additionally assume the existence of an $\mathbf{NP}$ search problem that is easy for QPT algorithms but hard for PPT algorithms (or equivalently publicly-verifiable one-round PoQ). To see this, we can consider a protocol where the honest quantum prover solves the $\mathbf{NP}$ search problem and then proves the knowledge of the solution by using statistical zero-knowledge arguments of knowledge for $\mathbf{NP}$, which exists if OWFs exist~\cite{HNORV09}.  
Examples of classically-hard and quantumly-easy $\mathbf{NP}$ search problems are the factoring and discrete-logarithm problems (assuming classical hardness of them)~\cite{FOCS:Shor94}. Another example based on a random oracle was recently found in \cite{10.1145/3658665}. Thus, based on the random oracle heuristic~\cite{CCS:BelRog93}, we have a candidate construction of statistical zero-knowledge PoQ from hash functions.\footnote{This is \emph{not} a construction in the quantum random oracle model since we use the hash function in a non-black-box manner. Instead, we rely on the assumption that the problem considered in \cite{10.1145/3658665} is classically hard when the random oracle is instantiated with a concrete hash function.} 
Though this is far from a construction solely based on OWFs, this can be seen as an evidence that ``structured'' assumptions are not necessary for statistical zero-knowledge PoQ, let alone for statistical zero-knowledge IV-PoQ.\footnote{``Structure'' is a commonly used informal term that refers to problems behind constructions of existing public key encryption such as the hardness of factoring, discrete-logarithm, learning with errors, etc. On the other hand, hash functions are often regarded as ``unstructured''. See \cite{Barak17} for more context.}

\ifnum\anonymous=1
\else
{\bf Acknowledgements.}
TM is supported by
JST CREST JPMJCR23I3,
JST Moonshot R\verb|&|D JPMJMS2061-5-1-1, 
JST FOREST, 
MEXT QLEAP, 
the Grant-in Aid for Transformative Research Areas (A) 21H05183,
and 
the Grant-in-Aid for Scientific Research (A) No.22H00522.
YS is supported by JST SPRING, Grant Number JPMJSP2110.
\fi

\appendix
\section{On Uniformity of Adversaries}\label{sec:uniform}
As mentioned in \Cref{footnote_nonuniform}, we consider the uniform adversarial model in this paper. 
The only place where the uniformity of the adversary plays a crucial role is in the proof of \Cref{lem:QAS_IVPoQ} where we relate the hardness of non-interactive search and sampling problems. In the proof, we make heavy use of Kolmogorov complexity, which is defined with respect to uniform Turing machines, and thus the same proof does not work when we consider non-uniform adversaries. Thus, all the results that rely on \Cref{lem:QAS_IVPoQ} do not extend to the non-uniform setting. They include \Cref{lem:QAS_to_IV-PoQ,thm:QAS/OWF_to_IV-PoQ,thm:QAS/OWF_to_OWPuzz,thm:equivalence_variants}.
On the other hand, \Cref{thm:QASOWF_OWF_Samp,lem:IV-PoQ_to_Int-SampQA,thm:Int-QAS_to_QAS/OWF,thm:OWPuzz_to_QAS/OWF,thm:ZKIVPoQ_to_OWF,thm:OWF_to_ZKIVPoQ}  seem to extend to the non-uniform setting with appropriate adaptations. 
\section{Proof of \cref{lem:OWF_on_N}}
\label{sec:OWF_on_N}

\begin{proof}[Proof of \cref{lem:OWF_on_N}] 
The only if part is trivial. We show the if part.
Let $f:\bit^*\to\bit^*$ be a OWF on $\mathbb{N}\setminus \Sigma$ for a finite subset $\Sigma\subseteq\mathbb{N}$. 
    Then by the definition of OWFs on $\mathbb{N}\setminus \Sigma$, there exists an efficiently-computable polynomial $n$ such that
    for any PPT adversary $\cA$ and any 
    polynomial $p$
    there exists $\secp^*\in\mathbb{N}$ such that
    \begin{equation}\label{eq:OWF_security}
    \Pr[f(x')=f(x): x\gets\bit^{n(\secp)}, x'\gets\cA(1^{n(\secp)},f(x))] \le\frac{1}{p(\secp)}
    \end{equation} 
    for all $\secp\ge\secp^*$ in $\mathbb{N}\setminus \Sigma$. 
    Since $\Sigma$ is a finite set, we can assume that $\secp^*$ is larger than the 
    largest element of $\Sigma$ without loss of generality. Then, \Cref{eq:OWF_security} holds for all $\secp\ge\secp^*$ in $\mathbb{N}$. 

    From such $f$, we construct a OWF $g:\bit^*\to\bit^*$ as follows. 
  \begin{enumerate}
  \item 
  On input $x\in\bit^\ell$, find the maximum $\secp$ such that
  $n(\secp)\le \ell$. Set $\secp_\ell$ to be such maximum $\secp$.
  \item
  Output $g(x)\coloneqq f(x_{1,...,n(\secp_\ell)})$,
  where $x_{1,...,n(\secp_\ell)}$ is the first $n(\secp_\ell)$ bits of $x$.
  \end{enumerate}
 Then
 for any PPT $\cA$ and any polynomial $p$, 
if we  
  define a polynomial $q(\secp)\coloneqq p(n(\secp+1))$,
 there exists $\secp^*_{\cA,q}\in\mathbb{N}$
 such that
    \begin{align}
    &\Pr[g(x')=g(x): x\gets\bit^\ell, x'\gets\cA(1^\ell,g(x))] \\
    &=\Pr[f(x'_{1,...,n(\secp_\ell)})=f(x_{1,...,n(\secp_\ell)}): x\gets\bit^\ell, x'\gets\cA(1^\ell,f(x_{1,...,n(\secp_\ell)})] \\
    &=\Pr[f(x'_{1,...,n(\secp_\ell)})=f(x): x\gets\bit^{n(\secp_\ell)}, x'\gets\cA(1^\ell,f(x))] \\
    &=\Pr[f(w)=f(x): x\gets\bit^{n(\secp_\ell)}, w\gets\cA(1^{n(\secp_\ell)},f(x))] \\
    &\le\frac{1}{q(\secp_\ell)}\\
    &=\frac{1}{p(n(\secp_\ell+1))}\\
    &\le\frac{1}{p(\ell)}
    \end{align} 
    if $\secp_\ell\ge\secp^*_{\cA,q}$.
    For all sufficiently large $\ell\in\mathbb{N}$, $\secp_\ell\ge\secp^*_{\cA,q}$,
    and therefore the above inequality is satisfied for all sufficiently large $\ell\in\mathbb{N}$.
    This means that $g$ is a OWF.
\end{proof}
\section{Proof of \cref{lem:univ_OWF}}
\label{sec:univ_OWF}

\begin{proof}[Proof of \Cref{lem:univ_OWF}]
    The proof is similar to the universal construction technique \cite{STOC:Levin85,EC:HKNRR05} of OWFs.
    Let $M_1, M_2,...$ be an enumeration of all Turing machines.
    We construct a function $g$ as follows:
    \begin{itemize}
        \item $g:\bit^*\to\bit^*$:
        \begin{enumerate}
            \item Take $y\in\bit^\ell$ as input. 
            \item Let $N\in\mathbb{N}$ be the maximum value such that $N^2\le\ell$.
            \item Let $y_1\|y_2\|...\|y_N$ be the $N^2$-bit prefix of $y$, where $y_i\in\bit^N$ for all $i\in[N]$.
            \item For all $i\in[N]$, run $M_i(y_i)$ for $N^3$ steps. If $M_i(y_i)$ halts, then set $v_i:=M_i(y_i)$. Otherwise, set $v_i:=\bot$.
            \item Output $v_1\|v_2\|...\|v_N$.
        \end{enumerate}
    \end{itemize}
    The time to compute $g(y)$ is $O(|y|^2)$ for all $y\in\bit^*$ and therefore $g$ is computable in classical deterministic polynomial-time.
    
    Fix a subset $\Sigma\subseteq\mathbb{N}$.
    Let $f:\bit^*\to\bit^*$ be a classically-secure OWF on $\Sigma$.
    From \cref{lem:padding} below, we can assume that the time to compute $f(x)$ is $O(|x|^2)$ for any $x\in\bit^*$.
    Then, by the definition of OWFs on $\Sigma$, 
    there exists an efficiently-computable polynomial $n$ such that for any PPT algorithm $\cR$ and for any polynomial $p$, 
    there exists $\secp^*\in\mathbb{N}$ such that
    \begin{align}\label{eq:quantifier_OWF_condition}
        \Pr[f(z)=f(x):x\gets\bit^{n(\secp)},z\gets\cR(1^{n(\secp)},f(x))] \le \frac{1}{p(\secp)}
    \end{align}
    holds for all $\secp\ge\secp^*$ in $\Sigma$.
    Our goal is to show that $g$ is a classically-secure OWF on $\Sigma$.

    For the sake of contradiction, we assume that $g$ is not a classically-secure OWF on $\Sigma$.
    Then, $\Sigma$ is an infinite subset of $\mathbb{N}$, because if $\Sigma$ is finite, $g$ is trivially a OWF on $\Sigma$.
   Moreover, by the definition of OWFs on $\Sigma$,
   for any efficiently-computable polynomial $m$, there exist a PPT algorithm $\cA$, a polynomial $q$ and an infinite subset $\Lambda\subseteq\Sigma$ 
   such that
    \begin{align}\label{eq:univ_OWF_break}
        \Pr[g(w)=g(y):y\gets\bit^{m(\secp)},w\gets\cA(1^{m(\secp)},g(y))] > \frac{1}{q(\secp)}
    \end{align}
    holds for all $\secp\in\Lambda$.
    Because $f$ is computable in classical deterministic polynomial-time, there exists a deterministic Turing machine $F$ such that $F(x)=f(x)$ for all $x\in\bit^*$.
    Moreover, because as we have said the time to compute $f(x)$ is $O(|x|^2)$, 
    the running time of $F$ is $O(|x|^2)$.
    The description size of $F$ is a constant of $|x|$, and therefore
    there exists $\alpha^*\in\mathbb{N}$ such that $M_{\alpha^*}=F$.
    We define the PPT algorithm $\cR^*$ as follows:
    \begin{itemize}
        \item $\cR^*(1^{n(\secp)},f(x))\to z$:
        \begin{enumerate}
            \item Take $1^{n(\secp)}$ and $f(x)$ as input, where $x\gets\bit^{n(\secp)}$.
            \item If $n(\secp)\ge\alpha^*$, set $j=\alpha^*$. Otherwise, set $j=1$. Set $v_j:=f(x)$.
            \item Set $N:=n(\secp)$.
            \item For all $i\in[N]$ except for $j$, sample $y_i\gets\bit^{N}$ and run $M_i(y_i)$ for $N^3$ steps.
            If $M_i(y_i)$ halts, set $v_i:=M_i(y_i)$. Otherwise, set $v_i:=\bot$.
            \item Run $\cA(1^{N^2},v_1\|...\|v_{N})$ to obtain $w\in\bit^m$.\footnote{Note that the length $m$ of the output $w$ of $\cA$ is not necessarily equal to the input length $N^2$, because
            $g(w)=g(y)$ could happen in \cref{eq:univ_OWF_break} for some $w$ whose length is longer than $N^2$.}
            \item Parse $w:=w_1\|...\|w_{N}\|w'$, 
            where $w_i\in\bit^{N}$ for all $i\in[N]$ and $w'\in\bit^{m-N^2}$.
            \item Output $z:=w_j$.
        \end{enumerate}
    \end{itemize}
    We can show that
    for any efficiently computable polynomial $n$, there exists a PPT algorithm $\cR^*$ and
    a polynomial $q$ and an infinite subset 
    $\Lambda'\coloneqq\{\secp\in\Lambda:n(\secp)\ge\alpha^*\}$ 
    of $\Sigma$ such that
    \begin{align}
    &\Pr[f(z)=f(x):x\gets\bit^{n(\secp)},z\gets\cR^*(1^{n(\secp)},f(x))] \\
        &\ge \Pr[ g(w)=g(y) : y\gets\bit^{n(\secp)^2}, w\gets\cA(1^{n(\secp)^2},g(y)) ] \label{eq:botdenai}\\
       & > \frac{1}{q(\secp)} \label{eq:botdenai2}
   \end{align} 
    holds for all $\secp\in\Lambda'$.
    Here, in \cref{eq:botdenai} we have used the fact that the time to compute $f(x)$ is $O(|x|^2)$, and in
    \cref{eq:botdenai2}, we have used \Cref{eq:univ_OWF_break}.
    This contradicts \Cref{eq:quantifier_OWF_condition}, and therefore $g$ is a classically-secure OWF on $\Sigma$.
\end{proof}

\begin{lemma}
\label{lem:padding}
Let $\Sigma\subseteq\mathbb{N}$ be a set.
If there exists a OWF $f$ on $\Sigma$, 
there exists a OWF $f'$ on $\Sigma$   
such that the time to compute $f'(x)$ is $O(|x|^2)$ for any $x\in\bit^*$.
\end{lemma}

\begin{proof}[Proof of \cref{lem:padding}]
The proof is based on the padding argument of \cite{Goldreich04}.
Let $\Sigma\subseteq\mathbb{N}$ be a set.
Let $f$ be a OWF on $\Sigma$.
If the time to compute $f$ is $O(|x|^c)$ for $c>2$, we define a function $f'$ as follows.
\begin{enumerate}
\item 
On input a bit string, $z\in\bit^r$, find the maximum $i\in[r]$ such that $|z|\ge |z_{1,...,i}|^c$,
where $z_{1,...,i}$ is the first $i$ bits of $z$.
\item
Output $f'(z)\coloneqq f(z_{1,...,i})\|z_{i+1,...,r}$, 
where $z_{i+1,...,r}$ is the last $r-i$ bits of $z$.
\end{enumerate}
The time to compute $f'(z)$ is $O(|z|^2)$ because finding the maximum $i$  
spends $O(|z|^2)$ time\footnote{
    A single-tape Turing machine spends $O(|z|^2)$ time to find the maximum $i$.
} and computing $f(z_{1,...,i})$ spends $O(|z_{1,...,i}|^c)\le O(|z|)$ time.

Moreover, we can show that $f'$ is a OWF on $\Sigma$. 
For the sake of contradiction, assume that $f'$ is not a OWF on $\Sigma$.
Then by the definition of OWFs on $\Sigma$,
for any efficiently computable polynomial $n$, there exist a PPT adversary $\cA$, a polynomial $p$, and an infinite subset $\Lambda\subseteq\Sigma$ such that
\begin{align}
\Pr[f'(z)=f'(w):z\gets\bit^{n(\secp)},w\gets\cA(1^{n(\secp)},f'(z))]    
>\frac{1}{p(\secp)}
\end{align}
holds for all $\secp\in\Lambda$.
From such $\cA$, we can construct a PPT adversary $\cB$ that breaks the security of $f$ as follows:
\begin{enumerate}
    \item 
    Let $s$ be a polynomial.
    Take $1^{s(\secp)}$ and $f(x)$ with $x\gets\bit^{s(\secp)}$ as input.
    \item Let $n(\secp):=s(\secp)^c$ and sample $u\gets\bit^{n(\secp)-s(\secp)}$.
    \item Run $w\gets\cA(1^{n(\secp)},f(x)\|u)$.
    \item Output $w_{1,...,s(\secp)}$.
\end{enumerate}
Then, for any polynomial $s$,
there exist a PPT adversary $\cB$,
a polynomial $p$, and an infinite subset $\Lambda\subseteq\Sigma$ such that
\begin{align}
& \Pr[f(x)=f(y):x\gets\bit^{s(\secp)},y\gets\cB(1^{s(\secp)},f(x))]   \\
& = \Pr[f(x)=f(w_{1,...,s(\secp)}):x\gets\bit^{s(\secp)},u\gets\bit^{n(\secp)-s(\secp)},w\gets\cA(1^{n(\secp)},f(x)\|u)]   \\
& = \Pr[f'(x\|u)=f'(w_{1,...,s(\secp)}\|u):x\gets\bit^{s(\secp)},u\gets\bit^{n(\secp)-s(\secp)},w\gets\cA(1^{n(\secp)},f'(x\|u))]\\   
& = \Pr[f'(z)=f'(\xi):z\gets\bit^{n(\secp)},\xi\gets\cA(1^{n(\secp)},f'(z))]\\   
&>\frac{1}{p(\secp)}
\end{align}
holds for all $\secp\in\Lambda$.
This contradicts the assumption that $f$ is a OWF on $\Sigma$.
Therefore, $f'$ is a OWF on $\Sigma$.
\end{proof}
\section{Proof of \Cref{lem:OWF_OWPuzz}}
\label{sec:OWF_OWPuzz}

\begin{proof}[Proof of \cref{lem:OWF_OWPuzz}]
Let $\Sigma\subseteq\mathbb{N}$ be a subset.
    Let $f:\bit^*\to\bit^*$ be a classically-secure OWF on $\Sigma$.
    Then by the definition of OWFs on $\Sigma$, there exists an efficiently computable polynomial $n$ such that
    for any PPT algorithm $\cA$ and a polynomial $p$, there exists $\secp^*\in\mathbb{N}$ such that
\begin{align}
\Pr[f(x)=f(x'):x\gets\bit^{n(\secp)},x'\gets\cA(1^{n(\secp)},f(x))]    
\le\frac{1}{p(\secp)}
\end{align}
holds for all $\secp\ge\secp^*$ in $\Sigma$.
    From $f$, we construct a classically-secure OWPuzz $(\Samp,\Ver)$ on $\Sigma$ 
    with 1-correctness and $\negl$-security as follows:
    \begin{itemize}
        \item $\Samp(1^\secp)\to(\ans,\puzz)$:
        \begin{enumerate}
            \item Take $1^\secp$ as input.
            \item Sample $x\gets\bit^{n(\secp)}$ and set $\ans\coloneq x$.
            \item Compute $f(x)=y$ and set $\puzz\coloneq(1^{n(\secp)},y)$.
            \item Output $(\ans,\puzz)$.
        \end{enumerate}
        \item $\Ver(\ans',\puzz)\to\top/\bot$:
        \begin{enumerate}
            \item Parse $\ans'=x'$ and $\puzz=(1^{n(\secp)},y)$.
            \item If $f(x')=y$, output $\top$. Otherwise, output $\bot$.
        \end{enumerate}
    \end{itemize}
    First we show that $(\Samp,\Ver)$ satisfies 1-correctness.
    In fact, for all $\secp\in\mathbb{N}$,
    \begin{align}
        \Pr[\top\gets\Ver(\ans,\puzz):(\ans,\puzz)\gets\Samp(1^\secp)] 
        &= \Pr[f(x)=y : x\gets\bit^{n(\secp)},y:=f(x)] \\
        &= 1.
    \end{align}

    Next, we prove that $(\Samp,\Ver)$ satisfies $\negl$-security on $\Sigma$.
    From the one-wayness of $f$, for any PPT adversary $\cA$, polynomial $p$, 
    there exists $\secp^*\in\mathbb{N}$ such that
    \begin{align}
        &\Pr[\top\gets\Ver(\cA(1^\secp,\puzz),\puzz):(\ans,\puzz)\gets\Samp(1^\secp)] \\
        &= \Pr[f(x')=f(x):x\gets\bit^{n(\secp)},x'\gets\cA(1^{n(\secp)},f(x))] \\ 
        &\le \frac{1}{p(\secp)}.   
    \end{align}
    holds for for all $\secp\ge\secp^*$ in $\Sigma$.
   This means that $(\Samp,\Ver)$ satisfies $\negl$-security on $\Sigma$. 
\end{proof}

\ifnum\submission=0
\bibliographystyle{alpha} 
\else
\bibliographystyle{splncs04}
\fi
\bibliography{abbrev3,crypto,reference,text}

\end{document}